\documentclass[11pt, table]{article}
\usepackage{amsmath}
\usepackage{amsthm}
\usepackage{float}
\usepackage{upgreek}
\usepackage{verbatim}
\usepackage{hyperref}
\usepackage{mathrsfs}
\usepackage{bm}
\usepackage{dsfont}
\usepackage{amsfonts}
\usepackage{amssymb}
\usepackage{bbm}
\usepackage{algorithm}
\usepackage{multirow}
\makeatletter
\renewcommand{\ALG@name}{Numerical Method}
\makeatother

\usepackage{algpseudocode}
\usepackage{mathtools}
\usepackage[font=small,labelfont=bf]{caption}
\usepackage{subcaption}
\usepackage[acronym]{glossaries}
\usepackage[margin=0.75in]{geometry} 
\usepackage{soul}
\usepackage{seqsplit}
\usepackage{xstring}
\usepackage{cleveref}

\usepackage{empheq}
\usepackage{BOONDOX-calo}
\usepackage{setspace}

\algrenewcommand\algorithmicrequire{\textbf{Input:}}
\algrenewcommand\algorithmicensure{\textbf{Output:}}

\usepackage{tikz}
\usetikzlibrary{shapes.geometric, arrows, calc}
\tikzstyle{blue} = [rectangle, minimum width=1cm, minimum height=0.75cm, text centered, draw=black, fill=blue!20]
\tikzstyle{green} = [rectangle, minimum width=1cm, minimum height=0.75cm, text centered, draw=black, fill=green!20]
\tikzstyle{red} = [rectangle, minimum width=1cm, minimum height=0.75cm, text centered, draw=black, fill=red!30]
\tikzstyle{yellow} = [rectangle, minimum width=1cm, minimum height=0.75cm, text centered, draw=black, fill=yellow!30]
\tikzstyle{gray} = [rectangle, minimum width=1cm, minimum height=0.75cm, text centered, draw=black, fill=gray!10]

\tikzstyle{black_arrow} = [black, thick,->,>=stealth]
\tikzstyle{blue_arrow} = [blue, thick,->,>=stealth]
\tikzstyle{red_arrow} = [red, thick,->,>=stealth]

\numberwithin{equation}{section}

\newtheorem{assumption}{Assumption}[section]
\newtheorem{theorem}{Theorem}[section]

\newtheorem{prop}[theorem]{Proposition}
\newtheorem{defn}{Definition}[section]
\newtheorem{remark}{Remark}

\DeclareMathOperator*{\argmax}{arg\,max}

\DeclareMathOperator{\sgn}{sgn}

\DeclareMathOperator*{\supp}{supp}

\newcommand{\Var}{\text{Var}}
\newcommand{\bell}{\bm{\ell}}

\newcommand{\E}{\mathbb E}

\newcommand{\R}{\mathbb R}

\newcommand{\one}{\mathbf{1}}
\title{DEX Specs: A Mean Field Approach to DeFi Currency Exchanges}
\author{Erhan Bayraktar, Asaf Cohen, April Nellis}
\date{\today}

\begin{document}
\maketitle

\begin{abstract}
    We investigate the behavior of liquidity providers (LPs) by modeling a decentralized cryptocurrency exchange (DEX) based on Uniswap v3. LPs with heterogeneous characteristics choose optimal liquidity positions subject to uncertainty regarding the size of exogenous incoming transactions and the prices of assets in the wider market. They engage in a game among themselves, and the resulting liquidity distribution determines the exchange rate dynamics and potential arbitrage opportunities of the pool. We calibrate the distribution of LP characteristics based on Uniswap data and the equilibrium strategy resulting from this mean-field game produces pool exchange rate dynamics and liquidity evolution consistent with observed pool behavior. We subsequently introduce Maximal Extractable Value (MEV) bots who perform Just-In-Time (JIT) liquidity attacks, and develop a Stackelberg game between LPs and bots. This addition results in more accurate simulated pool exchange rate dynamics and stronger predictive power regarding the evolution of the pool liquidity distribution.\\

\end{abstract}

\tableofcontents
\section{Introduction}
In recent years, the use of cryptocurrencies has grown from a niche technology to a trendy and widely used element of the wider financial market. These currencies offer an alternative to traditional financial institutions for digital payments, currency exchange, and investment products, among other more creative uses like betting games. The underlying mechanism of popular cryptocurrencies like Bitcoin and Ether is decentralization - transactions are recorded on a distributed ledger maintained by many users in a peer-to-peer network, known as a blockchain. The benefit of this decentralized structure is that no single user can unilaterally alter or manipulate transactions once they have been added to the blockchain. This removes the need for users to rely on a trusted authority to record or verify transactions. Decentralized cryptocurrency exchanges, or DEXs, seek to similarly remove the need for a trusted authority in facilitating cryptocurrency exchanges. By using an automated market maker (AMM) (essentially, an algorithm built into the exchange), DEXs eliminate the need for a limit-order-book-style market for cryptocurrency trading. Instead, they allow individual cryptocurrency users to provide liquidity to the exchange and subsequently earn transaction fees. The liquidity is managed according to the AMM algorithm, removing any need for centralized oversight or management of the DEX. These decentralized liquidity pools offer unique and exciting opportunities for investment and profit, but also expose investors to new sources of risk which must be studied and mitigated.

Investors in these liquidity pools are called \textit{liquidity providers} (LPs), and they compete among themselves to earn fees from users and therefore profit from their contributions to the pool. The decentralized nature of DEXs reduces the barriers to entry for LPs, leading to a game between many LPs as they each try to maximize their individual profits. Because liquidity is contributed by ``regular people" instead of an entity with a large amount of capital, the interactions between liquidity providers affect the executing of exchange transactions as well as the profitability of each liquidity provider's liquidity position. It is therefore interesting and relevant to study competition between liquidity providers. 

While $N$-player games have the potential to model this competition, they become difficult (and sometimes intractable) to solve as $N$ becomes large. Instead, we utilize mean field game theory to efficiently study the interactions between many LPs contributing to a DEX.
Mean-field game theory (introduced concurrently in \cite{lasrylions} and \cite{HuangMalhameCaines}) is another field of study with growing popularity and an expanding array of uses. These games consider a continuum of infinitesimally small players who compete to optimize their personal utility based on the collective actions of all other players. Because players are small, individuals are not able to directly influence outcomes, yet the state of the system can experience large changes based on the actions of the continuum of players. These games are able to efficiently model many scenarios with large numbers of interacting players, and are therefore very suitable for studying DEXs.

To this end, we model LPs using different approaches before finally determining that a mean-field model of a DEX most closely mimics observed DEX dynamics.
In addition to creating a comprehensive descriptive model of a liquidity pool and the uncertainties therein, we extend our model to study how LPs potentially anticipate and interact with outside attackers seeking to profit from the pool. 
We take a novel approach by modeling a Stackelberg game between the mean field of LPs and the outside attacker. Through our simulations, calibration, and verification, we are also able to draw assumptions about the risk-aversion of LPs and make predictions about how liquidity pools evolve over time. 

\subsection{Technical Background}
\subsubsection*{Uniswap v3 Liquidity Pools} 
There are many DEXs, but one of the most established is Uniswap. Uniswap facilitates swaps between hundreds of token and has a 24 hour trading volume of over \$1 billion\footnote{As of April, 2024. See https://app.uniswap.org/explore.}. Coin pairs are traded in ``pools" with a fixed transaction fee. There are no restrictions on who is able to contribute liquidity to the pool, and in return for their contributions, LPs earn fees at a predetermined rate associated with the given pool. The newest version of Uniswap, called Uniswap v3, allows LPs to take advantage of a feature called \textit{concentrated liquidity}, wherein they are able to set a range of exchange rates for which their liquidity will be evenly distributed, be utilized, and earn fees. The concentrated liquidity results in a \textit{liquidity distribution} across the exchange rates, or ticks, of the pool. 

Uniswap uses a constant product market maker (CPMM) to execute trades, where the essential idea is that the product of the reserves for the two currencies in a pool should remain constant. The concentrated liquidity approach implemented in v3 makes the actual calculation a little more complicated, and is expanded upon in \Cref{sec:dex}. However, the general idea of a CPMM is that swapping token A for token B within the pool raises the price of token B with respect to token A due to its relative scarcity after the transaction (and vice versa). Therefore, the exchange rate within the pool is wholly controlled by the transactions that are executed within it, rather than by a central market maker matching buy and sell orders as in a limit-order-book-style market. This can result in a discrepancy between the exchange rate within a pool and the more widely available \textit{market exchange rate}. The market exchange rate can be thought of as the exchange rate for a theoretical ``infinite-liquidity" pool, or as the rate obtained in a large centralized exchange like Coinbase where rates are carefully managed.

\subsubsection*{Maximal Extractable Value (MEV) Bots}
The algorithmic (and therefore predictable) nature of the CPMM algorithm, combined with the ``waiting room" setup of a blockchain's mempool, presents an opportunity for blockchain validators to \textit{extract} profit from a block by incorporating their own transactions and ordering the block in an optimal manner (that is, optimal for the validator). The profit that can be obtained by such a maneuver is called the Miner Extractable Value, or Maximal Extractable Value (MEV).  

In recent years, services like Flashbots \cite{flashbots} have arisen to allow users to more easily and reliably profit from the MEV arising in decentralized exchanges. These services often involve the use of \textit{bots}, or automated programs that find and take advantage of MEV opportunities. For simplicity (and to differentiate between users who simply wish to exchange currency and users who wish to take advantage of MEV opportunities), we will assume all MEV transactions are done by bots and we will call these transactions \textit{bot attacks}. 

One of the most common bot attacks is a \textit{sandwich attack}, also known as \textit{front-running}. In this scenario, a bot sees that a large currency swap has been submitted to the blockchain mempool (which serves as a transaction ``waiting room") but has not yet been added to a block. \textit{Blockchain validators} construct blocks from transactions waiting in the mempool. A block is only added to the blockchain every 12 seconds, and in this time, blockchain validators choose how to arrange transactions from the mempool into a block to be submitted to the blockchain.
This presents an opportunity for the bot to place orders before and after a target transaction by manipulating the sequence of orders being added to the blockchain and thus executed. Such bots are able to perform this manipulation through collaboration with blockchain validators. There are two different ways in which a bot can perform a sandwich attack - either through liquidity adjustment or through swap transactions. 

In the case of a \textit{swap-based sandwich attack}, the bot quickly submits two of its own transactions, one in the same direction as the preexisting transaction and one in the opposite direction. The swap in the same direction executes before the large swap, while the swap in the opposite direction executes after the large swap. This results in a worse exchange rate for the non-bot swap, but a better exchange rate for the second bot transaction. In this way, the bot can essentially ``make a profit from nothing" at the expense of the naive user. However, users can attempt to mitigate the impact of sandwich attacks by setting a slippage allowance, which ensures that a transaction only executes if the execution exchange rate is within a certain range of the expected execution exchange rate, as shown in \cite{HeimbachSandwich}.

In the case of a \textit{liquidity-based sandwich attack}, bots see an opportunity to earn fees by serving as an LP while also taking advantage of a potentially favorable exchange rate. In this scenario, a bot contributes a large amount of liquidity to the current tick right before a large swap transaction. After the swap has been executed, the bot LP withdraws its liquidity plus the fees earned from the transaction. In this way, the bot earns profit from fees, as well as any profit from the change in its liquidity composition due to the swap. While there usually \textit{are} other LPs with liquidity in the same tick, if the bot LP's contribution is significant enough, the bot will accrue almost all the profit. Another side effect of this type of attack is decreased pool exchange rate volatility because of the sizeable liquidity deposit in a single exchange rate tick, meaning that LPs in other ticks miss out on earning fees. So, in this scenario, the swapper is not affected by the bot's actions (the swapper will receive an effective exchange rate very close to the initial pool exchange rate), but other LPs will lose out on potential profits. These attacks are sometimes referred to as ``Just-In-Time" (JIT) liquidity attacks, and the attacker is sometimes called a JIT LP. However, to differentiate between the non-bot LPs that we study and the JIT LPs, we will refer to the JIT LPs as \textit{Bots}. We study the effects and potential mitigation strategies for these attacks later in our paper. 

\subsection{Literature Review}
Decentralized exchanges have grown in popularity in recent years, and so has the amount of research done on various aspects of their performance. For a survey of recent advances in blockchain technology and decentralized finance, we refer the reader to \cite{biais_survey}. Various aspects of liquidity provision in a decentralized exchange have also been studied. The paper \cite{HeimbachBehavior} studies blockchain data to determine how liquidity providers contribute liquidity between pools and how they react to market changes, and the authors find that many liquidity providers display caution in their behavior. In \cite{HeimbachLiquidity}, the authors model exchange rates by a geometric Brownian motion in order to estimate how long various symmetric Uniswap v3 liquidity positions will remain ``in-the-money" versus how much the LP will collect in fees. They find that the total fees earned by an LP are proportional to time in the money and the size of its position, and that most small providers are best served by a simple strategy but should also expect small returns. In \cite{Baron_onlinelearning}, the authors consider similar symmetric liquidity positions in a non-stochastic setting with positions that reset at every period, while in \cite{Cartea2023optimalLP} the authors also investigate optimal Uniswap v3 liquidity strategies given a pool exchange rate modeled using Brownian motion. Similarly, the authors of \cite{FanStrategicLPv3} model exchange rate changes within a Uniswap v3 pool as a Markov process moving between exchange rate ticks, and accordingly calculates dynamic liquidity positions which optimizes fee earnings. They study the frequency at which LPs should ``reset" their liquidity to avoid paying excessive gas fees. 

In \cite{Angeris_2020}, the authors solve a game between arbitrageurs to investigate whether decentralized exchanges are a reliable indicator (or \textit{oracle}) for broader market cryptocurrency exchange rates. They develop a framework for describing changes in liquidity reserves which allows them to solve an arbitrageur optimization problem where arbitrageurs restore the pool exchange rate to the market exchange rate. They therefore propose that exchange rates in a decentralized exchange are generally consistent with market exchange rates, and so are useful oracles for determining ``fair" exchange rates in blockchain smart contracts (programs that execute transactions under certain conditions). The authors later apply these results specifically to Uniswap v2 exchanges in \cite{Angeris2021-Uniswap}, and show the stability and consistency of DEXs with regard to centralized exchanges. In \cite{Capponi2024SVB}, the authors study the effect of price shocks on liquidity provider holdings and find that they are vulnerable to losses through the actions of arbitrageurs. In addition, \cite{Park} finds that, despite the potential for manipulation in AMMs, swap execution prices are generally lower in DEXs than in centralized exchanges. The paper \cite{JaimungalAMM} also studies arbitrage in a different context. The authors use deep learning to determine an optimal swap execution strategy for a trader who wishes to liquidate a large holding of a specific token. As part of the paper's data analysis, the authors calculate the conditional probability of buy and sell swaps in connection to arbitrage opportunities in empirical data. They model pool and market exchange rates accordingly using stochastic differential equations and use these equations to formulate an optimal liquidation problem.

Bot activity is also a popular area of recent study. In the paper \cite{FerreiraSequencing}, the authors develop a block transaction sequencing rule that is verifiable in polynomial time and eliminates swap-based sandwich attacks if reliably implemented. Taking another approach, the authors of \cite{HeimbachSandwich} calculate the exact slippage allowance that users should set to discourage swap-based sandwich attacks in Uniswap v2 and show that in general, the optimized slippage levels perform significantly better than the Uniswap-suggested auto-slippage levels. Taking a regulatory perspective, the authors of \cite{HeimbachSelfRegulation} investigate whether the incentives are correctly aligned for LPs to implement anti-bot-attack strategies by considering a game between LPs, swappers, bots, and arbitrageurs where LPs and swappers can choose to operate in a pool with or without bots. Another paper dedicated to studying bots' behavior which focuses more on experimental results is \cite{ZhouQin21}. The authors perform their own experiments to evaluate the possibility of the success of sandwich attacks when submitting these transactions to the mempool instead of collaborating directly with validators. 

In the field of JIT liquidity study, the paper \cite{Xiongetal} investigates the impact of JIT liquidity attacks. The authors identify 36,671 JIT attacks over a horizon from June 2021 to January 2023, and find that liquidity-based attacks negatively impact LP profits. JIT attacks are also studied in \cite{capponi2024paradox}, where the authors model a game based on a single transaction between a set number of ``passive" liquidity providers and a JIT LP, set in a liquidity pool without concentrated liquidity.
The authors find that JIT LPs potentially disincentive passive LPs from providing liquidity in shallow pools. They propose Cournot-style competition between JIT LPs as a way to mitigate their effects. 

Other elements of the DeFi space have been studied in the context of games. For example, in \cite{SircarBTCMFG} the authors consider a mean-field game model of cryptocurrency mining and discuss the potential for wealth concentration in a Bitcoin-like blockchain with high mining rewards. However, this paper focuses on the blockchain mechanism rather than transactions in a DEX. Also, the author of \cite{Aoyagi} studies a mean-field game between LPs in Uniswap v2 where all LPs are identical and contribute liquidity for only one incoming swap transaction. These LPs react to the expected swap size of an incoming trader, much like in \cite{Capponi2021}. We expand the mean-field game concept in a variety of ways in our paper.

\subsection{Our Contribution}
We wish to expand the study of DEXs by building a comprehensive simulation of such exchanges. Previous studies of DEXs have made certain simplifications, such as modeling pool exchange rates using a stochastic process or studying competition between users when reacting to only a single transaction. In contrast, we model liquidity providers who intend to leave their liquidity in the pool for a period covering multiple swaps, and who optimize their liquidity position based on a long time horizon. By studying real-world DEX transaction data, we identified that most exchanges are places by \textit{semi-strategic swappers} whose behavior depends on both the arbitrage level of the pool as well as potentially unknown personal incentives. Including stochastic swaps based on observed data allows us to calculate pool exchange rate dynamics as a function of swap transactions as well as the collective actions of the LPs in the pool. 
As a bonus, we avoid assumptions regarding the form of the pool exchange rate process. 

In this setting, we build a detailed framework to describe competition between LPs. 
To our knowledge, we are the first to consider a mean field of LPs with heterogeneous types and long-term profit functions, and to consider MFGs in the more complicated concentrated liquidity setup of Uniswap v3 where LPs can choose their liquidity positions. 
The MFG setting provides an efficient approach to studying the interactions of many LPs in the pool by avoiding the need to consider each LP in the pool individually. This is especially important because we introduce a type distribution over the LPs in the pool, which allows us to investigate why different LPs choose different liquidity positions within the same liquidity pool. 
We further increase the accuracy of our model by calibrating the distribution of types so that the equilibrium distribution matches observed liquidity data. We then use the calibrated model of mean field competition and semi-strategic swappers to simulate pool activity, which results in accurate predictions about the evolution of pool liquidity and pool exchange rates. Because the calibrated type distribution has strong predictive power and is consistent with observed data, this calibration also allows us to make inferences about un-observable characteristics like risk-aversion.

We are also the first to consider a Stackelberg game between a mean field of LPs as the leader and a Bot performing JIT liquidity adjustments as the follower. We find that this formulation has increased predictive power when compared to the mean-field-only model and allows us to formalize a relationship between the actions of LPs and the actions of bots. We emphasize the following novel elements in our model: 
\begin{itemize} 
    \item Instead of separating the roles of swappers and arbitrageurs, we introduce a new class of \textit{semi-strategic swappers} who arrive at the pool at a rate that corresponds to the pool arbitrage opportunities. This is a more natural model for ``casual" pool users, who are not perfectly strategic yet are not wholly naive. 
    \item When choosing their optimal liquidity positions, LPs optimize their decisions over time horizons which encompass multiple swaps, instead of over a single swap as in previous works. We model fluctuations in the pool exchange rate through the arrival of these swappers, avoiding any assumptions on the mathematical structure of pool exchange rates.
    \item Liquidity optimization decisions depend on the actions of other LPs. This elevates the problem to a game between heterogeneous LPs. Previous papers have studied optimal liquidity strategies for individual Uniswap v3 LPs, but have not attempted to explain overall pool liquidity dynamics. Interactions between a mean-field of LPs with varying risk-aversion levels and other characteristics provide such explanatory power. 
    \item We extend the literature and analysis of JIT liquidity provision to encompass competition with strategic and interacting LPs in a concentrated-liquidity pool. We also investigate how profits are affected when LPs anticipate such JIT bot attacks. 
\end{itemize}
These components combine to create a more complex and comprehensive model than those previously studied. With this model, we are able to more closely match our results to real LP behavior and pool dynamics. In particular, we wish to emphasize two major contributions to the literature.
\begin{itemize}
    \item In \Cref{sec:mfgpredict}, we find that the liquidity behavior the LPs active in a DEX is best modeled via a mean-field game between heterogeneous agents, as this formulation allows us to deduce the high risk-aversion of LPs and make accurate 24-hour-ahead liquidity predictions consistent with observed data.
    \item In \Cref{sec:bots}, we find that a simulation which models a Stackelberg game between a mean-field of LPs and a Bot performing a JIT liquidity attack is able to decrease the Wasserstein 1-distance between the predicted and observed liquidity distributions by 21.54\%. In addition, when LPs choose their liquidity positions while anticipating Bot attacks, they are able to increase their collective profits by approximately 3\%. This means that strategic action from LPs can reduce loss of profit due to Bot attacks.
\end{itemize}

\subsection{Outline}
In \Cref{sec:dex}, we provide a detailed explanation of the exact CPMM used to calculate exchanges in a Uniswap v3 exchange pool. In \Cref{sec:games}, we discuss the formulations of various games between one or more LPs. In \Cref{sec:mfgpredict} we compare the behavior of the pool and market exchange rates in our generated liquidity distributions to data obtained from Uniswap, and demonstrate their consistency. In \Cref{sec:bots}, we introduce the presence of bots to the scenario, and create a Stackelberg game between a mean field of LPs and a bot entering the pool. 

\subsection{Data}
All transaction data is obtained from the Etherscan API and exchange rate data was obtained from the Coinmetrics API and the CoinGecko API. Our code and the necessary datasets can be found on \href{https://github.com/april-nellis/dex-specs}{Github}.

\section{Implementation of a CPMM with Concentrated Liquidity}
\label{sec:dex}
The description of the entire mechanism of a Uniswap v3 concentrated liquidity pool was first laid out in the Uniswap v3 whitepaper \cite{Uniswapv3}, but in this section, we will summarize the most relevant mathematical aspects of the exchange rate calculations. For those who are curious about the specifics, further details are provided in \Cref{sec:appendix_dex}.

\subsection{General Overview}
\label{sec:cpmm_structure}
Without loss of generality, in the remainder of the paper we will consider token A as the \textit{asset token}, which has a fluctuating exchange rate with respect to the \textit{base token}, referred to as token B. For example, token B could be a stablecoin like USD Coin or Tether, which are collateralized in such a way that one coin is roughly equivalent to one US dollar.

The range of supported exchange rates is fixed when the pool is created. This range is divided into $d$ intervals with corresponding exchange rate functions. These intervals, or \textit{ticks}, are labeled according to $d+1$ exchange rate points represented by $\bm{p} := \{p_i\}_{i=1}^{d+1} \in \R^{d+1}_+$. Each tick contains tokens of A, B, or some combination of the two. These tokens are used to facilitate swaps where the exchange rate with units B/A falls in the interval $[p_i, p_{i+1})$. When discussing ticks, we will refer to each tick by its index, and the reader can think of tick $i$ as a ``bin" containing tokens, with a ``label" given by $[p_i, p_{i+1})$. The token distribution across these $d$ ticks is given by two vectors $\bm{A}:= \{A_i\}_{i=1}^d \in \R^d_+$, representing the amount of token A in each tick, and $\bm{B}:= \{B_i\}_{i=1}^d \in \R^d_+$, representing the amount of token B in each tick. We sometimes refer to $\bm A$ and $\bm B$ as \textit{token reserves}. At any one time, only one tick is \textit{active} -- this is the tick which contains the current pool exchange rate, represented by $p^*$ (where $p^*$ is an exchange rate with units B/A), and the index of this tick is designated by $i^* \in [d]:=\{1, 2, \ldots, d\}$. The value of $i^*$ is therefore a function of $p^*$, and we define
\begin{equation*}
    i^* := i(p^*) := \max\{i\in[d]:p_i < p^*\}.
\end{equation*}

We must also make a clear distinction between the token reserves in the pool and the \textit{liquidity} of the pool. These terms are very closely related, and more tokens in the pool is naturally correlated with a higher amount of pool liquidity. In addition, both these concepts relate to the ability of the pool to facilitate currency swaps. However, following the convention described in \cite{Uniswapv3}, we will use the term \textit{liquidity} to describe the overall ``capacity" of the pool for handling swaps, and will represent it mathematically as the vector $\bell := \{\ell_i\}_{i=1}^d \in \R^d_+$. At times we will also refer to this vector as the pool's \textit{liquidity distribution}. This liquidity has its own units separate from units of token A or units of token B. 
To reference \cite{Uniswapv3} directly, the token reserves in any tick $i$ must satisfy the CPMM function
\begin{equation*}
    \Big(B_i + \ell_i\sqrt{p_i}\Big)\Big(A_i + \ell_i/\sqrt{p_{i+1}}\Big) = \ell_i^2.
\end{equation*}
Swaps will change the ratio of $A_i$ to $B_i$ in the tick, whereas liquidity addition and removal will change the values of $A_i$, $B_i$, and $\ell_i$ while maintaining the ratio of $A_i$ to $B_i$. 
The reserves in the active tick $i^* = i(p^*)$ must also satisfy the pool exchange rate condition given by
\begin{equation*}
    p^* = \dfrac{B_{i^*} + \ell_{i^*}\sqrt{p_{i^*}}}{A_{i^*} +
    \ell_{i^*}/\sqrt{p_{i^*+1}}},
\end{equation*}
and we note that swaps affect $p^*$ but liquidity additions and removals do not. For any tick $i$, $p_i$ and $p_{i+1}$ are fixed and known. Therefore, knowing two of the four quantities $\ell_i, A_i, B_i, p^*$ allows us to calculate the other two quantities for tick $i$. These equations are used both when liquidity providers add or remove liquidity, and when swappers place a swap transaction. 
We will define the \textit{swap function} as 
\begin{equation*}
    \psi: \R \times \R^d_+ \times \R_+ \to \R,
\end{equation*} 
where $x$ is an amount of token B and $\psi(x, \bell, p^*)$ is an amount of token A. If $x > 0$, then the swapper wishes to exchange token B for token A. If $x < 0$, then the swapper wishes to exchange token A for token B. The exact form of $\psi$ is not necessary for our analysis, but we note that it is differentiable, concave down, and increasing. For the curious reader, we discuss $\psi$ more precisely in \Cref{sec:appendix_swap} and present its form in \Cref{eq:psi_i}.

Below, we provide a brief overview of the liquidity provider's decision problem.
\subsection{Liquidity Providers}
In our model, LPs contribute to the liquidity pool by choosing a range of exchange rates in which to contribute and an amount of capital to contribute. The LP chooses a range of exchange rates $[p_{j^1}, p_{j^2})$ corresponding to ticks $\{j^1, \ldots, j^2-1\}$, where $j^1 < j^2$ and $j^1, j^2 \in [d+1]$. This is the range of exchange rates in which the LP will contribute liquidity and subsequently earn fees from swap transactions. The LP earns fees in proportion to its contribution to each tick's overall liquidity, which is why we prefer to discuss the pool in terms of its liquidity instead of its token reserves. The range itself is also important because LPs might not wish to contribute liquidity for trades in ticks where they think the exchange rate is ``too cheap" or ``too expensive", or they might wish to concentrate their liquidity in a smaller range so that they can earn a greater proportion of fees in those ticks. We denote the transaction fee rate as $\gamma$, and emphasize that this is a fixed part of the setup of the pool. 

\begin{remark}
\label{rem:intervals}
    There are now two types of intervals in our discussion - the exchange rate intervals $[p_i, p_{i+1})$ which discretize the range of exchange rates traded in the pool, and the interval chosen by each LP for contributing liquidity. In the rest of the paper, we will attempt to alleviate this confusion by using {\rm ticks} to describe exchange rate intervals, and {\rm positions} to describe players' liquidity intervals. 
\end{remark}

The capital that the LP adds to the pool is usually a mix of token A and token B, contributed in a ratio dependent on the chosen exchange rate range and the current pool exchange rate. We will track the value of the capital solely in terms on token B. We assume that LPs value their total capital using the {\it market exchange rate}, which we denote as $m^*$, because it is impervious to manipulation from within the pool and so can be thought of as the ``real" valuation of token A in terms of token B. 
Other approaches could be taken, for example converting holdings of both token A and token B to their value in US dollars, but for simplicity we use the convention of valuing all holdings in terms of token B. 

Furthermore, we will assume that each LP has a fixed amount of capital (as valued in terms of token B) that it is able to contribute to the pool. 
Given a total capital endowment $k$ and a chosen position $(j^1, j^2)$, the LP will add $u(j^1, j^2, k, m^*)$ units of liquidity to each tick between $j^1 \in [d]$ and $j^2 \in [d]$, where $u$ is defined as
\begin{equation}
\label{eq:deltaliq}
    u(j^1, j^2, k, m^*) := \dfrac{k}{\sqrt{p^*} - \sqrt{p_{j^1}} + m^* \left(1/\sqrt{p^*} - 1/\sqrt{p_{j^2}}\right)}.
\end{equation}
The liquidity position corresponding to a contribution of $u$ units of liquidity per tick and a chosen position $(j^1, j^2)$ is given by $\bell_{\text{new}}(u, j^1, j^2)$, defined as 
\begin{equation}
\label{eq:ell-new}
    \bell_{\text{new}}(u, j^1, j^2) := [0,\ldots, 0, \underbrace{u, \ldots, u,}_{\text{ticks }j^1\text{ to }j^2 - 1} 0, \ldots, 0] \in \R^d_+.
\end{equation}
The bulk of our paper will discuss how to calculate the optimal position for a liquidity provider in different scenarios. 

\subsection{Uniswap Details}
We chose to analyze data from the Uniswap USDC/ETH pool with a 0.05\% pool exchange fee, as that is the pool with the current largest 24 hour trading volume\footnote{As of April 13, 2024. See https://app.uniswap.org/explore/pools.}. This pool facilitates exchanges between Ether (the native coin on the Ethereum blockchain) and USD Coin (a coin which attempts to maintain a value of 1 USD by backing each coin 1-1 with US dollar reserves). We designate USDC as our base token, since its value is relatively stable and consistent with the US dollar, and designate ETH as our asset token. All exchange rates are expressed in USDC/ETH. We pulled blockchain data from Etherscan and aggregate data from CoinMetrics to analyze the token reserves and user behavior in this pool in the the month of transactions between September 1, 2023 and September 30, 2023.

\section{Liquidity Provision in Different Settings: Optimal Control Problem, N-Player, and Mean-Field Games}
\label{sec:games}
We have provided a brief description of the CPMM mechanism in a concentrated liquidity pool. Now, we wish to model how LPs can behave optimally within such a pool. 

In order to be consistent with the behavior observed in real LPs, we will assume in the rest of the paper that LPs wish to optimize their returns over a time horizon of $T > 0$ blocks which encompass multiple swap transactions. 
In our model, the LP(s) adjust liquidity at $t=0$ in order to maximize their profit over the next $T$ blocks. Once the liquidity distribution of the pool is set, swappers place transactions in the pool and the LPs earn fees based on their positions. We impose the simplification that no LPs adjust their liquidity for the next $T$ blocks. Therefore, the randomness in the model arises from the unknown future actions of the swappers, and it is only the swappers who are conducting transactions between block $t=0$ and block $t=T$. We say this to emphasize that the LPs face a \textit{static} optimization problem, rather than a dynamic one.

Numerical data indicates that on average LPs leave their liquidity unadjusted for approximately 27.157 hours. For simplicity and to err on the side of more ``alert" LPs, we calculate $T$ based on the number of blocks added to the blockchain in a finite time horizon of 24 hours. Based on the structure of the Ethereum blockchain, a block is ``finalized" once every 12 seconds, meaning that we define $\Delta t = 12$ seconds when needed. There are 7200 blocks added in a 24 hour period and so for the rest of the paper we will use $T = 7200$. Because we now have an element of time in our model, we will replace $p^*$ and $m^*$ with the discrete, adapted stochastic processes $(p^*_t)_{t=0}^T$ and $(m^*_t)_{t=0}^T$. 

Below, we will describe our method for modeling swap transactions, then subsequently develop and analyze three different optimization problems involving liquidity providers. First, we consider the optimization problem for a single LP arriving at the pool in \Cref{sec:monopoly}. Then, we introduce competition between LPs. We will examine an $N$-player game in \Cref{sec:n-player} and extend this setting to a mean-field game model for the pool in \Cref{sec:mfg}. 

\subsection{Semi-Strategic Swapper Behavior}
\label{sec:swap_calibration}
Recall that LPs contribute liquidity to the DEX in return for the opportunity to earn fees from swap transactions at a fixed rate $\gamma$. Therefore, the profits of these LPs are heavily dependent on the magnitude and direction of swaps made in the pool. In addition, the evolution of the pool exchange rate depends on the swaps made in the pool, because $p^*_t$ is determined by $p^*_0$ and the swaps made in blocks 0 through $t-1$. 

We will represent the full sequence of swaps that occur between blocks $t=0$ and $t=T$ by the random sequence $\bm \xi = (X_t, \xi_t)_{t=0}^T$, where $X_t \in \{0, 1\}$ and $\xi_t \in \R$. The swap arrival indicator $X_t$ is a Bernoulli random variable that denotes whether a swap occurs in block $t$. We assume for now that at most one swap can occur in each block. This could later be extended to allow multiple swaps in one block. If block $t$ contains a swap, the size of the swap is given by $\xi_t$. Recall that if $\xi_t > 0$, the swapper exchanges token B for token A, and if $\xi_t < 0$ then the swapper exchanges token A for token B.

Without an idea of how swappers behave, LPs are unable to make educated choices when determining their optimal position $(j^1, j^2)$. Naturally, the LP cannot see the future, but it is possible to predict the behavior of swappers based on historical swap transaction data, and we will assume throughout this paper that all LPs in the pool are strategic and are able to perform the necessary analysis. In the rest of this section, we give a brief overview of the method used to fit simulated swaps to observed data. Extra details and figures are provided in \Cref{sec:appendix_calib}.

By looking at historical Uniswap data over the 13 month period from November 1, 2022 to November 30, 2023, we found that swappers react to the current difference between the pool and market exchange rates, $p^*_t - m^*_t$, which we refer to as the pool's \textit{arbitrage level}. 
Intuitively, it makes sense that when one token is cheaper to obtain in the pool than in the market, swappers wish to take advantage of this arbitrage opportunity. However, there are still swappers who instead exchange tokens in the pool based on their own private preferences. Using this reasoning, we define \textit{semi-strategic swappers} who are aware of the arbitrage level in the pool, but also have other incentives for making swaps. Based on this logic, one can calibrate the likelihood of a swap arriving at a given block and, if it arrives, can deduce the distribution of the size of the swap. 

We begin by discussing the distribution of the random swap arrival indicator variable $X_t$. The arrival of a swap at time $t$ is determined by the value of $X_t$ where
\begin{equation*}
    X_t\sim Bernoulli(\rho(p^*_t - m^*_t)),\qquad \text{ where } \rho: \R \to [0,1].
\end{equation*} 
Based on empirical data,  $\rho$ is monotonically increasing. By analyzing swap arrival times, we approximated the probability of a swap arriving at block $t$ by 
\begin{equation*}
    \rho(x) = \dfrac{1}{1 + e^{-0.01145|x| + 0.6169}}.,
\end{equation*} 
such that when $|p^*_t - m^*_t| = 0$, the probability of a swap arriving is $0.3505$. Note that as the absolute value of the arbitrage level increases, the probability of a swap arriving increases to 1. The exact rationale for our choice of parameters in $\rho$ is given in \Cref{sec:appendix_arrivals}.

If a swap does occur in block $t$ (meaning $X_t = 1$), the size of the swap is given by the random variable $\xi_t$ which also depends on the level of arbitrage, $p^*_t-m^*_t$. In \Cref{fig:size1}, we show that both the transaction sizes and the log of transaction sizes are negatively correlated with $p^*_t - m^*_t$. This aligns with an intuitive understanding 
\begin{figure}
    \centering
    \begin{subfigure}{0.45\textwidth}
         \centering
    \includegraphics[width=\textwidth]{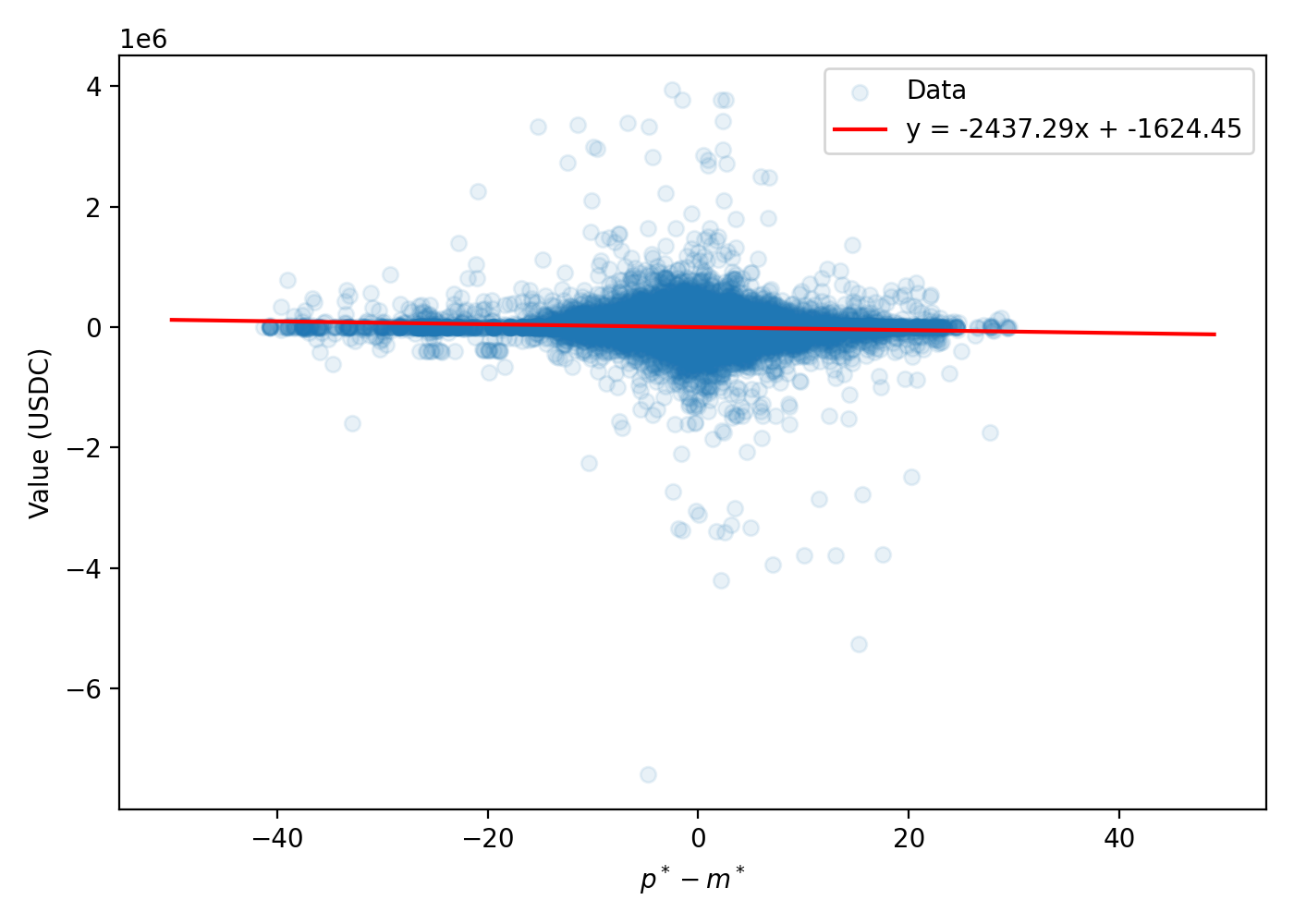}
     \end{subfigure}
     \begin{subfigure}{0.45\textwidth}
         \centering
         \includegraphics[width=\textwidth]{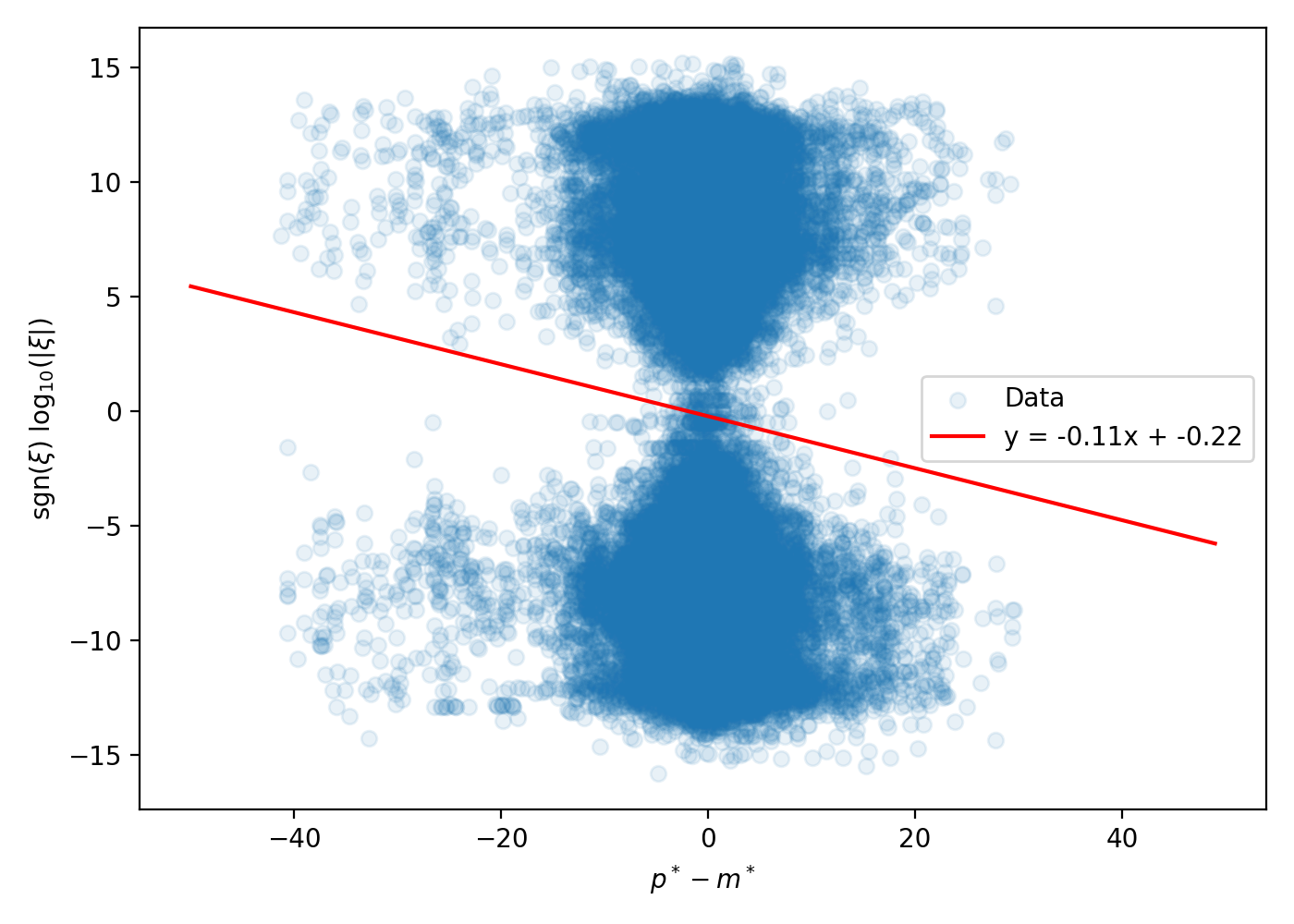}
     \end{subfigure}
    \caption{Plot of arbitrage level versus transaction size (left) and signed log transaction size (right).}
    \label{fig:size1}
\end{figure}
The transactions demonstrate a negative correlation between size and arbitrage level, and a two-tailed t-test for non-negative slope returned a p-value of 0, meaning that the slope is statistically different from zero. However, the correlation coefficient between the arbitrage level and the transaction size was only -0.06, meaning that much of the variation in the data cannot be explained by the variation in $p^*_t - m^*_t$. 

Therefore, we fit a Gaussian kernel density estimate to the log-transaction values (the log of the absolute value of the size of the USDC transaction) and the arbitrage levels, with the resulting estimated density shown in \Cref{fig:bimodal_density}. 
\begin{figure}
    \centering
    \includegraphics[width=0.8\textwidth]{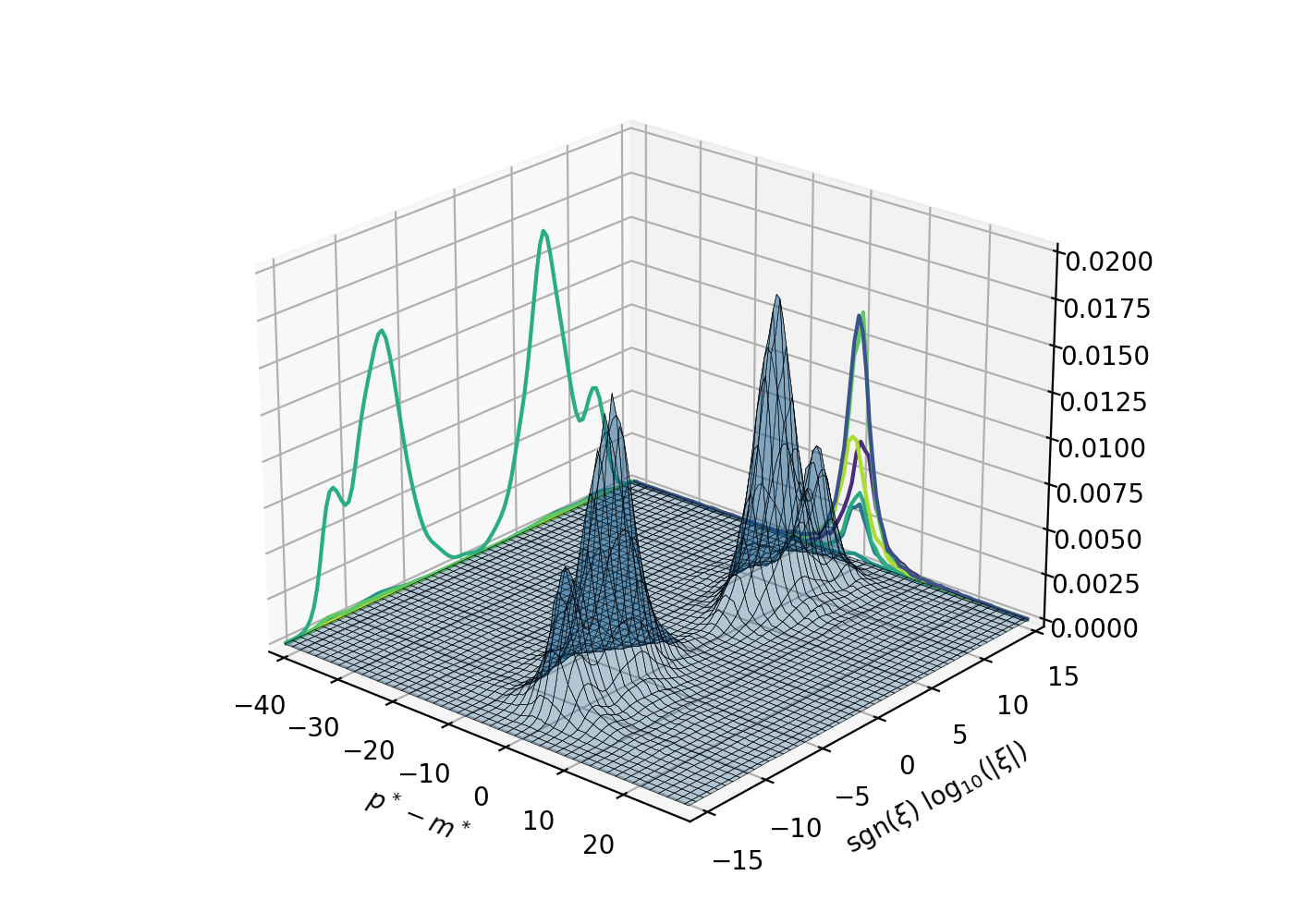}
     \caption{Joint density of the arbitrage level and of the log-value of transactions in USDC, $f_{\Xi, Y}(\xi, y)$.}
     \label{fig:bimodal_density}
\end{figure}
We describe the joint distribution of $\xi_t$ and $p^*_t - m^*_t$ through this estimated probability density function, which we denote $f_{\Xi,Y}(\xi, y)$, where $\xi$ is the size of a swap and $y$ is the arbitrage level before block $t$.
We again emphasize that the swapper knows $p^*_t-m^*_t$ before submitting transaction $\xi_t$ to the pool. Therefore, for a given arbitrage level $p^*_t - m^*_t$, we sample $\xi_t$ from the conditional distribution with conditional density given
\begin{equation*}
    f_{\Xi|Y}(\xi \mid Y = p^*_t - m^*_t).
\end{equation*}

Now that we have developed a model for the behavior of semi-strategic swappers arriving at the pool and described the distribution of the swap arrival and size $(X_t, \xi_t)$ in any block $t$, we can discuss optimal strategies for LPs reacting to the sequence of swaps $\bm \xi$ in various settings.

\subsection{Single Liquidity Provider Problem}
\label{sec:monopoly} 

We first simplify the liquidity optimization problem as a control problem for only one LP. This LP arrives at the pool and wishes to contribute liquidity in an optimal manner based on the current state of the pool. 

At time $t=0$, the LP knows $p^*_0$, $m^*_0$, and the preexisting liquidity distribution $\bell^0$. As discussed, it wishes to optimize its utility over a time horizon of $T = 7200$ blocks based on its choice of position at $t=0$. We will give this LP a few characteristics, namely that it has a {\it capital endowment} valued at $k>0$ units of token B, it has a {\it risk aversion level} $\lambda >0$, and it has a \textit{belief} about the trend of the market exchange rate, represented by $\delta \in \{-1, 0, 1\}$. 
These three characteristics will be summarized by the LP's \textit{type}, represented by $\theta := (k, \lambda, \delta)$.
\begin{remark}
    If we wish to specifically refer to the risk aversion, capital endowment, or beliefs of an LP with type $\theta$, we will refer to $\lambda(\theta)$, $k(\theta)$, or $\delta(\theta)$ respectively.
\end{remark}
In our paper, we will assume that $\lambda$ takes on values in a closed interval $[0, \lambda_{max}]$, and we will choose $\lambda_{max} = 5$ to cover a wide range of risk aversion levels. In addition, we will assume that $k \in \{2124, 35786, 1706034\}$. These values were chosen according to a cluster analysis which produced these three values as cluster centers. More details on this analysis can be found in \Cref{sec:appendix_calib_k}.This can be interpreted as LPs being loosely sorted into three groups: ``small" contributors, ``medium" contributors, and ``large" contributors. Finally, the belief parameter $\delta$ encodes the LP's belief about the trend of $m^*_t$. Because the market exchange rate is exogenous, it also introduces randomness into the LP's calculations. The LP knows that the exchange rate $m^*$ is equal to $m_0^*$ at block $t=0$, but the exchange rate will fluctuate between block 0 and block $T$. Cryptocurrency exchange rates are very sensitive to consumer sentiment and outside events, so it is difficult to make accurate predictions. Nonetheless, it is reasonable to assume that all participants in cryptocurrency markets have their own beliefs about the future performance of various tokens. The belief parameter $\delta$ incorporates the LP's preconceived estimate of whether the market exchange rate will increase ($\delta > 0$), decrease ($\delta < 0$), or remain relatively constant ($\delta = 0$). We assume that, based on $\delta$, the LP believes that the market exchange rate will evolve according to
\begin{equation}
\label{eq:trend}
    m^*_t = m^*_{t-1} e^{(\alpha + \delta/1000 - \sigma^2) \Delta t + \sigma \Delta W_t},
\end{equation}
where $\Delta t$ is the time between blocks (12 seconds), $\{\Delta W_t \sim N(0, \Delta t)\}_{t=1}^T$ are i.i.d. random variables, and $\alpha$ and $\sigma$ are estimated from historical data. Based on our data, we estimated that $\alpha = 0$ and $\sigma = 0.00106$.

We now have two different sources of uncertainty in our model, $\bm \xi$ and $(m^*_t)_{t=0}^T$. Note that the distribution of $\bm \xi$ at time $t$ is formulated to depend on $m^*_t$, and the LP's expectations about future values of $m^*_t$ depend on its belief $\delta$. The LP must calculate its expected utility when taking all of this into account. To shorten our notation, we will write $\E_{\bm \xi, \delta}[\ldots]$ and $\Var_{\bm \xi, \delta}(\ldots)$ to denote that expectation and variance are being taken over $\bm \xi$, given the LP's belief $\delta$.
The LP's control problem is to choose the ticks over which it wishes to distribute its capital endowment $k$. We will denote the space of controls as
\begin{equation*}
    J := \{(j^1, j^2)\ :\ j^1 \in [d], j^2 \in [d+1], j^1 < j^2\},
\end{equation*}
and we refer to an element of $J$ as a tuple $(j^1, j^2)$. Recall that if an LP contributes liquidity between ticks $j^1$ and $j^2$, the ticks $\{j^1, \ldots, j^2-1\}$ will all receive the same additional liquidity, though not necessarily the same token contributions. To save space, we will denote the new liquidity distribution arising from a choice of position $(j^1, j^2)$ and the LP's fixed endowed capital $k(\theta)$ as 
\begin{equation*}
    \bell^1(j^1, j^2) := \bell_{\text{new}}(u(j^1, j^2, k(\theta), m^*_0), j^1, j^2) + \bell^0,
\end{equation*}
where $\bell_{\text{new}}$ is defined in \Cref{eq:ell-new} and $u$ is defined in \Cref{eq:deltaliq}. For the value function in this section, and for those in the rest of the paper, we will not include $p^*_0$, $m^*_0$, and $T$ as explicit variables, but will rather use them as known, fixed constants. For a given choice of $j^1$ and $j^2$, the LP's utility is defined via a mean-variance objective function given by
\begin{equation}
\label{eq:V}
    V(j^1, j^2;\ \theta, \bell^0) := \E_{\bm \xi, \delta(\theta)} [\pi(j^1, j^2, \bm \xi;\ \theta, \bell^0)] - \lambda(\theta) \Var_{\bm \xi, \delta(\theta)}(\pi(j^1, j^2, \bm \xi;\ \theta, \bell^0,)),
\end{equation}
where $\pi$ represents the LP's profit associated with a sequence of swaps $\bm \xi$. Because $\bm \xi$ is random, $\pi$ is also random. It is defined as
\begin{equation}
     \pi(j^1, j^2, \bm \xi;\ \theta, \bell^0) := \sum_{t = 0}^T r(\xi_t, j^1, j^2, \bell^0, p^*_t, \theta), \label{eq:pi} 
\end{equation}
where $r$ represents the reward (fees) earned by the LP for a swap of size $\xi_t$ at block $t$ when the new liquidity distribution is given by $\bell^1(j^1, j^2)$. The pool exchange rate $p^*_t$ is successively calculated from $p^*_{t-1}$ and $\xi_{t-1}$ according to \Cref{eq:p_new}, where $p^*_0$ is known. Because we are currently considering only one strategic LP, we assume that the rest of the liquidity pool, represented by $\ell^0$, is fixed at block $t=0$ and remains unchanged for the next $T$ blocks. The LP earns fees when a transaction $\xi$ uses its liquidity. An LP's liquidity is used if at least some of an incoming swapper's tokens are added to some subset of its position. Each LP earns fees in proportion to its contribution to the liquidity in each tick. We define $r$ as
\begin{equation}
     r(\xi, j^1, j^2, \bell^0, p^*, \theta) :=  \sum_{i=j^1}^{j^2-1} \underbrace{\frac{u(j^1, j^2, k(\theta), m^*_0)}{\ell^1(j^1, j^2)_i}}_{\text{Share of Fees}}\underbrace{\phi(i, \xi; \bell^1(j^1, j^2), p^*)}_{\text{Fees in units of Token B}}, \label{eq:reward} 
\end{equation}
where $\phi(i, \xi; \bell^1(j^1, j^2), p^*)$ represents the portion of transaction fees from swap $\xi$ shared among LPs with liquidity in tick $i$, when the pool's liquidity distribution is given by $\bell^1(j^1, j^2)$ and the current pool exchange rate is $p^*$. The formulation of $\phi$ is discussed in \Cref{sec:swap_explanation} and the equation is given in \Cref{eq:fees}.
We remind the reader that $u(j^1, j^2, k(\theta), m^*_0)$ is the LP's liquidity contribution to each tick in its chosen position $(j^1, j^2)$ with capital endowment $k(\theta)$ based on a market exchange rate $m^*_0$.
The LP now wishes to solve
\begin{equation*}
    \max_{(j^1,j^2)\in J} V(j^1, j^2; \theta, \bell^0),
\end{equation*}
which is a discrete optimization problem with a (relatively) small control space. While the stochasticity arising from $\bm \xi$ still makes this difficult to solve analytically, it is straightforward to utilize Monte Carlo simulations to approximate the optimal choice of $(j^1, j^2)$. For the specific method used to calculate each LP's optimal liquidity position, the reader can refer to \Cref{sec:method_single}.

While this optimization problem is tractable and is beneficial for demonstrating the incentives which motivate each LP, it relies on the assumption that LPs do not compete with each other and therefore do not react to the behaviors of other LPs. However, we theorize that this is not the case and that in a more accurate model LPs should engage in a game among themselves. Therefore, we will extend this individual LP optimization problem to competitive games between LPs.

\subsection{Game Among N Liquidity Providers}
\label{sec:n-player}
We now consider an $N$-player game among multiple LPs who wish to maximize their profits earned from DEX fees. In our data analysis, we have found that Uniswap distributions are relatively stable over time but now we wish to study whether they can be modeled as an equilibrium of a competitive game. Because LP profits are determined by share of liquidity in each tick, the LPs compete against each other to gain a bigger share of fees in profitable ticks. In this section, we will develop a game and numerically approximate a Nash equilibrium.

Each LP is a player in the game and has characteristics similar to the single LP in \Cref{sec:monopoly}. LP $n$, for $n\in\{1,2,\ldots,N\}$, has an assigned type $\theta_n := (k(\theta_n), \lambda(\theta_n), \delta(\theta_n))$. 
It has a fixed capital endowment $k(\theta_n)$ to distribute across its chosen position $(j^1_n,j^2_n)$, its value function takes on a mean-variance form with risk aversion level $\lambda(\theta_n)$, and it has a belief $\delta(\theta_n)$ regarding the movement of the market based on \Cref{eq:trend}.   
In the $N$-player game, this belief parameter is reminiscent of asymmetric information between players, but we emphasize that none of the LPs have true insider information about market exchange rate fluctuations. We will also hold the time horizon of the problem fixed for all players at $T = 2700$.

Types are distributed across the $N$ players according to some joint probability distribution $\Theta_N$. For symmetry, we will assume that each player's type, $\theta_n$, is independently sampled from a given distribution $\Theta$. We can define the $N$-dimensional vector of player types as $\bm \theta := (\theta_1, \ldots, \theta_N) \sim \Theta_N$ where $\Theta_N := \prod_{n=1}^N \Theta$. We make the following assumption on the distribution $\Theta$. 
\begin{assumption}
\label{asm:theta}
The space of characteristics/types, $\supp(\Theta)$, is a separable and complete metric space. We also impose that $\Theta$ is atomless.
\end{assumption}
This game can be viewed as a Bayesian game. Players know their own types and are able to observe the liquidity distribution generated by the actions and types of the other players, but do not have specific information about the types of any other players. 
So, instead of one LP approaching a preexisting market, many LPs are now interacting to form the market. 
Players choose their actions based on their assigned type $\theta_n$, and so we define a type-dependent \textit{strategy profile} as
\begin{equation*}
    \bm j(\bm \theta) := \{(j^1_n(\theta_n), j^2_n(\theta_n))\}_{n=1}^N,
\end{equation*}
where player $n$ plays action $(j^1_n(\theta_n), j^2_n(\theta_n))$ based on its assigned type $\theta_n$. We will describe the type-dependent strategy profile of all players except player $n$ as 
\begin{equation*}
    \bm j(\bm \theta^{-n}) := \{(j^1_1(\theta_1), j^2_1(\theta_1)),\ldots, (j^1_{n-1}(\theta_{n-1}), j^2_{n-1}(\theta_{n-1})), (j^1_{n+1}(\theta_{n+1}), j^2_{n+1}(\theta_{n+1})), \ldots, (j^1_N(\theta_N), j^2_N(\theta_N))\},
\end{equation*} 
where
\begin{equation*}
    \bm \theta^{-1} := (\theta_1, \ldots, \theta_{n-1}, \theta_{n+1}, \ldots, \theta_N)
\end{equation*}
represents the $(N-1)$-dimensional type vector containing the types of all players except player $n$.
For a given type assignment vector $\bm \theta$, the liquidity distribution arising from all players except player $n$ is given by $\bell(\bm j(\bm \theta^{-n}))$ as 
\begin{equation*}
    \ell(\bm j(\bm \theta^{-n}))_i := \sum_{n' \neq n} u(j^1_{n'}(\theta_{n'}), j^2_{n'}(\theta_{n'}), k(\theta_{n'}), m^*_0) \one_{\{j^1_{n'} \leq i < j^2_{n'}\}}, \qquad \forall i \in [d].
\end{equation*}
Given a type vector $\bm \theta$ and a strategy profile $\bm j(\bm \theta)$, player $n$'s value function is thus defined as
\begin{equation*}
 V_{n}(j^1, j^2, \theta_n, \bm j(\bm \theta^{-n})) := \ \E_{\bm \xi, \delta(\theta_n)}\left[ \pi(j^1, j^2, \bm \xi;\ \theta_n, \bell(\bm j(\bm \theta^{-n})))\right] - \lambda(\theta_n) \Var_{\bm \xi, \delta(\theta_n)}\left(\pi(j^1, j^2, \bm \xi;\ \theta_n, \bell(\bm j(\bm \theta^{-n})))\right),
\end{equation*}
where $\pi$ is defined as in \Cref{eq:pi} but with $\bell(\bm j(\bm \theta^{-n})))$ replacing $\bell^0$. This is because player $n$ reacts to the actions of all other players instead of the preexisting liquidity distribution. It is assumed that all LPs share the same belief about the distribution of swap sizes as a function of the arbitrage level as described in \Cref{sec:swap_calibration}, but have different beliefs about the evolution of the market exchange rates $(m^*_t)_{t=0}^T$, which are determined by $\delta(\theta_n)$ through \Cref{eq:trend}. 
Since $\delta(\theta_n)$ plays a role in LP $n$'s predictions regarding $\bm \xi$, we will write $\E_{\bm \xi, \delta(\theta_n)}[\ldots]$ and $\Var_{\bm \xi, \delta(\theta_n)}(\ldots)$ to represent that LP $n$ is optimizing based on its belief $\delta(\theta_n)$.

\begin{defn}[Pure Strategy Bayes-Nash Equilibrium]
    A strategy profile $\bm j(\bm \theta):= \{(j^1_n(\theta_n), j^2_n(\theta_n))\}_{n=1}^N$ is called a pure Bayes-Nash equilibrium if, for all $n \in \{1, \ldots, N\}$, 
    \begin{equation*}
        (j^1_n(\theta_n), j^2_n(\theta_n)) \in \argmax_{(j^1, j^2) \in J} \E\left[ V_n(j^1, j^2, \theta_n, \bm j(\bm \theta^{-n}))\right].
    \end{equation*}
    This means that if player $n$ holds the strategies of all other players, given by $\bm j(\bm \theta^{-n})$, constant, then player $n$'s type-dependent strategy $(j^1_n(\theta_n), j^2_n(\theta_n))$ maximizes its expected payoff when the expectation is taken over $\bm \theta^{-n}$.
\end{defn}
\begin{prop}
    Under \Cref{asm:theta}, there exists a Nash equilibrium in pure strategies for the Bayesian game described above.
\end{prop}
\begin{proof}
    In Theorem 1 of \cite{radnerrosenthal82}, it is shown that in a Bayesian $N$-player game with a finite action space, if the private observations of players in an $N$-player game are mutually independent and the observations are drawn from atomless distributions, then a pure-strategy Bayes-Nash equilibrium exists. In our game, the action space $J$ is finite and in fact $|J| = d(d+1)/2$. In addition, in \Cref{asm:theta} we have imposed that $\Theta$ is atomless. Because each player's type $\theta_n$ is drawn independently from $\Theta$, our game is consistent with that of \cite{radnerrosenthal82} and a pure strategy Nash equilibrium exists.
\end{proof}
\begin{remark}
If all players' types $\bm \theta = (\theta_1, \ldots, \theta_N)$ are common knowledge at the start of the game, this is a finite, complete information game and by \cite{Nash1950}, a mixed strategy Nash equilibrium exists. 
We could also consider a Bayesian game in which the number of type configurations in $\supp(\Theta_N)$ is finite and thus $\Theta_N$ is a discrete joint distribution over all type configurations; in which case, player types are not drawn from atomless distributions. We also relax the independence between types. Nevertheless, an equilibrium can still be found. Theorem 2 of \cite{harsanyi2} states that, because $J$ is finite, this game must have at least one mixed strategy Bayesian-Nash equilibrium. 
\end{remark}

\subsubsection*{Numerical Results for N-Player Game}
We now want to use a numerical approach to investigate whether an $N$-player game like the one described above can model the dynamics that we have observed in Uniswap pools.
We again assume that all players wish to maximize their profits over the next $T=7200$ blocks. We want the total liquidity in the pool to remain consistent with the real Uniswap v3 ETH/USDC liquidity pool, which has a total volume of approximately 150 million USDC. In Uniswap pools, the number of LPs contributing to a pool is on the order of hundreds or thousands. However, we wish to see if we can analyze the pool using a simplified and therefore more tractable model. Therefore, we initially choose $N=10$ and accordingly we will adjust the capital endowments such that $k(\theta_n) \in \{212400, 3578600, 170603400\}$. As in \Cref{sec:monopoly}, the beliefs of players are dictated by $\delta(\theta_n) \in \{-1, 0, 1\}$. Finally, we define $\lambda(\theta_n) \in [0, \lambda_{max}]$ where we set $\lambda_{max} = 3$. Therefore, $\supp(\Theta) := \{212400, 3578600, 170603400\}  \times [0, 3] \times \{-1, 0, 1\}$. However, we also need to specify the distribution $\Theta$ itself, where we will assume that it is atomless to maintain consistency with \Cref{asm:theta} and \cite{radnerrosenthal82}. 
\begin{figure}
    \centering
    \begin{subfigure}{0.45\textwidth}
        \centering
        \includegraphics[width=\textwidth]{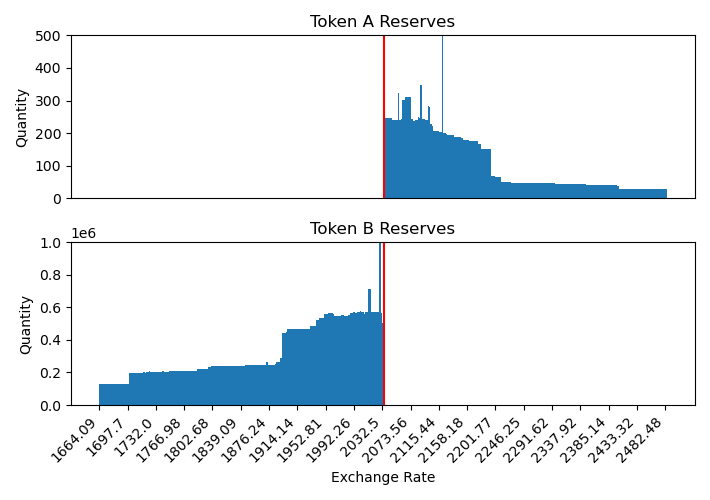}
         \caption{The observed Uniswap ETH/USDC token reserves on November 29, 2023. }
    \end{subfigure}
    \hspace{1.5em}
    \begin{subfigure}{0.45\textwidth}
        \centering
        \includegraphics[width=\textwidth]{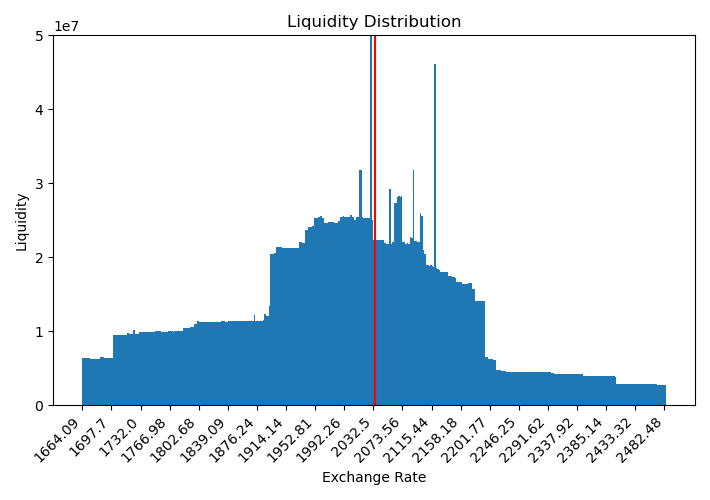}
         \caption{The observed Uniswap ETH/USDC liquidity distribution on November 29, 2023. }
    \end{subfigure}
\caption{Uniswap ETH/USDC pool on November 29, 2023.}
\label{fig:real-np}
\end{figure} 

Recall that in this $N$-player Bayesian game, each player reacts to an $(N-1)$-dimensional random vector of the other players' actions, which are a function of their randomly assigned types. Each type is randomly sampled from $\Theta$. We calibrate $\Theta$ such that an equilibrium of the $N$-player game will resemble the observed Uniswap ETH/USDC token reserves and liquidity distribution on November 29, 2023, shown in \Cref{fig:real-np}. We provide a more detailed description of our numerical method in \Cref{sec:method_np}. We will call the target liquidity distribution $\bell^0$. Once we have calibrated $\Theta$, the resultant expected token distribution for a 10-player game are shown in \Cref{fig:theta_np} and \Cref{fig:sim-np}. The calculations took 111.57 minutes\footnote{All computations were performed on a 2021 Apple MacBook Pro with M1 Pro processor chip.} and the Wasserstein-1 distance between our calibrated distribution and the target distribution is given by 18.198, while the R-score was 0.9179. While our calibrated distribution does bear a resemblance to $\bell^0$, note that this is only the \textit{expected} liquidity distribution arising from the calculated optimal strategy when types are distributed according to our estimate of $\Theta$. This is not the same as the \textit{realized} liquidity distribution arising from the actions of players whose types correspond to a specific $\bm \theta \sim \Theta_N$.

\begin{figure}
    \centering
    \includegraphics[width = 0.6\textwidth]{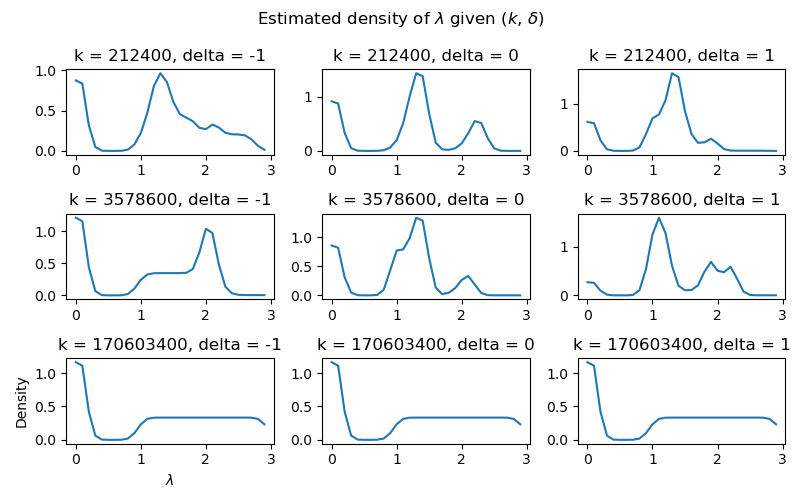}
    \caption{The estimated conditional distributions of $\lambda(\theta_n)$ based on $(k(\theta_n), \delta(\theta_n))$ for a 10-player game. }
    \label{fig:theta_np}
\end{figure}

\begin{figure}
    \centering
    \begin{subfigure}{0.45\textwidth}
        \centering
        \includegraphics[width=\textwidth]{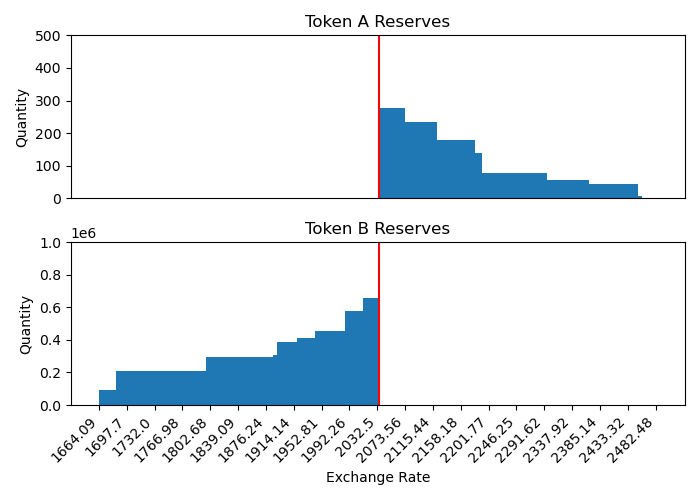}
         \caption{Token reserves based on calibrated $\Theta$. }
    \end{subfigure}
    \hspace{1.5em}
    \begin{subfigure}{0.45\textwidth}
        \centering
        \includegraphics[width=\textwidth]{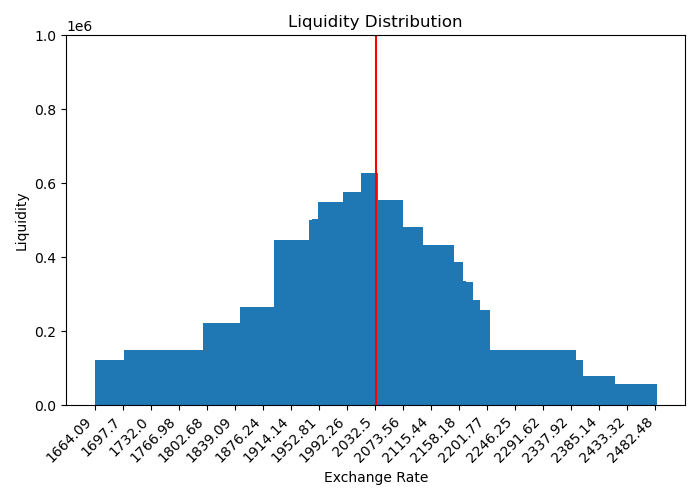}
         \caption{Liquidity distribution based on calibrated $\Theta$. }
    \end{subfigure}
    \caption{Expected pool appearance based on calibrated $\Theta$.}
    \label{fig:sim-np}
\end{figure}

After determining $\Theta$, it is natural to wonder what will result from sampling from $\Theta_N := \prod_{n=1}^N \Theta$. 
In \Cref{fig:l-nplayer}, we compare the equilibrium liquidity distributions arising from different assignments of $\{\theta_n\}_{n=1}^N$, where the parameters are shown in \Cref{tab:nplayer} and drawn from $\Theta$. In these two games, types are assigned and fixed before the game begins, but each player knows only its own type and $\Theta_N$, and reacts to the (type-dependent) actions of all the other players.
In Numerical Method \ref{alg:FP} of \Cref{sec:method_np}, we discuss our method for calculating these equilibria. 
It can be seen that Sample (a) and Sample (b) produced very different equilibrium liquidity distributions, even though both samples were drawn from $\Theta_N = \prod_{n=1}^N \Theta$ and the type assignments are fairly similar. This is because $\supp(\Theta)$ is atomless, and even the discretization of $\supp(\Theta)$, given by $\Delta$, contains 270 types when we discretize $[0, \lambda_{max}]$ as $\{\lambda_1, \ldots, \lambda_{30}\}$. Therefore, it is clear that we cannot obtain a truly representative sample of player types with only 10 players.

The most obvious solution is to increase the number of players. However, it is known that $N$-player games become much more difficult to analyze as $N$ becomes very large. 
This is because all analysis must be conducted with respect to an $N$-dimensional type vector.
To combat this difficulty, in the following section we will use mean-field games to more efficiently describe competition between LPs in a DEX. 

\begin{figure}
    \centering
     \begin{subfigure}{0.45\textwidth}
         \centering
         \includegraphics[width=0.9\textwidth]{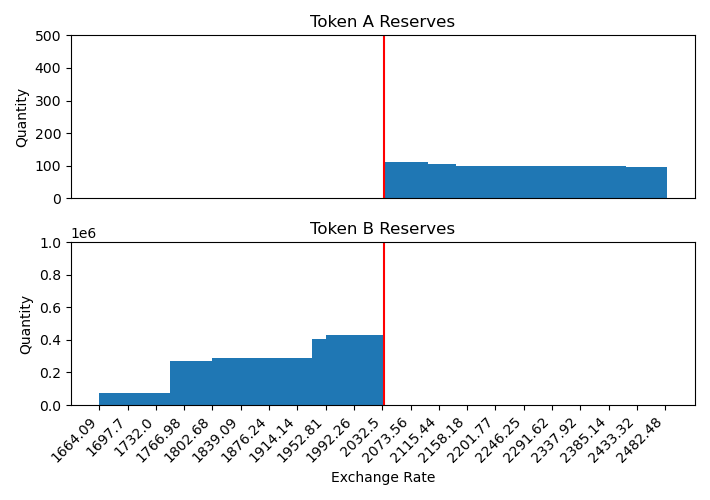}
         \caption{The simulated token reserves associated with $\bm j^*$ after a game between players with types from Sample (a) of \Cref{tab:nplayer}.\label{fig:np-b}}
     \end{subfigure}
     \hspace{1.5em}
     \begin{subfigure}{0.45\textwidth}
         \centering
         \includegraphics[width=0.9\textwidth]{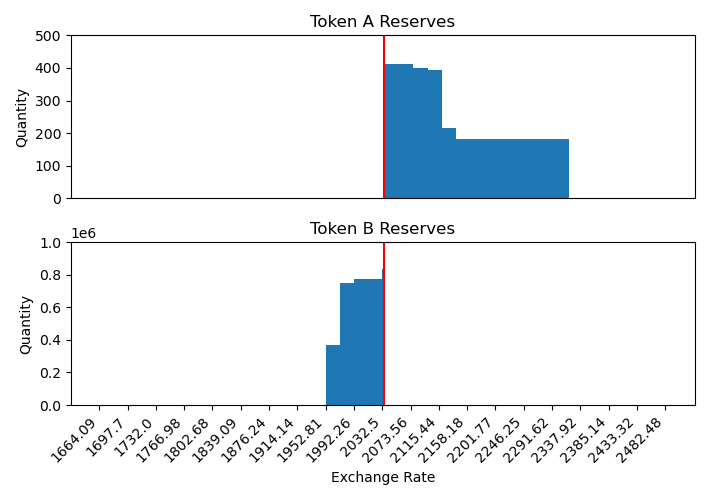}
         \caption{The simulated token reserves associated with $\bm j^*$ after a game between players with types from Sample (b) of \Cref{tab:nplayer}.\label{fig:np-d}}
     \end{subfigure}
     \caption{Numerically Nash equilibria of 10-player games with different realizations of $\{\theta_n\}_{n=1}^N$.}
     \label{fig:l-nplayer}
\end{figure}

\begin{table}
    \centering
    \scriptsize
    \begin{tabular}{|c|c|c|c|c|c|c|c|c|c|c|c|}
        \hline
        & Player & 1 & 2 & 3 & 4 & 5 & 6 & 7 & 8 & 9 & 10\\
        \hline 
        \multirow{3}{*}{(a)} & $k_n$ & 212,400 & 212,400 & 212,400 & 3,578,600 & 3,578,600 & 3,578,600 & 170,603,400 & 170,603,400 & 170,603,400 & 212,500 \\
        \cline{2-12}
        & $\lambda_n$ & 0.1 & 0.1 & 0.1 & 1 & 0.02 & 0.02 & 0.05 & 0.02 & 0.05 & 0.1\\
        \cline{2-12}
        & $\delta_n$ & 0 & 1 & -1 & 0 & 1 & -1 & 0 & 1 & -1 & 0\\
        \hline 
        \hline 
        \multirow{3}{*}{(b)} & $k_n$ & 212,400 & 212,400 & 212,400 & 3,578,600 & 3,578,600 & 3,578,600 & 170,603,400 & 170,603,400 & 170,603,400& 212,500 \\
        \cline{2-12}
        & $\lambda_n$ & 0.1 & 0.1 & 0.1 & 0.01 & 0.02 & 0.02 & 0.05 & 0.02 & 0.05 & 0.05\\
        \cline{2-12}
        & $\delta_n$ & 0 & 1 & 1 & 0 & 1 & 1 & 0 & 1 & 1 & 0\\
        \hline
    \end{tabular}
    \caption{Samples of types $\{\theta_n\}_{n=1}^N = \{(k_n, \lambda_n, \delta_n)\}_{n=1}^N$ for a 10-player game with type distribution $\Theta$ as in \Cref{fig:theta_np}. }
    \label{tab:nplayer}
\end{table}

\subsection{The Liquidity Provider Mean-Field Game }
\label{sec:mfg}
We now reformulate the N-player game as a mean-field game. We generalize the $N$-player game to the mean-field setting by defining a representative player whose characteristics $\theta$ are distributed according to $\Theta$ with properties as in \Cref{asm:theta}. As a reminder, the type $\theta$ encompasses an LP's capital endowment ($k$), risk aversion ($\lambda$), and beliefs about swap sizes ($\delta$), where the effect of $\delta$ on the representative player's beliefs about $(m^*_t)$ is the same as in \Cref{eq:trend}. Recall that $\theta$ does not change over time. The efficiency of a mean-field approach is two-fold. First, players only interact with the mean-field of other players' actions, instead of individual other players. In addition, we can consider a \textit{representative player} instead of looking at each individual player one by one. This allows us to determine the representative player's optimal response to $\Theta$, instead of an $N$-dimensional random type vector $\bm \theta$. 
The mean field game can be broken into two stages:
\begin{enumerate}
    \item The representative player chooses if/how to optimally adjust its positions based on a given liquidity distribution $\bell^0$ and its type $\theta \sim \Theta$. 
    \item Integrating with respect to the type $\theta\sim\Theta$, the action of this representative player creates a new liquidity distribution $\bell^1$. In equilibrium, $\bell^1 = \bell^0$.
\end{enumerate}
The representative player's payoff depends on the overall pool liquidity distribution $\bell^0$ as well as its own position. Because $\theta$ is distributed according to $\Theta$ which satisfies \Cref{asm:theta}, the representative player's strategy (and possibly payoff) is random with respect to $\theta$ and thus the player chooses a measurable function $\bm j: \supp(\Theta) \to J$ which assigns a position $(j^1, j^2)$ to each possible type $\theta$. We will sometimes use the notation $\bm j(\theta) = (j^1(\theta), j^2(\theta))$.  

\begin{remark}
   From an economic perspective, this static MFG can also be considered a game between a continuum of players. However, we choose to use MFG terminology throughout our paper because in the future we would also like to consider dynamic games in the same setting. There is a rich literature covering both theoretical and numerical results for dynamic MFGs, and using MFG terminology will allow us to more easily extend our setting to stochastic differential MFGs.
\end{remark}

For fixed values of $\theta$ and $\bell^0$, define the representative player's payoff as
\begin{equation}
V_{\text{MFG}}((j^1, j^2), \theta, \bell^0) := \E_{\bm \xi, \delta(\theta)}\left[\pi_{\text{MFG}}(j^1, j^2, \bm \xi;\ \theta, \bell^0)\right] - \lambda(\theta) \Var_{\bm \xi, \delta(\theta)}(\pi_{\text{MFG}}(j^1, j^2, \bm \xi;\ \theta, \bell^0)), 
\end{equation}
where $\pi$ is the total profit corresponding to a series of swaps $\bm \xi = (\xi_i)_{t=1}^T$, defined by
\begin{equation*}
    \pi_{\text{MFG}}(j^1, j^2, \bm \xi;\ \theta, \bell^0) := \sum_{t=0}^T r_{\text{MFG}}(\xi_t, j^1, j^2, k(\theta), p^*_t, \bell^0),
\end{equation*}
 with $k(\theta)$ being the capital endowment of a player with type $\theta$ and $p^*_t$ is defined from $p^*_{t-1}$ and $\xi_{t-1}$ according to \Cref{eq:p_new}. The function $r_{\text{MFG}}$ is the reward for a single swap $\xi$, given by 
 \begin{equation*}
     r_{\text{MFG}}(\xi_t, j^1, j^2, k(\theta), p^*_t, \bell^0) := \sum_{i = j^1}^{j^2-1} \frac{u(j^1, j^2, k(\theta), m^*_0)}{\bell^0_i}\phi(i, \xi; \bell^0, p^*).
 \end{equation*}
The liquidity distribution of the mean field resulting from a generic strategy $\bm j$ (not necessarily a best response) is given by $\bell_{\text{MFG}}(\bm j) \in \R^d_+$ where the liquidity in tick $i$ is defined as
\begin{equation}\label{eq:ell_MFG}
    \bell_{\text{MFG}}(\bm j)[i] := \int u\left(j^1(\theta), j^2(\theta), k(\theta), m^*_0\right)\one_{\{j^1(\theta) \leq i < j^2(\theta)\}} \Theta (d\theta),\quad \forall i \in [d].
\end{equation}

We now wish to discuss the equilibrium conditions for our model. Because the representative player reacts to the current liquidity distribution $\bell^0$ by choosing a type-dependent strategy $\bm j(\theta)$, we define the equilibrium of this game as a pair of a liquidity distribution and a control function.
\begin{defn}[Pure Strategy MFE]
    A mean-field equilibrium (MFE) is defined as a liquidity distribution and choice of controls $(\bell^*, \bm j^*)$ such that 
    \begin{enumerate}
        \item The controls $\bm j^*$ are a best response to $\bell^*$, meaning that for all $\theta \in \supp(\Theta)$, $\bm j^*$ satisfies
        \begin{equation*}
            \bm j^*(\theta) \in \argmax_{(j^1, j^2) \in J} V_{\text{MFG}}((j^1, j^2), \theta, \bell^*).
        \end{equation*}
        \item The liquidity distribution $\bell^*$ results from the optimal controls, meaning that 
        \begin{equation*}
            \bell^* = \bell_{\text{MFG}}(\bm j^*).
        \end{equation*}
    \end{enumerate}
\end{defn}

\begin{prop}
    Under \Cref{asm:theta} for $\Theta$, the MFG described in this section possesses at least one pure-strategy MFE.
\end{prop}

\begin{proof}
For a game with $\theta$-dependent utility functions that are continuous in action and distribution, Theorem 2 of \cite{mascolell84} states that a symmetric Cournot-Nash equilibrium for a fixed distribution of types $\Theta$ exists whenever the action space $J$ is finite and the distribution of types $\Theta$ is atomless. 
Informally, a \textit{symmetric} equilibrium is one in which all players of the same type play the same strategy. 
Therefore, if our model satisfies the requirements of the above theorem, a pure strategy MFE exists for our game.
We have already established that $J$ is a finite action space, and it is also a compact metric space in $\R^2$. 
Because $J$ is finite, it follows that the value function $V_{\text{MFG}}((j^1, j^2), \theta, \bell)$ is continuous in $(j^1, j^2)$.
In addition, $V_{\text{MFG}}((j^1, j^2), \theta, \bell)$ is continuous in $\bell$. 
It is a continuous function of $\pi_{\text{MFG}}$, which is a continuous function of $r_{\text{MFG}}$, which is also continuous in $\bell$. 
Finally,  \Cref{asm:theta} establishes the necessary requirements on $\Theta$.
Our model therefore satisfies all the requirements of the above theorem, and so a pure strategy MFE must exist in our game.
\end{proof}

\subsubsection*{Numerical Results for the MFG}
In this section, our goal is to determine whether an observed Uniswap distribution can be represented as a MFE. We will work with the distribution shown in \Cref{fig:mfg-a}, and we will call it $\bell^0$. We note that the calibration of $\Theta$ is much simpler in the current MFG setting because we do not need to work with an $(N-1)$-dimensional random vector representing the actions of the other players. Instead, the representative player reacts to the target liquidity distribution $\bell^0$ directly. A discussion of the increased efficiency of the MFG calibration, as well as the specifics of our calibration method, are provided in \Cref{sec:method_mfg}. The joint distribution of the discrete variables $k$ and $\delta$ is shown in \Cref{tab:mfgweights}, and the densities of $\lambda$ for each possible configuration of $(k, \delta)$ are shown in \Cref{fig:mfgweights}. The numerical calculations took 5.295 minutes, so the calibration for the MFG was over 20 times faster than for the 10-player $N$-player game. 

Once $\Theta$ is calibrated, we can approximate the MFE via iterated best response calculations. The resulting liquidity distribution $\bell^*$ and the corresponding token reserves are shown in \Cref{fig:mfg-b} and \Cref{fig:mfg-d}. In our experiment, the Wasserstein-1 distance between $\bell^0$ and $\bell^*$ was 4.140 with an R-score of 0.9999. This is a 77\% decrease in the Wasserstein-1 distance from $\bell^0$ and a 9.879\% increase in the R-score when compared to the $N$-player game. In addition, this shows that when players have beliefs similar to those we have set, then \Cref{fig:mfg-a} is very close to being a MFE.

\begin{figure}
    \centering
    \includegraphics[width = 0.8\textwidth]{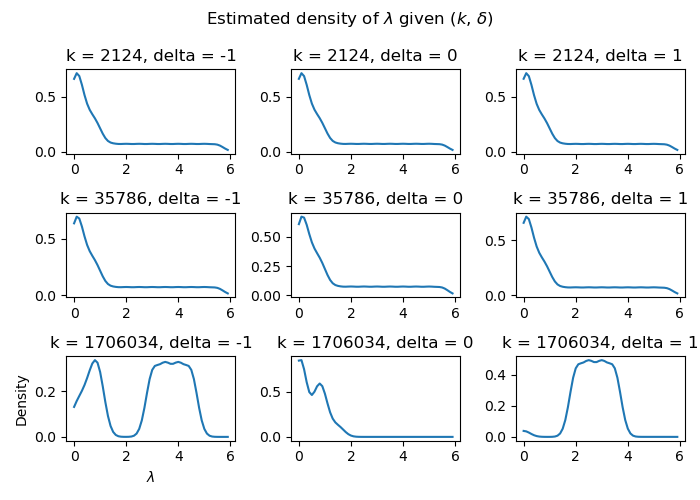}
    \caption{The estimated conditional distributions of $\lambda(\theta)$ based on $(k(\theta), \delta(\theta)) \in \{2124, 35786, 1706034\} \times \{-1, 0, 1\}$ for the MFG scenario. }
    \label{fig:mfgweights}
\end{figure}

\begin{table}
    \centering
    \begin{tabular}{|c|c|c|c|c|c|}
    \hline
        $(k, \delta)$ & $\delta = -1$ & $\delta = 0$ & $\delta = 1$ & Total & \cellcolor{gray!25}Observed\\
         \hline
        $k = 2124$ & 0.161 & 0.157 & 0.158 & 0.477 & \cellcolor{gray!25}0.4725\\ 
        \hline
        $k = 35786$ & 0.150 & 0.148 & 0.153 & 0.451 & \cellcolor{gray!25}0.4894\\
        \hline
        $k = 1706034$ & 0.018 & 0.042 & 0.012 & 0.072 & \cellcolor{gray!25}0.0381 \\
        \hline
        Total & 0.329 & 0.348 & 0.323 & 1.0 & \cellcolor{gray!25}1.0\\
        \hline
    \end{tabular}
    \caption{Joint probability distribution over the discrete variables $k$ and $\delta$, with the calibrated distribution of capital compared to the observed distribution of capital. The observed capital distribution is discussed further in \Cref{sec:appendix_calib_k}.}
    \label{tab:mfgweights}
\end{table}

Based on our calibration, the most important factor in determining the MFE distribution was each LP's risk-aversion, $\lambda$. More risk-averse players preferred wider liquidity positions to guard against sudden fluctuations in pool and/or market exchange rates, as expected. In addition, each player's trend belief $\delta$ contributed to the skew of its position. Players with more capital $k$ were also more likely to choose wider ticks for a fixed $\lambda$. From the final row of \Cref{tab:mfgweights}, we can see that our calibration assigned relatively equal weight to each of the three belief trends, which is reasonable and supports the hypothesis that LPs do not have concrete knowledge regarding future ETH prices. In addition, the distribution of capital among the three sizes of LPs which arises from the calibration is surprisingly consistent with the distribution shown in the ``Observed" column of the table, which is based off of the the data analysis done in \Cref{sec:appendix_calib_k}. Our analysis found that approximately 47.25\% of players are ``small", 48.94\% are ``medium", and 3.81\% are ``large". This indicates that our calibration is able to identify important characteristics of the model which are observable and therefore verifiable. This lends support to the accuracy of calibrated parameters that are more difficult to observe in reality. The accuracy of the resultant simulated liquidity distribution also supports the descriptive power of our model. 

\begin{figure}
    \centering
     \begin{subfigure}{0.45\textwidth}
         \centering
         \includegraphics[width=0.9\textwidth]{Figures/01-26-24-obsliq.png}
         \caption{The observed Uniswap ETH/USDC pool liquidity distribution on November 29, 2023. We call this $\bell^0$. \label{fig:mfg-a}}
     \end{subfigure}
     \begin{subfigure}{0.45\textwidth}
         \centering
         \includegraphics[width=0.9\textwidth]{Figures/01-26-24-obstokens.png}
         \caption{The observed Uniswap ETH/USDC tokens on November 29, 2023 corresponding to $\bell^0$. \label{fig:mfg-c}}
     \end{subfigure}
     \begin{subfigure}{0.45\textwidth}
         \centering
         \includegraphics[width=0.9\textwidth]{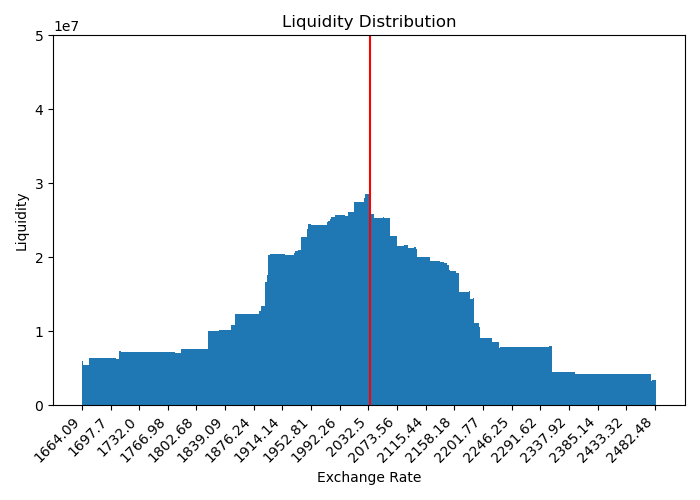}
         \caption{The MFG equilibrium $\bell^*$ generated by a MFG starting from $\bell^0$. \label{fig:mfg-b}}
     \end{subfigure}
     \begin{subfigure}{0.45\textwidth}
         \centering
         \includegraphics[width=0.9\textwidth]{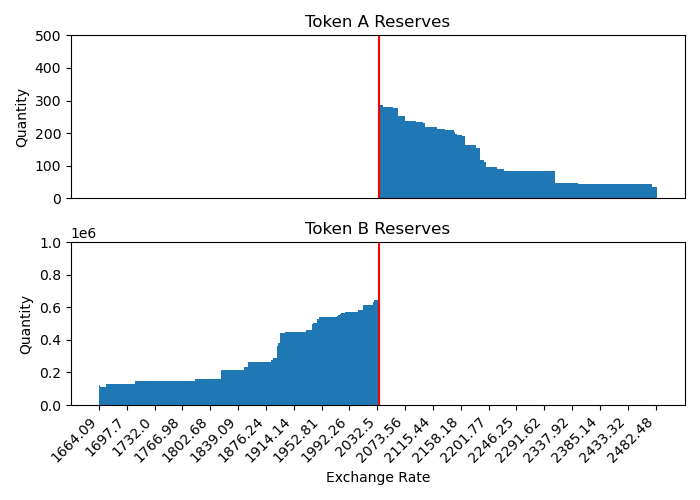}
         \caption{The token reserves corresponding to the calculated equilibrium $\bell^*$. \label{fig:mfg-d}}
     \end{subfigure}
     \caption{Comparison of the observed Uniswap distribution (left) versus the calibrated liquidity distribution (right) for November 29, 2023.}
     \label{fig:l-mfg}
\end{figure}

\section{Predictive Performance}
\label{sec:mfgpredict}
We have calibrated LP types to suit a specific observed liquidity distribution and demonstrated that it is possible to imitate observed liquidity distributions using heterogeneous LPs in a mean field game setting. However, it is difficult to draw sound inferences about LP characteristics unless we can show that the calibrated LP types and weights are able to perform well when used in predictive modeling. Therefore, we will now demonstrate the performance of these LP types when predicting future pool exchange rate and liquidity distribution evolution. 

We collected data on the liquidity distribution of the same pool at the start and end of a 24 hour period starting at 4 pm on November 29, 2023 and continuing until 4 pm on November 30, 2023, and obtained market USDC/ETH exchange rates with 1-minute granularity. We then simulated a pool with a mean field of many small LPs, subject to the arrival of many random swappers during the 24 hour period. We also introduced an additional detail - LPs still consider a horizon of $T=7200$ blocks (24 hours) when making their optimization decision, but are able to re-adjust their liquidity every 900 blocks (3 hours) if they desire. Therefore, there were $H=8$ ``mini" periods and we defined block intervals $t_0 = 0, t_1 = 900, t_2 = 1800, \ldots t_H = 7200$. 
At period $h$, the LPs play a game according to $\bell^{h-1}$, $p^*_{t_h}$, and $m^*_{t_h}$ which will result in a new distribution $\bell^h$. For the next 3 hours, $\bell^h$ will be fixed. Swappers will arrive with transactions and LPs will collect fees. The market exchange rate evolves according to the data. 

The evolution of the pool and market exchange rates is shown in \Cref{fig:er_comp}. We found that the pool exchange rate appears to follow the market rate in a manner similar to the true relationship, and the MAPE is fairly small at 0.1357\%. This indicates that our model is capturing the essence of the interactions between swappers and LPs, and between the pool exchange rate and the market exchange rate. This also supports our hypothesis that pool and market consistency does not need to come only from perfectly-timed and perfectly-sized arbitrageurs, but rather can be obtained via semi-strategic swappers looking for a ``good deal" based on the difference between the pool and market exchange rates. 
\begin{figure}
    \centering
    \begin{subfigure}{0.45\textwidth}
         \centering
    \includegraphics[width=\textwidth]{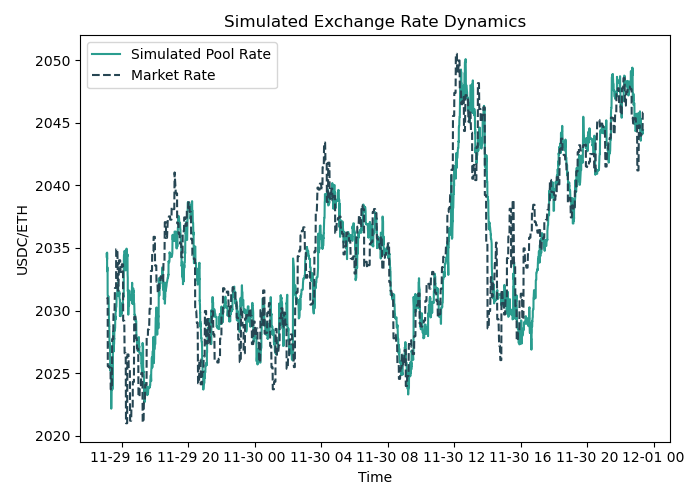}
     \end{subfigure}
     \begin{subfigure}{0.45\textwidth}
         \centering
    \includegraphics[width=\textwidth]{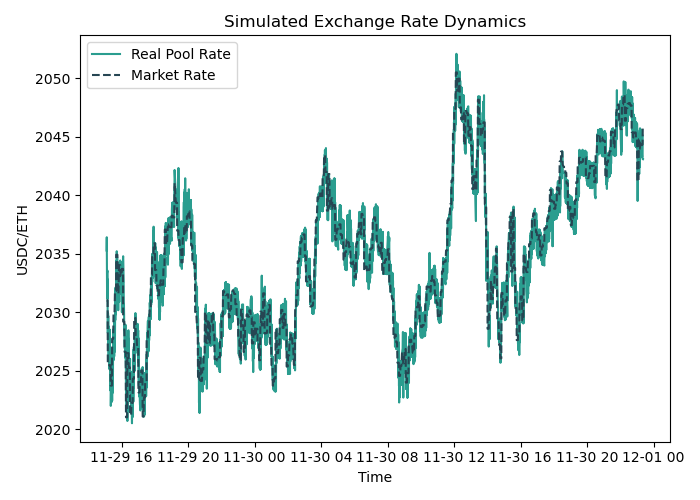}
     \end{subfigure}
    \caption{Visualization of simulated (left) and real (right) pool exchange rates in comparison to the market rate.}
    \label{fig:er_comp}
\end{figure}
\begin{remark}
    As a potential avenue of future work, we suggest that if pool parameters, such as the swap distribution parameters and the players' belief parameters, were re-calibrated more frequently, investors might get an even more accurate picture of the future. However, in our paper we wish to demonstrate that the long-term predictive power of our model is impressive even without frequent re-calibration. 
\end{remark}
In \Cref{fig:liq_future}, we see that the final simulated liquidity distribution has moved in accordance with the new pool exchange rate, and also bears a resemblance to the true liquidity distribution on November 30. The Wasserstein-1 distance between the simulated and real November 30 distributions is 4.4313. This shows that LPs using the MFG optimal strategy will adjust their liquidity over time as the pool exchange rate evolves in a way which is consistent with the observed real-world strategies of LPs. 
Given the consistency of the simulated exchange rate and liquidity distribution evolutions, as well as the previously discussed consistency of the joint distribution of $k$ and $\delta$, we feel confident in inferring that our approximation of $\Theta$ is accurate.
\begin{figure}
    \centering
    \begin{subfigure}{0.45\textwidth}
        \centering
        \includegraphics[width=\textwidth]{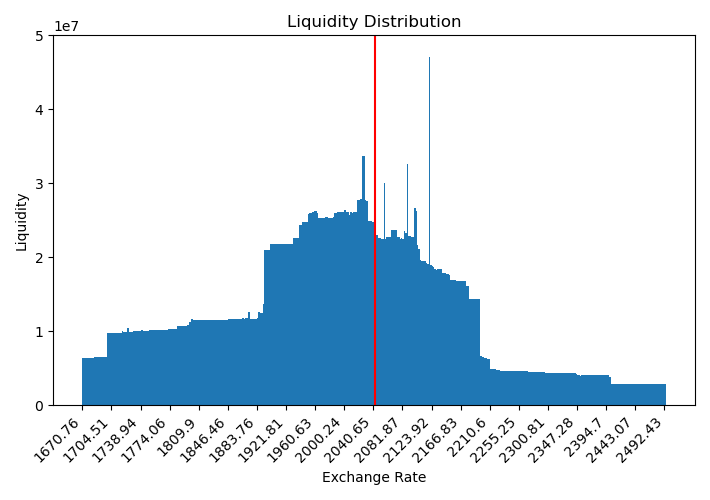}
     \end{subfigure}
     \begin{subfigure}{0.45\textwidth}
        \centering
        \includegraphics[width=\textwidth]{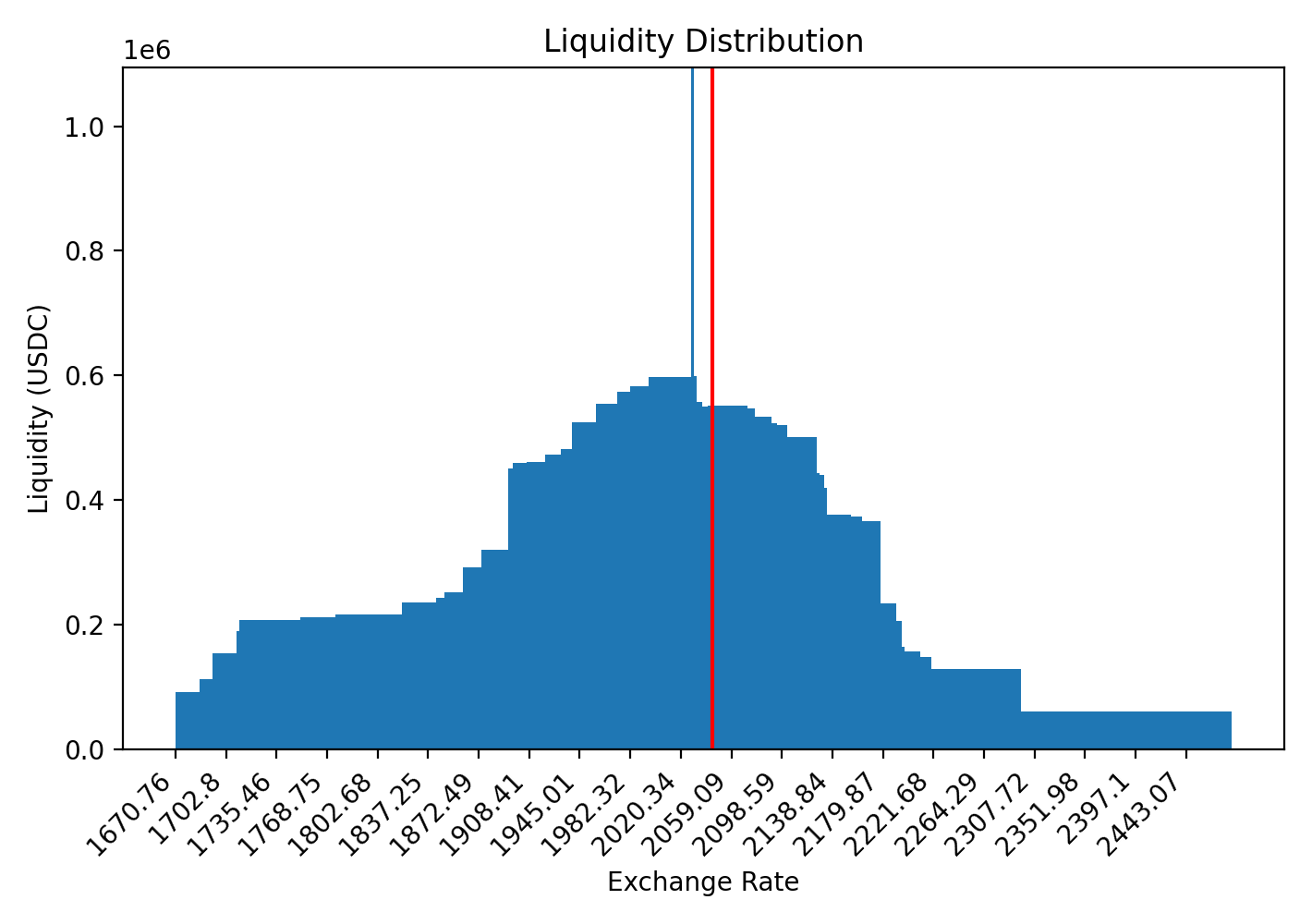}
    \end{subfigure}
    \caption{The final simulated liquidity distribution (right) compared to the observed liquidity distribution on November 30, 2023 (left).}
    \label{fig:liq_future}
\end{figure}

While these results are promising, we note that there are still certain discrepancies between the real and simulated exchange rates and liquidity distributions. One such issue is that the error between the observed and simulated November 30 liquidity distribution (4.4313) is almost double the error from between the observed and calibrated distributions on November 29 (2.3902). In addition, we see that the observed pool exchange rates follow the market exchange rate almost perfectly (MAPE of 0.063\%) while there is more than twice as much error when looking at the simulated pool exchange rate (MAPE of 0.1357\%). One possible explanation for these discrepancies is a missing element. One possibility is the absence of bots. Therefore, in the following section we introduce a new formulation for the pool in which LPs and bots interact, and we subsequently examine the effects of bots on the accuracy of model simulations.

\section{A Stackelberg Game with Bots}
\label{sec:bots}
We have thus far considered only a game between LPs. In reality, LPs are also affected by the behavior of other participants in the pool. We have already accounted for swappers, and assumed that their behavior is stochastic and correlated only with the pool's arbitrage level, but MEV bots are much more strategic. 
In addition, bot attacks are relevant to our study of DEXs because bot transactions affect exchange rate dynamics and LP profits, just as swapper transactions do. More importantly, these bot transactions also affect the attractiveness of a pool for ``regular" swappers and LPs. Logically, if bots are making money, someone must be losing money. In the case of JIT liquidity attacks, the LPs are missing out on fees that they would have otherwise earned from large transactions. Recall that JIT liquidity attacks are attacks in which a bot places a large liquidity deposit in the active tick right before a large transaction and removes it immediately after, thereby receiving a majority of the fees from the swap. An example of an observed JIT liquidity attack is displayed in \Cref{fig:attack_viz}. Notice that the Bot's liquidity addition and removal are symmetric and are an order of magnitude larger than the swap.

\begin{figure}
    \centering
    \begin{subfigure}{0.45\textwidth}
        \centering
        \includegraphics[width=\textwidth]{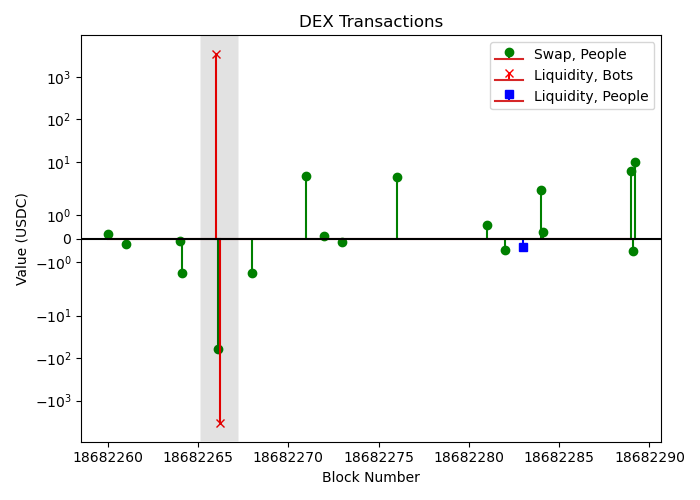}
     \end{subfigure}
     \begin{subfigure}{0.45\textwidth}
        \centering
        \includegraphics[width=\textwidth]{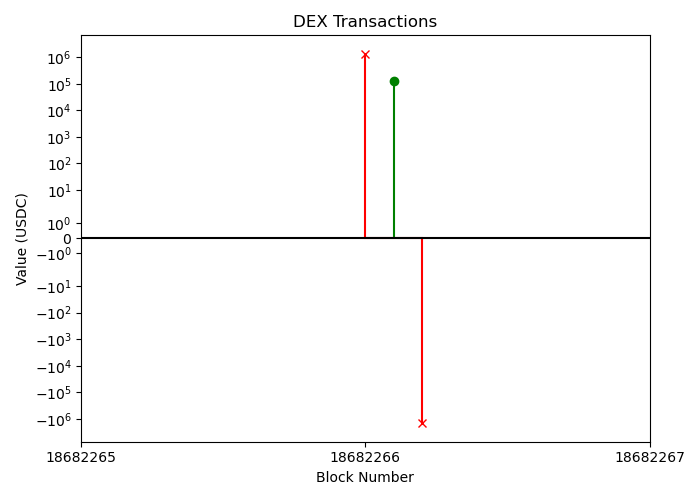}
    \end{subfigure}
    \caption{A real JIT liquidity attack from the Uniswap v3 ETH/USDC pool with 0.05\% fee is shown in the shaded region (left), with the zoomed-in version (right). }
    \label{fig:attack_viz}
\end{figure}

\subsection{Frequency of Liquidity-Based Attacks}
To identify sandwich attacks, we use a heuristic very similar to the one used in \cite{Xiongetal}. Our criteria for identifying a sandwich attack are:
\begin{enumerate}
    \item Two transactions of opposite directions which are separated by a single swap transaction. The two ``sandwich" transactions can either be two swaps in opposite directions (for a swap-based sandwich attack) or a liquidity addition then a liquidity removal (for a JIT liquidity attack).
    \item The sandwich transactions are done by the same user, and the middle transaction is done by a different user.
    \item All three transactions occur in the same block.
    \item The sandwich transactions must be symmetric in either token A or token B, with an allowed error of 5\%. This allows for possible ``failed" attacks which lose money, while eliminating scenarios where a single user places two unrelated transactions.
\end{enumerate}
Criteria 1-3 are the same as those in \cite{Xiongetal}, while we added criterion 4. This was necessary because we found that otherwise, transactions made by swap routers which help users place cost-effective swaps (for example, the 1inch protocol)  would also be counted as attacks. In reality, we found that transactions placed by routers were almost always placed on behalf of two unrelated users. For an analysis of DeFi protocols like swap routers, the reader can refer to \cite{kitzler}, among others. The analysis can be found in our Github at XXX, and we show a screenshot of an example swap-based sandwich attack in \Cref{fig:swapattack_ex} and an example JIT liquidity attack in \Cref{fig:liqattack_ex}. We note that in these figures, negative token amounts mean that the Bot is paying tokens and positive tokens amounts mean that the Bot is receiving tokens. 
\begin{figure}
    \centering
    \includegraphics[width = \textwidth]{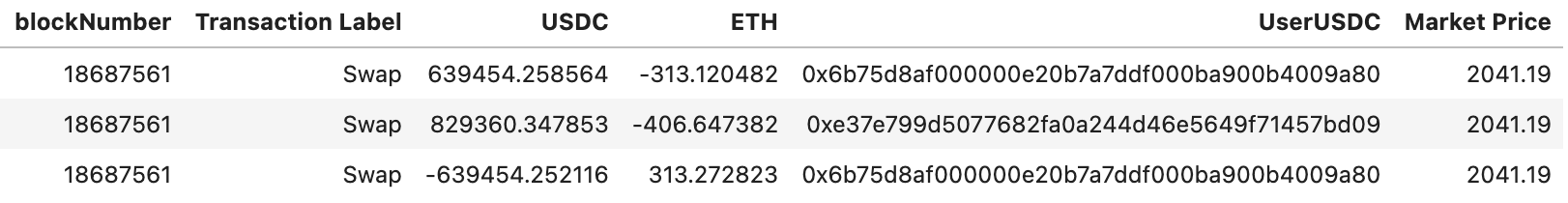}
    \caption{Example of a swap-based sandwich attack. Observe that the first and third transactions are nearly symmetric and are performed by the same entity. Since they are both swaps, this is a swap-based attack. Based on the market price, the Bot earned \$310.59.}
    \label{fig:swapattack_ex}
\end{figure}

\begin{figure}
    \centering
    \includegraphics[width = \textwidth]{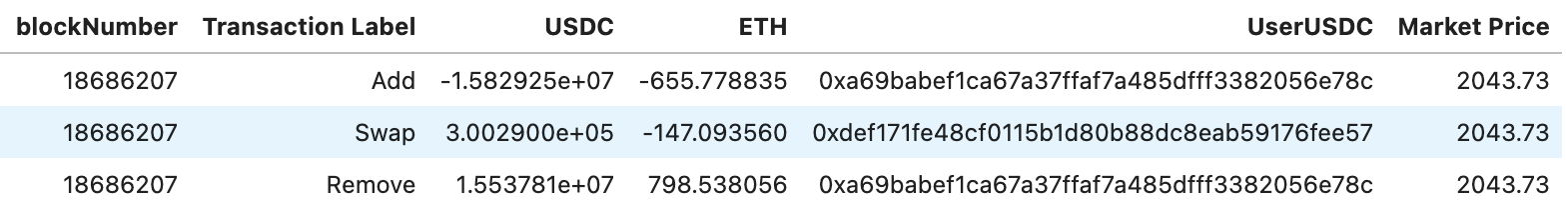}
    \caption{Example of a JIT liquidity attack. Observe that the first and third transactions are nearly symmetric and are performed by the same entity. Since one is a liquidity addition and one is a liquidity removal, this is a liquidity-based attack. Based on the market price, the Bot earned \$319.83.}
    \label{fig:liqattack_ex}
\end{figure}

Our analysis found that, over time, the portion of sandwich attacks that were JIT liquidity attacks increased. When looking at data from November 1-30, 2022, we found that only 48.84\% of attacks were liquidity-based. However, when looking at data from May 25-June 25, 2023, we found that 90.75\% of attacks were liquidity-based attacks and when looking at data from September 1-30, 2023, we found that 92.67\% of attacks were liquidity-based. 
Therefore, it seems likely that bots are turning increasingly towards liquidity-based attacks, perhaps due to new technologies like the strategic slippage constraints discussed in \cite{HeimbachSandwich}. Based on the analysis of \cite{Xiongetal}, these attacks directly affect LPs. Therefore, we wish to model a game between LPs and a Bot performing JIT liquidity attacks.

\subsection{Stackelberg Game Setup}
With this in mind, we wish to investigate whether LPs anticipate bot attacks. We do this by introducing a Stackelberg game in which many LPs consider the impacts of JIT liquidity attacks when choosing their optimal liquidity positions and a Bot reacts to the resulting liquidity distribution. 
The mean field of LPs are the ``leaders" in this game. This is because the pool's liquidity distribution is formed before the Bot arrives to attack and the Bot cannot alter the overall distribution, nor does it want to. The LPs determine the form of the CPMM which determines how swaps execute in the pool. The Bot, on the other hand, is the ``follower" who chooses whether to perform a sandwich attack based on the actions of the mean field of LPs. In this way, Bots react to the actions of the mean field of LPs and the mean field of LPs anticipates the actions of the Bot when forming the liquidity distribution.

We further differentiate the Bot and the LPs by the information they have at the start of the game. The LPs have imperfect knowledge of the size of incoming swaps $\{\xi_t\}_t$, but know the distribution of swaps as discussed in \Cref{sec:swap_calibration}. The Bot is the follower so it observes the equilibrium liquidity distribution $\bell^*$, and it also has more knowledge regarding the size of each incoming swap $\xi_t$ (perhaps via the blockchain's mempool or from Flashbots). More precisely, at block $t$, the Bot observes the transaction waiting in the mempool based on the value of a random variable $Z_t \sim Bernoulli(\zeta)$. This randomness encapsulates two concepts: (1) a transaction might not remain in the mempool long enough to be observed and attacked by a bot, and (2) bots might be reluctant to attack every single eligible transaction for fear that swappers would respond by ceasing to utilize the DEX for transactions or LPs would withdraw their liquidity. This second point could be modeled more rigorously in a variety of ways (for example, a secondary game between bots and swappers or a behavior-dependent arrival rate), but for now we prioritize the readability of the model and rely on $Z_t$ to approximate these various effects. If $Z_t = 1$, the Bot can see the incoming transaction and can decide whether to attack it. We simplify the Bot's optimization problem as a discrete decision problem of whether to attack $(x=1)$ or ignore $(x=0)$ the incoming swap $\xi_t$ at block $t$. 

\begin{remark}
    We assume for now that there is a single Bot which can potentially make multiple bot attacks, because the identity of the Bot is irrelevant from the LPs' point of view. However, future works could consider an additional game between bots. 
\end{remark}

The LPs are aware of the presence of the Bot, and so choose their liquidity position while anticipating Bot activity. In this way, we form a Stackelberg game between the LPs and the Bot. The LPs still optimize a mean-variance utility function based on their individual risk-aversion levels, while the Bot moves so quickly that it faces negligible risk and so only wishes to maximize its profit (as determined by the difference in the value of its final versus initial holdings). We assume that the Bot will convert its holdings back to token B immediately after its (nearly instantaneous) attack, so we hold the current market exchange rate, $m^*_t$, fixed during the Bot's optimization. This allows us to analytically compute the Bot's strategy.

To summarize, first, the LPs play a game to determine a mean-field equilibrium $\bell^*$ given the current pool exchange rate $p^*$ contained in tick $i^*$. While playing this game, they take into account the Bot's expected behavior. The resultant $\bell^*$ fixes the overall liquidity distribution in the pool for the next $T$ blocks, while swappers and the Bot make transactions. In this interim period, for each block $t$ the following steps occur:
    \begin{enumerate}
        \item If $Z_t = 1$, the Bot will engage with an incoming swap and sees the value of $\xi_t$. Given $\bell^*$, the Bot decides whether to attempt a JIT liquidity attack based on the incoming swap $\xi_t$ as well as $p^*_t$ and $m^*_t$. If $Z_t = 0$, the Bot does not engage.
        \item If the Bot decides to attack (meaning it chooses $x=1$), it adds $L$ liquidity to tick $i^*$ by contributing token A and token B in the appropriate ratio. This forms a new liquidity distribution $\bell^1 = \bell_{\text{new}}(L, i^*, i^*+1) + \bell^*$ where the only affected tick is $i^*$, the active tick.
        \item The exogenous swapper arrives and executes its swap of $\xi_t$ token B exchanged for $\psi(\xi_t, \bell^1, p^*-t)$ token A. The amount of token B in the pool changes by $\xi$ and the amount of token A changes by $-\psi(\xi, \bell^1, p^*_t)$, where we introduced $\psi$ in \Cref{sec:dex} and its detailed definition is provided in \Cref{eq:psi_i}.
        \item The Bot withdraws its share of the updated liquidity in tick $i^*$. In doing so, the Bot receives the proportion $L/(\bell^*_{i^*} + L)$ of the new amounts of token A and token B in tick $i^*$, as well as the same proportion of the fees $\gamma |\xi_t|$. Non-bot LPs also receive their share (if any) of the fees from the swap.
    \end{enumerate}

\subsection{The Bot's Optimization Problem}
In this situation, the LPs are able to anticipate the strategy of the Bot with perfect accuracy, though they still lack concrete knowledge about the swap sequence. Therefore, we first solve the Bot's optimization problem before moving on to the LP's optimization problem. As discussed, the Bot has a large, fixed amount of liquidity $L$. We make the assumption that the Bot has a fixed amount of liquidity for ease of notation, but we could also assume that the Bot has a fixed capital endowment similar to the LPs. The Bot must choose whether or not to attack a transaction by choosing the value of an indicator variable $x \in \{0,1\}$. Thus, the Bot's value function for block $t$ is given by
\begin{equation*}
    V_B(x; \xi_t, L, \bell^*) = \frac{xL}{\ell^1_{i^*}}(\xi_t -m^*_t \psi(\xi_t, \bell^1, p^*_t) + \gamma |\xi_t|) -xG,
\end{equation*}
where $\bell^1 = \bell_{\text{new}}(L, i^*, i^*+1) + \bell^*$.
Because this optimization problem does not depend directly on the block number $t$, we will drop the subscripts on $\xi_t$, $p^*_t$, and $m^*_t$ for the rest of the section.
Notice that the Bot's decision affects its profit in two ways: firstly through the portion of fees earned, and secondly through the net change in the value (in USDC) of the Bot's token holdings. This depends on the form of the exchange function $\psi$ and therefore on $\bell^*$. 

\begin{remark}
Note that we choose to have Bots to consider gas fees ($G$, which we set as fixed for now) while LPs do not. Transactions in a blockchain are not executed ``for free". Instead, users must pay a gas fee to miners or validators to incentivize the inclusion of their transaction in an upcoming block. These fees range from less than a dollar to hundreds of dollars depending on the current total transaction volume and the urgency of a specific transaction, but the average fee for a Uniswap currency swap or liquidity adjustment hovers between \$10 -\$30, depending on the urgency of the transaction and its size. This is less significant for LPs who leave liquidity in the pool for hours or days, and more significant for the Bot which places multiple transactions in rapid succession. 
\end{remark}

We can solve this simple optimization problem to find thresholds $\bar \xi^+(\bell^*)$ and $\bar \xi^-(\bell^*)$ which will determine whether the Bot will profit from an attack. These values will determine a strategy $\bar \xi_B(\bell^*) = (\bar \xi^-(\bell^*), \bar \xi^+(\bell^*))$ based on $\bell^*$, where the Bot's choice of $x$ is given by
\begin{equation*}
    x(\xi; \bar \xi_B(\bell^*)) = \begin{cases} 1, & \xi \notin [\bar \xi^-(\bell^*), \bar \xi^+(\bell^*)], \\
    0, & \xi \in [\bar \xi^-(\bell^*), \bar \xi^+(\bell^*)].
    \end{cases}
\end{equation*}
The values of $\bar \xi^-(\bell^*)$ and $\bar \xi^+(\bell^*)$ can be determined by solving a quadratic equation and the exact formulas are provided in \Cref{sec:appendix_bot}. The Bot may choose a difference $x_t$ for each swap $\xi_t$ at block $t$, but its strategy based on $\bar \xi^+(\bell^*)$ and $\bar \xi^-(\bell^*)$ will not change. Furthermore, the LPs can also infer this strategy from their knowledge of $V_B$.

\subsection{The LP's Optimization Problem}
As previously stated, the LPs are perfectly rational and so can intuit the Bot's optimal strategy from the structure of the pool and the Bot's value function. Therefore, the LPs again engage in a MFG to determine the equilibrium liquidity distribution $\bell^*$, but this time they anticipate the actions of the Bot, given by $\bar \xi_B(\bell^*)$.
Therefore, if the initial liquidity distribution is $\bell^*$, we can write each LP's optimization problem as
\begin{equation*}
     \max_{j^1, j^2 \in [d]} V_{\text{LP}}(j^1, j^2; \theta, \bell^*, \bar \xi_B(\bell^*)),
\end{equation*}
where we define the value function as\begin{align*}
    V_{\text{LP}}(j^1, j^2; \theta, \bell^*, \bar \xi_B(\bell^*) ):= &\E_{\bm \xi, \delta(\theta)} \left[\pi_{\text{LP}}(j^1, j^2, \bm\xi, \bar \xi_B(\bell^*), \theta, \bell^*)\right] \\
    &- \lambda(\theta) \Var_{\bm\xi, \delta(\theta)}(\pi_{\text{LP}}(j^1, j^2, \bm\xi, \bar \xi_B(\bell^*), \theta, \bell^*)).
\end{align*}
In this scenario, the profit function $\pi_{\text{LP}}$ is defined as
\begin{equation*}
\pi_{\text{LP}}(j^1, j^2, \bm\xi,\bar \xi_B(\bell^*), \theta, \bell^*) := \sum_{t = 0}^T  R(\xi_t, Z_t x(\xi_t; \bar \xi_B(\bell^*)), j^1, j^2, \bell^*, p^*_t, \theta),
\end{equation*}
where $Z_t \sim Bernoulli(\zeta)$. 
Recall that $\bm \xi$ is the random sequence of swaps submitted to the pool by swappers, while $\bar \xi_B(\bell^*)$ is the Bot's strategy based on $\bell^*$. 
The effect of the Bot's choice $x$ on the LP's fees earned from a swap of size $\xi$ when the pool exchange rate is $p^*$ is expressed through the new LP reward function $R$, given by
\begin{equation*}
     R(\xi, j^1, j^2, \bell, p^*, \theta) := \begin{cases}
         \gamma |\xi|  \left(\frac{u(j^1, j^2, k(\theta), m^*_0)}{\ell_{i^*} + L} \right),  & x = 1\text{ and } j^1 \leq i^* < j^2,\\
          r(\xi, j^1, j^2, \bell, p^*, \theta), & \text{else,}
     \end{cases}
\end{equation*}
where $r$ is defined in \Cref{eq:reward}. Intuitively, this reward function means that if a Bot attacks transaction $\xi$ (meaning $x=1$ based on the Bot's strategy $\bar \xi_B(\bell^*)$), then tick $i^*$ will gain $L$ units of liquidity and $\xi$ will therefore only affect tick $i^*$. The LPs will still receive some proportion of the fees, but this proportion will be decreased by the addition of $L$. If $x=0$, then the LP reward will be calculated using the original reward function and the Bot will earn no reward. 

Note that since the Bot withdraws its liquidity after each attack, the liquidity distribution always reverts to the original $\bell^*$ between transactions. One could also allow $L$ to be random from the perspective of the LPs, but in this paper we assume that $L$ is known to all players. As in \Cref{sec:mfg}, the LP's optimization problem does not have an analytical solution, but we discuss numerical outcomes of the Stackelberg game in the following section.

\subsection{Numerical Results}
We now simulate a scenario where LPs anticipate the arrival of liquidity-based bot attacks. We use the same setup as \Cref{sec:mfg}, but assume that LPs know that transactions could be attacked by the Bot. Again, the optimization is difficult to solve analytically but can be solved numerically using Monte Carlo simulations. We calculate from Uniswap data that bots attack approximately 52.78\% of ``attackable" transactions, so we set $\zeta = 0.5278$.

We re-calibrate the LPs under these new conditions, using the same process as that described in \Cref{sec:mfg}, but now we include the Stackelberg game in the LP's optimization and Monte Carlo simulations. We obtain a new set of calibrated weights, shown in \Cref{fig:botweights}. We note that the densities are similar for LPs with ``small" and ``medium" capital endowments, but LPs with ``large" capital endowments are less risk-averse, as the calibration assigned more weight to smaller values of $\lambda(\theta)$ when $k(\theta) = 1706034$. The Wasserstein-1 distance between this liquidity distribution and the true November 29, 2023 liquidity distribution is 2.9433, and the R-score is 0.9999. 

\begin{figure}
    \centering
    \includegraphics[width = 0.5\textwidth]{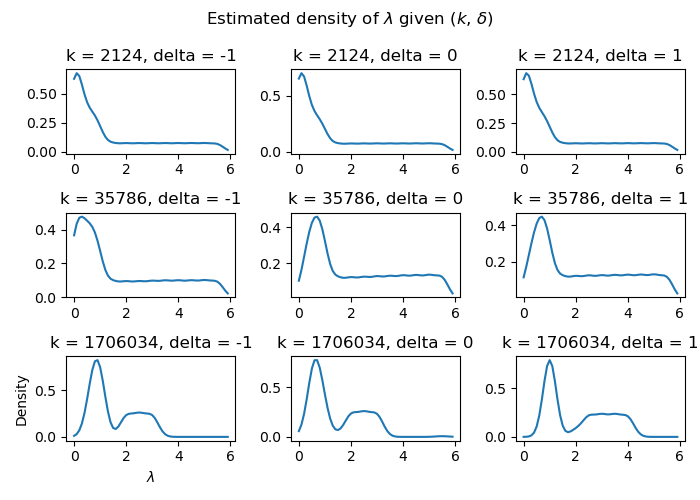}
    \caption{The estimated conditional distributions of $\lambda(\theta)$ based on $(k(\theta), \delta(\theta)) \in \{2124, 35786, 1706034\} \times \{-1, 0, 1\}$ for the Stackelberg scenario. }
    \label{fig:botweights}
\end{figure}

\begin{figure}
    \centering
     \begin{subfigure}{0.45\textwidth}
         \centering
         \includegraphics[width=0.9\textwidth]{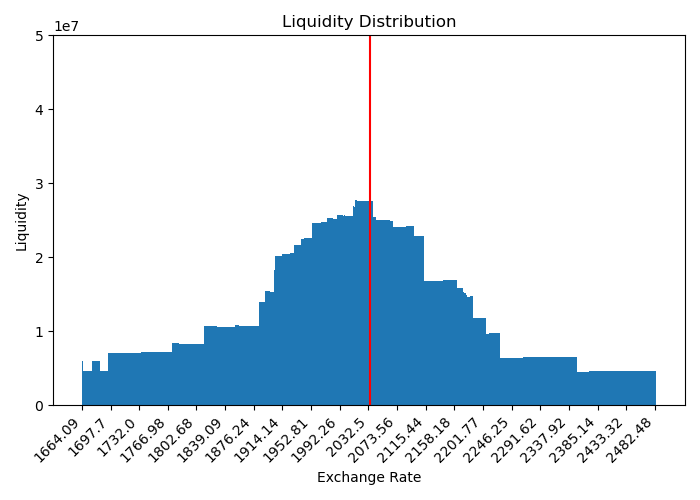}
         \caption{The MFG equilibrium $\bell^*$ generated by a MFG starting from $\bell^0$.}
     \end{subfigure}
     \begin{subfigure}{0.45\textwidth}
         \centering
         \includegraphics[width=0.9\textwidth]{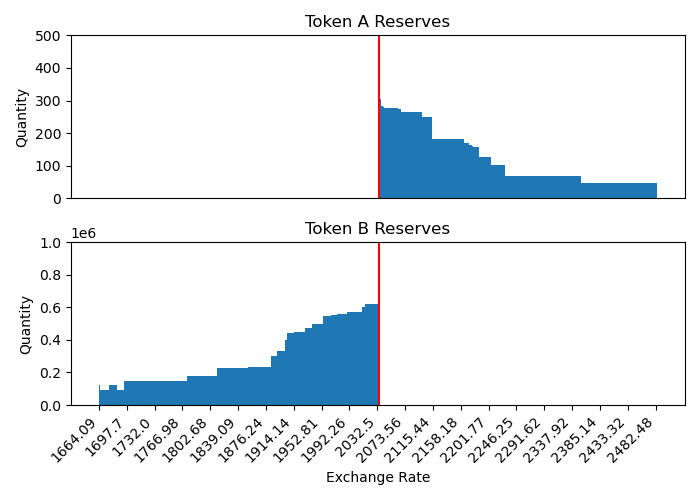}
         \caption{The token reserves corresponding to the calculated equilibrium $\bell^*$.}
     \end{subfigure}
     \caption{Comparison of the observed Uniswap distribution (left) versus the calibrated liquidity distribution (right) for November 29, 2023 when LPs anticipate the strategy of the Bot.}
     \label{fig:l-stack}
\end{figure}

We then repeat the simulation described in \Cref{sec:mfgpredict}, adding in the potential for bot attacks with probability $\zeta = 0.5278$. The resultant pool dynamics are shown in \Cref{fig:erBots}, and the liquidity distribution is shown in \Cref{fig:liq_futureBots}. The MAPE between the simulated pool exchange rate and the market exchange rate improved from the previous value of 0.1357\% to 0.0732\%, which is very close to the observed MAPE between the pool and market exchange rates, given by 0.063\%. In addition, the Wasserstein-1 distance between the simulated and real November 30 liquidity distributions decreased from 4.4313 to 3.384, an improvement of 21.54\%. Therefore, we conclude that Bots play a crucial role in accurate simulation of Uniswap dynamics. 

\begin{figure}
    \centering
    \begin{subfigure}{0.45\textwidth}
        \centering
        \includegraphics[width=\textwidth]{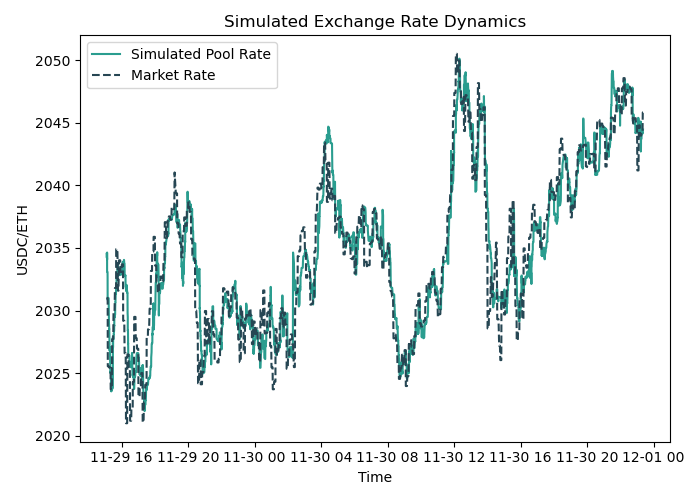}
     \end{subfigure}
     \begin{subfigure}{0.45\textwidth}
        \centering
        \includegraphics[width=\textwidth]{Figures/04-09-24-realER.png}
    \end{subfigure}
    \caption{Comparison of the simulated exchange rate when incorporating a Bot (left) to the observed exchange rate (right).}
    \label{fig:erBots}
\end{figure}

\begin{figure}
    \centering
    \begin{subfigure}{0.45\textwidth}
        \centering
        \includegraphics[width=\textwidth]{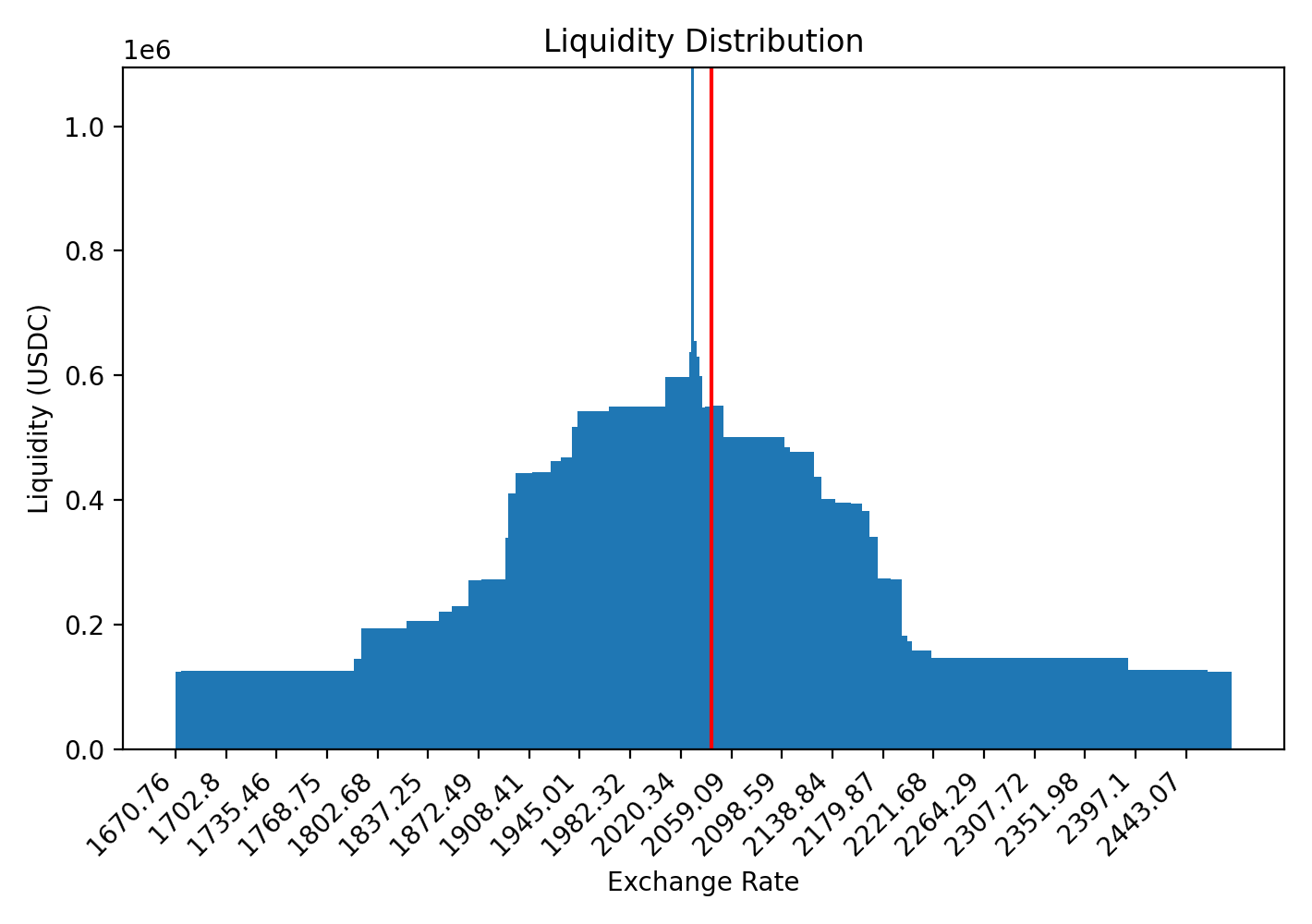}
     \end{subfigure}
     \begin{subfigure}{0.45\textwidth}
        \centering
        \includegraphics[width=\textwidth]{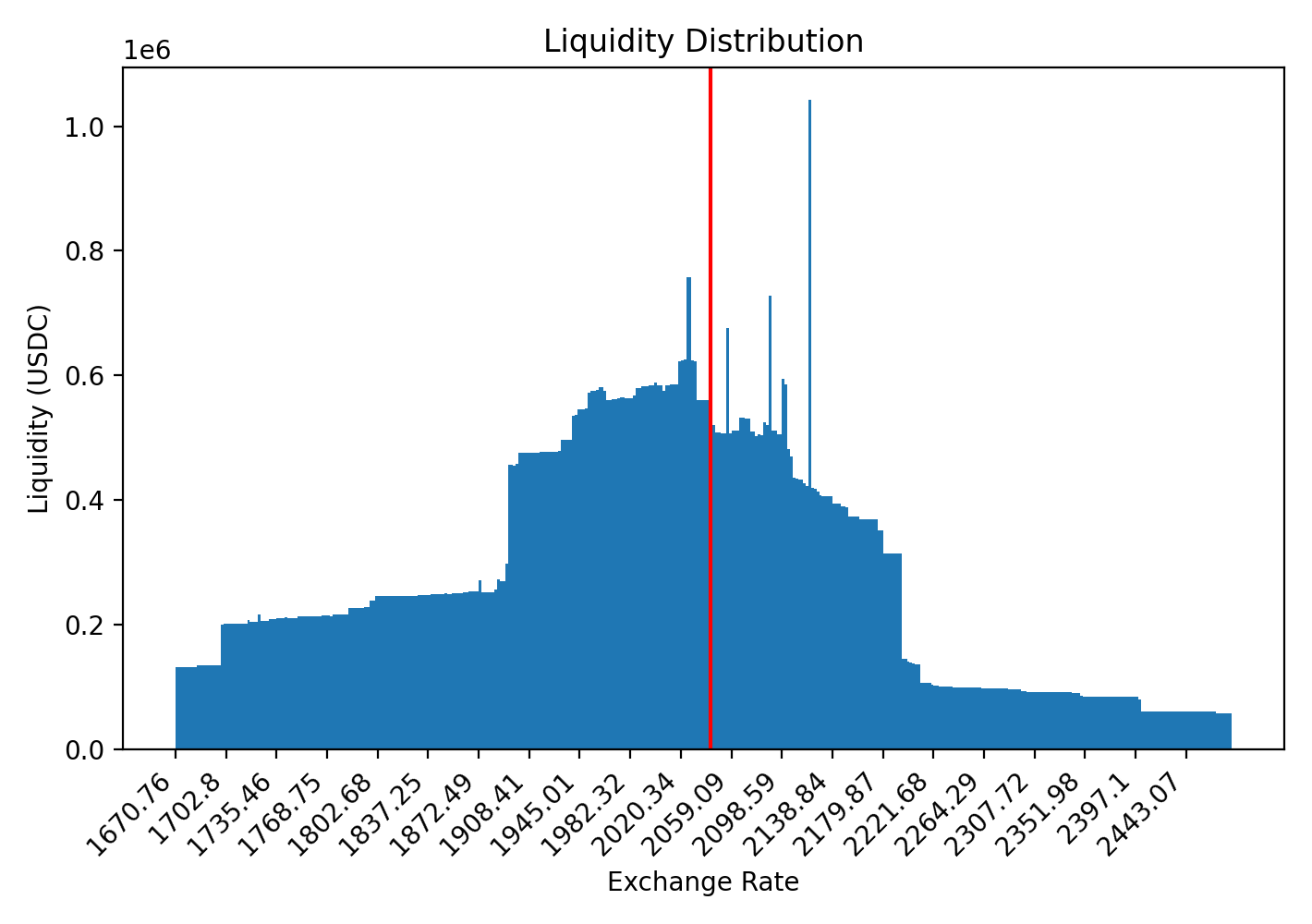}
    \end{subfigure}
    \caption{Comparison of the simulated final liquidity when incorporating a Bot (left) to the observed final liquidity distribution (right) on November 30, 2023.}
    \label{fig:liq_futureBots}
\end{figure}

Furthermore, in simulations, the LPs lose less profit to the Bot when they anticipate the Bot's presence and choose their liquidity positions accordingly. This suggests that LPs are able to mitigate some of the effects of Bot attacks by choosing their positions in a strategic manner. In our model we observed that LPs with the same risk aversion chose narrower positions in the Stackelberg scenario. Accordingly, we observed that the liquidity distribution formed by LPs in the Stackelberg game  is slightly more concentrated around the current active tick than the liquidity distribution formed by LPs in the simple mean field game. 
When there is more liquidity in the active tick supplied by LPs, then the Bot is not able to capture as large of a share of the liquidity in the tick, and therefore earns a smaller percentage of the fees generated by the swaps. 

The results of simulations performed using data from two different days and three different values of $L$ are shown in \Cref{tab:profits}. It can be seen that in all scenarios, Bot profit decreases when LPs anticipate the arrival of the Bot. In addition, we note that Bots experience diminishing returns for larger choices of $L$. In our simulations, we found that if the Bot makes a liquidity contribution equivalent to 5 million USDC in capital, on average it will earn 91\% of fees. Quadrupling the liquidity contribution to the equivalent of 20 million USDC only increases the amount of fees received to 97\%. Analysis of DEX data revealed that 50.5\% of liquidity-based attacks involved less than 10 million USDC of capital and 97\% used less than 20 million USDC of capital, which is consistent with our finding that 20 million USDC is sufficient to capture nearly all profit from fees when a liquidity-based attack is executed.
\begin{table}
    \footnotesize
    \centering
    \begin{tabular}{|c|c|c|c|c|}
        \hline
        Date & LPs Anticipate Bot? & $L$ & Bot Profit & Collective LP Profit \\
        \hline
        \multirow{6}{*}{ November 29-30, 2023} & No & 1 million USDC & 22212 & 128395 \\
         & Yes & 1 million USDC &  21763 & 128844 \\
        \cline{2-5}
         & No & 5 million USDC & 37050 & 113557 \\
         & Yes & 5 million USDC & 35135 &  115472\\
        \cline{2-5}
         & No & 20 million USDC & 39440 & 106012 \\
         & Yes & 20 million USDC & 35525 & 109927 \\
        \hline
        \hline
         \multirow{6}{*}{ December 12-13, 2023} & No & 1 million USDC & 75871 & 282924 \\
        & Yes  & 1 million USDC & 74935 & 283859 \\
        \cline{2-5}
         & No & 5 million USDC & 108388 & 250407 \\
         & Yes  & 5 million USDC & 99818 & 258977 \\
        \cline{2-5}
         & No & 20 million USDC & 110473 & 248322 \\
        & Yes  & 20 million USDC & 102932 & 255863 \\
        \hline
    \end{tabular}
    \caption{Difference in Bot vs. LP profit when LPs do or do not anticipate the presence of Bots.}
    \label{tab:profits}
\end{table}

\section{Conclusion}
We have created a flexible and descriptive model of a decentralized cryptocurrency exchange like Uniswap v3. The dynamics of the pool are simulated through the liquidity decisions of LPs and the swaps made by semi-strategic swappers who take into account the level of arbitrage in the pool. In this setting, we studied various types of competition between LPs in a DEX, starting with an $N$-player Bayesian game then moving to a MFG with a distribution of types over the population. Characteristics of each LP's type include its risk-aversion, capital endowment, and belief regarding market trends. We found that these characteristics were sufficient to produce a spectrum of liquidity positions of varying widths and skews that resemble reality. While the $N$-player game provided initial insight into the possible behaviors of LPs in a competitive setting, it was very expensive to calibrate and compute equilibria for the game. The MFG formulation provided an efficient approach for modeling competition among many LPs with different incentives, and we were able to produce accurate predictions of the future pool exchange rate and pool liquidity. In order to further increase the accuracy of our model, we then introduced JIT liquidity attacks performed by a MEV bot. We modeled a Stackelberg game between a mean field of LPs as the leader and the Bot as the follower. This led to a 21.54\% increase in the accuracy of the predicted liquidity distribution and a 46.1\% decrease in the overall MAPE between the market and pool exchange rates throughout our simulations. Furthermore, when risk-averse LPs took action to concentrate their liquidity positions in anticipation of Bot attacks, the LPs were able to increase their share of transaction fee profits.

Our work has important implications for both DEX users and DEX designers. For LPs themselves, our work provides a way to quantify the implications of MEV bot attacks and optimize their strategies. In another direction, our model could be utilized in order to price financial derivatives grounded in the DeFi space, as it provides a novel method to simulate exchange rate evolution and LP behavior in DEXs.

This work is a preliminary foray into the topic of algorithmic market makers, and we hope to expand on our research in the future. As discussed in the body of our paper, we have made certain simplifications regarding the frequency of bot attacks, the Bot optimization problem, and the frequency of LP movements. For example, we assume that players only adjust their liquidity at fixed time intervals and that all LPs adjust their liquidity simultaneously. In addition, the number of LPs is fixed throughout the games that we consider. In reality, LPs are able to enter, exit, and adjust their liquidity at any point in time. Future works could implement a more complex system where LPs are not playing a game in a perfectly synchronized manner or where the mean field experiences flows. This would elevate the static MFG to a stochastic differential mean field game, which would require more in-depth analysis. Another potential future work could build on the asymmetric information introduced in the Stackelberg game. For example, we could consider a scenario in which certain players have insider information regarding the future evolution of market exchange rates. As a final suggestion, we could take a different approach to modeling large liquidity providers by considering a pool with many small and one large player, akin to the setting studied in \cite{BayraktarMunk_Bets}.
There are many promising areas of future research in this field, as algorithmic market makers provide a unique new direction of study within the wider discussion of market microstructure.

\bibliographystyle{apalike}
\bibliography{bib_DEX.bib}

\appendix
\section{Detailed Description of CPMM}
\label{sec:appendix_dex}
In \Cref{sec:dex}, we provided the basic overview of the mechanism of a CPMM. For the curious reader, we now provide the exact mathematical methods used to calculate liquidity contribution and swaps.

\subsection{Pool Liquidity}
\label{sec:cpmm_liq}
We first explain the mathematical relationship between different pool characteristics when the pool is static and no transactions or liquidity adjustments are taking place. Recall that we use the term \textit{liquidity} to describe the overall ``capacity" of the pool for handling swaps, and represent it mathematically as the vector $\bell := \{\ell_i\}_{i=1}^d \in \R^d_+$. In this section, we will provide a rigorous mathematical definition of liquidity and its relationship to the number of tokens in the pool, as well as a description of how liquidity providers contribute to the pool.

We previously discussed that the token reserves in any tick $i$ must satisfy the CPMM function
\begin{equation}
\label{eq:cpmm}
    \Big(B_i + \ell_i\sqrt{p_i}\Big)\Big(A_i + \ell_i/\sqrt{p_{i+1}}\Big) = \ell_i^2,
\end{equation}
and the reserves in the active tick $i^* = i(p^*)$ must also satisfy the pool exchange rate condition given by
\begin{equation*}
    p^* = \dfrac{B_{i^*} + \ell_{i^*}\sqrt{p_{i^*}}}{A_{i^*} +
    \ell_{i^*}/\sqrt{p_{i^*+1}}}.
\end{equation*}
Furthermore, if we instead look at a tick $i \neq i^*$, then the following requirement holds:
\begin{equation*}
    \begin{cases}
        A_i = 0, & i < i(p^*),\\
        B_i = 0, & i > i(p^*).
    \end{cases}
\end{equation*}
The above equations provide a mathematical representation of the difference between liquidity and token reserves.
We remind the reader that the liquidity in each tick, $\ell_i$, has its own units, and does not have a direct monetary interpretation. 
However, it is clear that $p^*$, $\bell$, $\bm{A}$, and $\bm{B}$  are intricately connected. From the perspective of an LP, we must take the current pool exchange rate $p^*$ as given and calculate an optimal strategy by looking at pool liquidity instead of token amounts. Using this approach, we only need to consider $d + 1$ quantities, instead of $2d$ quantities (if we tracked the amounts of token A and token B in each tick). However, the amount of tokens in tick $i$ can always be calculated using $p^*$ and $\ell_i$ via
\begin{align}
\label{eq:L2}
    A_i(\ell_i; p^*)&= \begin{cases}
        0,& i < i(p^*), \\
        \ell_i(1/\sqrt{p^*} - 1/\sqrt{p_{i+1}}),& i = i(p^*), \\
        \ell_i(1/\sqrt{p_i} - 1/\sqrt{p_{i+1}}),& i > i(p^*),
    \end{cases}\\
    B_i(\ell_i; p^*) &= \begin{cases}
        \ell_i(\sqrt{p_{i+1}} - \sqrt{p_i}),& i < i(p^*), \\
        \ell_i(\sqrt{p^*} - \sqrt{p_i}),& i = i(p^*), \\
        0,& i > i(p^*), 
    \end{cases}
\end{align} 
for all $i \in [d]$. For ease of notation, we will at times suppress the dependence of $A_i$ and $B_i$ on $\ell_i$ and $p^*$ in the future.

Note that only tick $i^*$ contains tokens of both A and B. Ticks above $i^*$ only contain token A, while ticks below $i^*$ only contain token B. This is because the pool is designed such that the token reserves in tick $i > i^*$ will first be used when the pool exchange rate reaches $p_i$, and thus
\begin{equation*}
    \dfrac{B_i(\ell_i; p^*) + \ell_i\sqrt{p_i}}{A_i(\ell_i; p^*) +
    \ell_i/\sqrt{p_{i+1}}} = \dfrac{0 + \ell_i\sqrt{p_i}}{\ell_i/\sqrt{p_i} - \ell_i/\sqrt{p_{i+1}} + \ell_i/\sqrt{p_{i+1}}} = \dfrac{\ell_i\sqrt{p_i}}{\ell_i/\sqrt{p_i}} = p_i.
\end{equation*}
Similarly, the token reserves in tick $i < i^*$ will first be used when the pool exchange rate reaches $p_{i+1}$, and so only contain token B. Further details about the relationship between swaps and the pool exchange rate $p^*$ are provided in \Cref{sec:swap_explanation} below.

\begin{remark}
    Note that the CPMM for a tick $i$ depends not only upon $A_i$ and $B_i$, but also on $\ell_i\sqrt{p_i}$ and $\ell_i/\sqrt{p_{i+1}}$. These two quantities contribute to the {\rm virtual reserves} in the tick, and serve to shift the constant product function and ensure the differentiability of the full-pool exchange function. The concept of virtual reserves was introduced alongside concentrated liquidity in \cite{Uniswapv3}, and the shift ensures that each tick's exchange rate curve intersects the axes at points that correspond with $p_i$ and $p_{i+1}$. The concept is illustrated in \Cref{fig:virtualres} and the mathematics is discussed further in \Cref{sec:swap_explanation}.
\begin{figure}[H]
    \centering
    \includegraphics[width=0.4\textwidth]{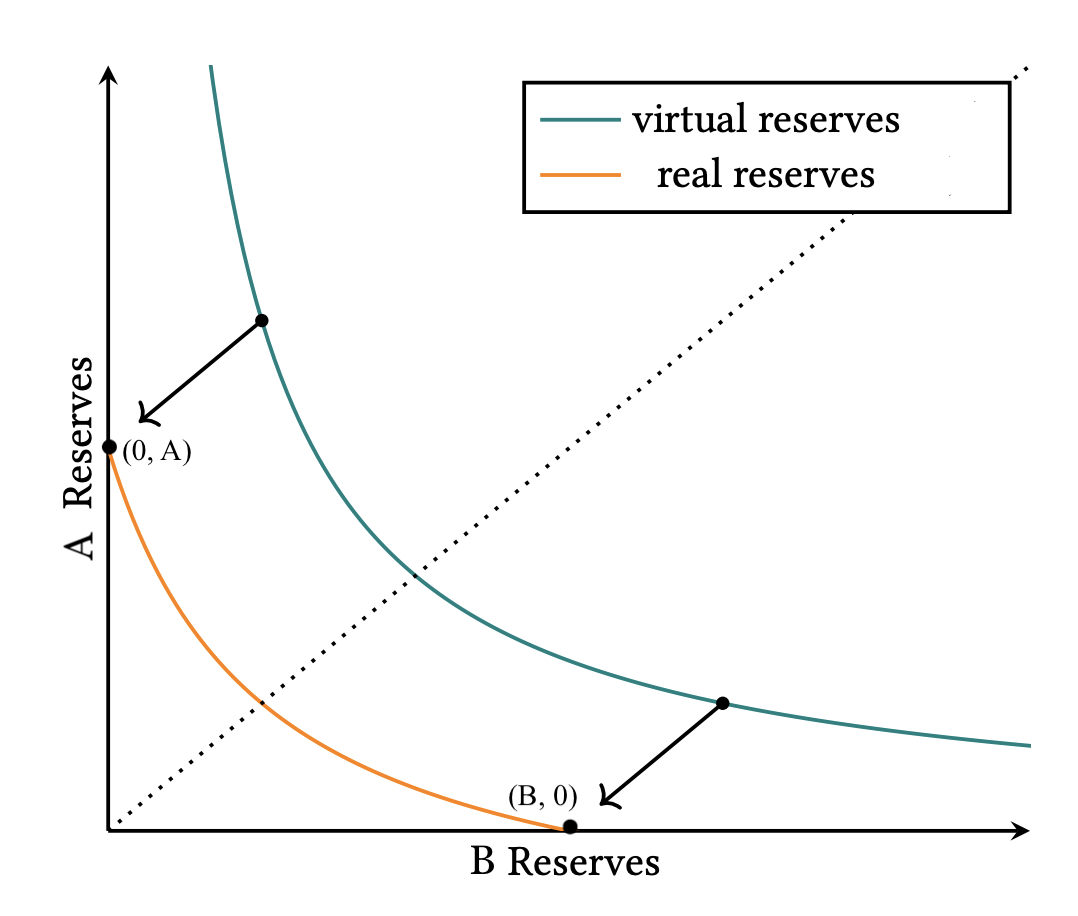}
    \caption{A conceptual illustration of real versus virtual reserves in Uniswap v3. The points $(B,0)$ and $(0, A)$ represent the real reserves corresponding to an exchange rate of $p_i$ and $p_{i+1}$, respectively.}
    \label{fig:virtualres}
\end{figure}
\end{remark}

\subsubsection*{Detailed Description of Liquidity Addition Process}
Now that we have established the relationship between liquidity, token reserves, and pool exchange rate, we will discuss exactly how LPs are able to participate in this exchange. In our model, LPs contribute to the liquidity pool by choosing a range of exchange rates in which to contribute and an amount of capital to contribute. 

Recall that we assume that LPs value their total capital using the {\it market exchange rate}, which we denote as $m^*$ and that the LP has a fixed amount of capital (as valued in terms of token B) that it is able to contribute to the pool. 
The contributed capital is distributed among the chosen ticks such that each tick in the range receives the same amount of additional liquidity. Note that this is not equivalent to contributing the same number of tokens to each tick. However, it is straightforward to calculate the token contributions for a given choice of exchange rate range and capital. To do this, we first determine the appropriate token \textit{ratio} by calculating the token contributions necessary to add one unit of liquidity to each tick in the selected position. Then, these token contributions can be scaled such that the LP's contribution has the desired value in terms of token B.
Let us say that the LP chooses a range of exchange rates $[p_{j^1}, p_{j^2})$ corresponding to ticks $\{j^1, \ldots, j^2-1\}$, where $j^1 < j^2$ and $j^1, j^2 \in [d+1]$. This is the range of exchange rates in which the LP will contribute liquidity and subsequently earn fees from swap transactions. 
The LP then calculates the amount of token A and token B needed to contribute one unit of liquidity to each tick in the exchange rate range using \Cref{eq:L2}. 
There is a linear relationship between $A_i$, $B_i$, and $\ell_i$, so
\begin{equation}
\label{eq:deltaAB}
    \forall i \in \{j^1, \ldots, j^2 - 1\}: \begin{cases}
    \Delta A_i = 0, \quad  \Delta B_i = \sqrt{p_{i+1}} - \sqrt{p_i}, & i < i^*,\\
    \Delta A_i = 1/\sqrt{p^*} - 1/\sqrt{p_{i+1}}, \quad  \Delta B_i = \sqrt{p^*} - \sqrt{p_i}, & i = i^*,\\
    \Delta A_i =  1/\sqrt{p_i} - 1/\sqrt{p_{i+1}}, \quad \Delta B_i = 0, & i > i^*,
    \end{cases}.
\end{equation}
The LP uses the current market exchange rate $m^*$ to calculate the current \textit{market value} of those tokens in units of the base token B, given by
\begin{equation*}
    \sum_{i = j^1}^{j^2 - 1} \Delta B_i + m^* \left(\sum_{i = j^1}^{j^2 - 1} \Delta A_i\right) =  \sqrt{p^*} - \sqrt{p_{j^1}} + m^* \left(1/\sqrt{p^*} - 1/\sqrt{p_{j^2}}\right).
\end{equation*}
Because liquidity and token reserves are linearly related, the LP can then scale its token contributions to contribute a total market value of $k$ tokens of B. In this case, the LP will add $u(j^1, j^2, k, m^*)$ units of liquidity to each tick between $j^1 \in [d]$ and $j^2 \in [d]$, where
\begin{equation*}
    u(j^1, j^2, k, m^*) := \dfrac{k}{\sqrt{p^*} - \sqrt{p_{j^1}} + m^* \left(1/\sqrt{p^*} - 1/\sqrt{p_{j^2}}\right)}.
\end{equation*}
The LP contributes liquidity with the intention of earning fees. If it contributes capital valued at $k$ tokens of B at the current market exchange rate $m^*$, and it contributes this capital between ticks $j^1$ and $j^2$, the size of its share of fees in tick $i$ is given by
\begin{equation}
\label{eq:share}
    \frac{u(j^1, j^2, k, m^*)}{
    \bell_i + u(j^1, j^2, k, m^*)} \one_{\{j^1 \leq i < j^2\}}, \qquad \forall i \in [d].
\end{equation}

\subsubsection*{Example: Toy Liquidity Addition}
Consider a liquidity distribution like the one shown in \Cref{fig:exliq1}, which has corresponding token reserves shown in \Cref{fig:extok1}. In this rudimentary example, we have $d = 6$, $p^* = m^* = 1.6$, and $\bm p = [1, 1.21, 1.44, 1.56, 1.69, 1.96, 2.25, 2.56]$, where $\bm p$ was chosen so that $\sqrt{p_i}$ are simple. Note that in this case, $i^* = 3$. The exact values of $\bell$, $\bm A$, and $\bm B$ are shown in \Cref{tab:liqcalc1}. 
\begin{figure}[H]
    \centering
    \begin{subfigure}{0.45\textwidth}
         \centering
         \includegraphics[width = \textwidth]{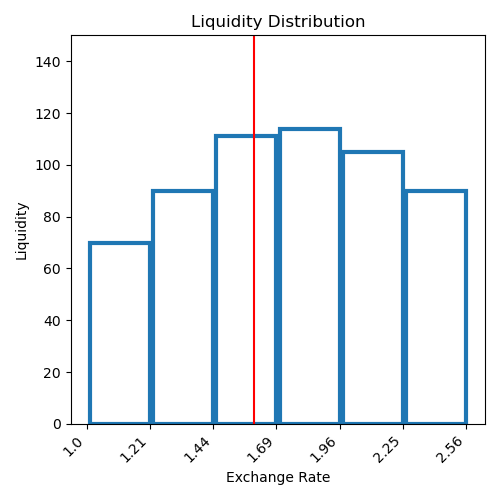}
    \caption{The original liquidity distribution of the pool.
    \label{fig:exliq1}}
     \end{subfigure}
    \begin{subfigure}{0.45\textwidth}
         \centering
         \includegraphics[width = \textwidth]{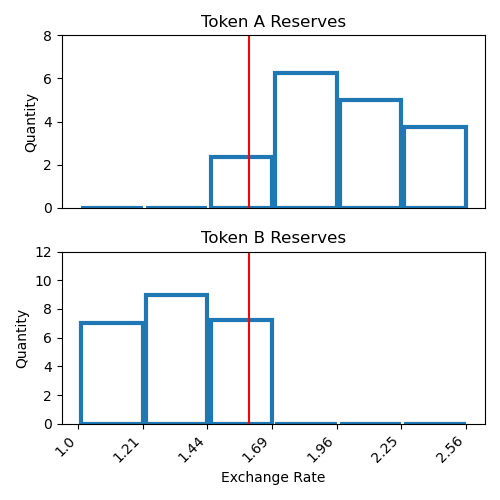}
    \caption{The original token reserves in the pool.
    \label{fig:extok1}}
     \end{subfigure}
     \caption{Toy liquidity pool}
\end{figure}
\begin{table}[H]
    \centering
    \footnotesize
    \begin{tabular}{|c|c|c|c|c|c|c|}
    \hline
        Tick $i$ & 1 & 2 & 3 & 4 & 5 & 6 \\
        \hline
        $[p_i, p_{i+1})$ & [1, 1.21) & [1.21, 1.44) & [1.44, 1.69) & [1.69, 1.96) & [1.96, 2.25) & [2.25, 2.56) \\
        \hline
        $A_i$ & 0 & 0 & 2.37 & 6.25 & 5 & 3.75  \\
        \hline
        $B_i$ & 7 & 9 & 7.208 & 0 & 0 & 0 \\
        \hline
        $\ell^0_i$ & 70 & 90 & 111.052 & 113.75 & 105 & 90 \\
        \hline
    \end{tabular}
    \caption{The original state of the liquidity pool before a new LP arrives to contribute liquidity.}
    \label{tab:liqcalc1}
\end{table}
\begin{remark}
    Note that the token reserves of B are equal in ticks 1 and 2, but the token reserves of A are not equal in ticks 4, 5, and 6. This is because, given our choice of $\bm p$, $\sqrt{p_{i+1}} - \sqrt{p_i} = 0.1$ for all ticks, but $1/\sqrt{p_i} - 1/\sqrt{p_{i+1}}$ is different for every tick. 
\end{remark}
Now, a new LP has arrived and wishes to contribute to this pool. It must choose a range of exchange rates (or ticks) in which to contribute liquidity as well as an amount of capital (in units of token B) to contribute to the pool. Let us say that the LP wants to contribute liquidity to the exchange rate range $[1.21, 1.96)$ which corresponds to ticks 2, 3, and 4. The LP has a total amount of capital valued at 2.87 tokens of B. If it wishes to determine the numbers of tokens needed to achieve this contribution, as well as the amount of liquidity that will be added according to this contribution, it will use \Cref{eq:deltaAB} to calculate how many tokens are required to contribute one unit of liquidity uniformly across the selected ticks. The calculations are shown below:
\begin{align*}
    i = 2: & \begin{cases}
        \Delta A_2 = 0,\\
        \Delta B_2 = \sqrt{1.44} - \sqrt{1.21} = 0.1.
    \end{cases}\\
    i = 3: & \begin{cases}
        \Delta A_2 = 1/\sqrt{1.6} - 1/\sqrt{1.69} =0.021339, \\
        \Delta B_2 = \sqrt{1.6} - \sqrt{1.44} = 0.064911.
    \end{cases}\\
    i = 4: & \begin{cases}
        \Delta A_2 = 1/\sqrt{1.69} - 1/\sqrt{1.96} = 0.054945,\\
        \Delta B_2 = 0.
    \end{cases}
\end{align*}
The total required token reserves to add one unit of liquidity across the chosen ticks is
\begin{equation*}
    \sum_{i=2}^4 \Delta A_i = 1/\sqrt{1.6} - 1/\sqrt{1.96} = 0.076284, \qquad \sum_{i=2}^4 \Delta B_i = \sqrt{1.6} - \sqrt{1.21} = 0.164911.
\end{equation*}
Therefore, the capital in terms of token B required to contribute one unit of liquidity per tick across ticks 2, 3, and 4, where the conversion rate from A to B is $m^* = 1.6$, is given by
\begin{equation*}
    0.164911 + m^* \times 0.076284 = 0.28700 \text{ tokens of B.}
\end{equation*}
Since the LP has 2.87 tokens of B, it is able to contribute $2.87/0.287 = 10$ units of liquidity to every tick, and its final token contributions will be 0.76284 tokens of A and 1.64911 tokens of B. In this case, $u(2, 4, 2.87, 1.6) = 10$ and so $\bell^1 = \bell_{\text{new}}(10, 2, 4) + \bell^0$. The resultant liquidity distribution is shown in \Cref{fig:exliq2} and the resultant token reserves are shown in \Cref{fig:extok2}. The new values of $\bell$, $\bm A$, and $\bm B$ are shown in \Cref{tab:liqcalc2}.  
\begin{figure}[H]
    \centering
    \begin{subfigure}{0.45\textwidth}
         \centering
         \includegraphics[width = \textwidth]{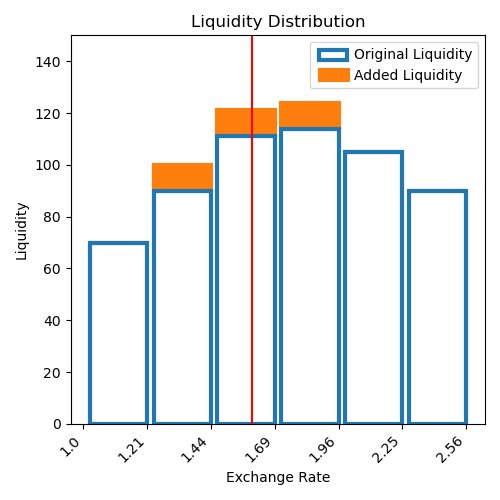}
    \caption{The updated liquidity distribution of the pool.
    \label{fig:exliq2}}
     \end{subfigure}
    \begin{subfigure}{0.45\textwidth}
         \centering
         \includegraphics[width = \textwidth]{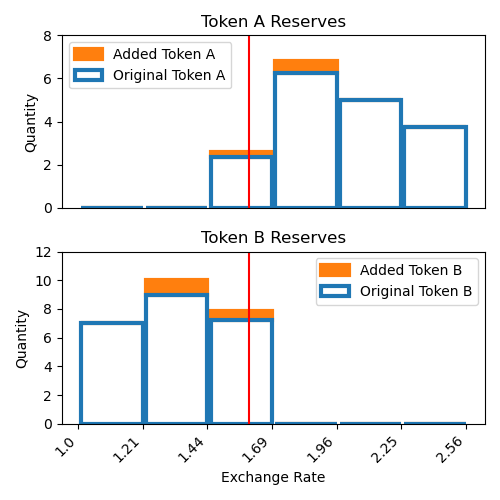}
    \caption{The updated token reserves in the pool.
    \label{fig:extok2}}
     \end{subfigure}
     \caption{Example liquidity addition}
\end{figure}

\begin{table}[H]
    \centering
    \footnotesize
    \begin{tabular}{|c|c|c|c|c|c|c|}
    \hline
        Tick $i$ & 1 & 2 & 3 & 4 & 5 & 6\\
        \hline
        $[p_i, p_{i+1})$ & [1, 1.21) & [1.21, 1.44) & [1.44, 1.69) & [1.69, 1.96) & [1.96, 2.25) & [2.25, 2.56) \\
        \hline
        $A_i$ & 0 & 0 & \cellcolor{gray!25}2.583& \cellcolor{gray!25}6.8 & 5 &  3.75  \\
        \hline
        $B_i$ & 7 &  \cellcolor{gray!25}10 & \cellcolor{gray!25}7.858 & 0 & 0 & 0 \\
        \hline
        $\ell^1_i$ & 70& \cellcolor{gray!25}100& \cellcolor{gray!25}121.052& \cellcolor{gray!25}123.75& 105& 90 \\
        \hline
    \end{tabular}
    \caption{The state of the liquidity pool after an LP arrives to contribute 2.87 tokens of B distributed across $[1.21, 1.96)$.}
    \label{tab:liqcalc2}
\end{table}
We also note that our LP of interest has contributed 10 units of liquidity to every tick, but has contributed 9.09\% of the liquidity in tick 2, 8.26\% of the liquidity in tick 3, and 8.08\% of the liquidity in tick 4. This \textit{share} of liquidity (not the absolute liquidity contribution) will determine the amount of the fees earned by the LP for swaps with execute at exchange rates within $[1.21, 1.96)$. The LP will earn nothing for swaps or parts of swaps which are executed in the other three ticks.

\subsubsection*{Example: Liquidity Addition in Uniswap Interface}
In \Cref{fig:interface}, we show annotated screenshots of the Uniswap interface for creating a new liquidity position in a ETH/USDC pool with a $\gamma=$0.30\% transaction fee. 
Note that the Uniswap interface prompts users to choose their contribution range and their amount for one of the two tokens, and calculates the remaining token amount to ensure equal liquidity contributions. 
However, as we hae discussed, it is straightforward to calculate the amount of one token that will ensure the correct overall contribution.
In \Cref{fig:interface}, we assume that the LP wishes to contribute tokens with a value of 5000 USDC (token B). Notice that when the selected exchange rate range for liquidity changes, the number of tokens of ETH (token A) needed relative to the number of tokens of USDC will also change, as will the amounts needed of both tokens. If a position is centered around the current pool exchange rate, then the value of the ETH and USDC contributions is roughly equal, as seen in \Cref{fig:screenshot2}. However, a position that skews towards lower exchange rates (meaning that ETH is relatively cheaper) requires more USDC than ETH, while a position that skews towards higher exchange rates (meaning ETH is more expensive) requires less USDC and more ETH. This can be seen in \Cref{fig:screenshot1} and in \Cref{fig:screenshot3} respectively.
\begin{figure}[H]
    \centering
     \begin{subfigure}{0.3\textwidth}
         \centering
         \includegraphics[width=\textwidth]{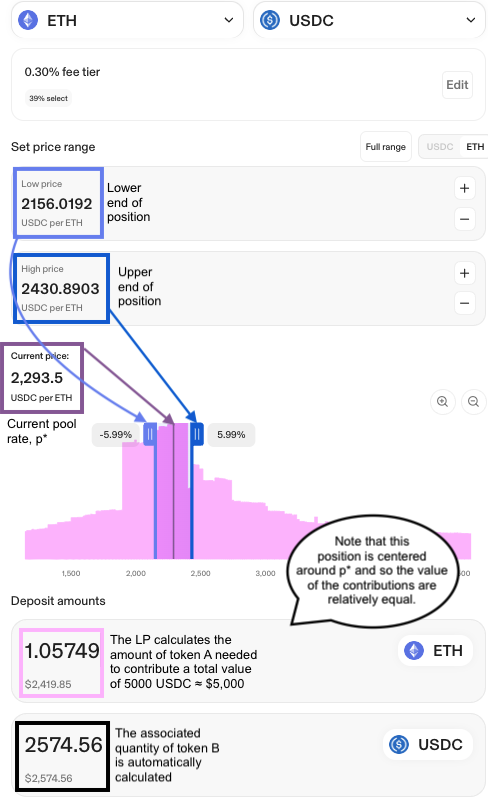}
         \caption{\label{fig:screenshot2}}
     \end{subfigure}
    \begin{subfigure}{0.3\textwidth}
         \centering
         \includegraphics[width=\textwidth]{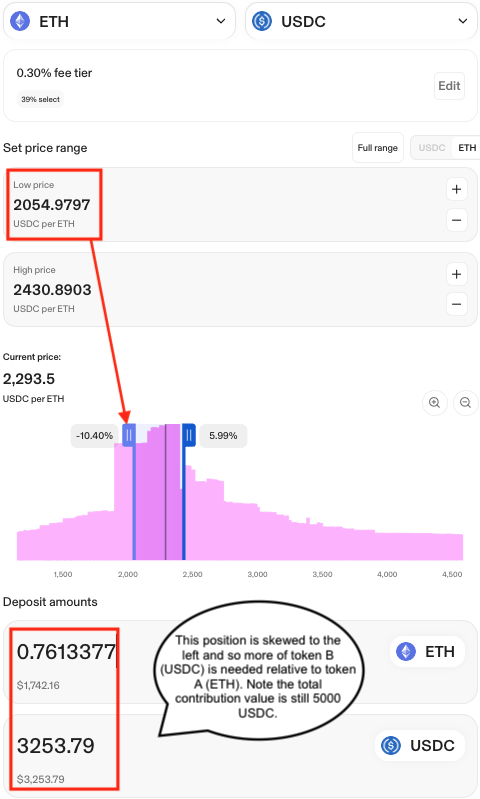}
         \caption{\label{fig:screenshot1}}
     \end{subfigure}
     \begin{subfigure}{0.3\textwidth}
         \centering
         \includegraphics[width=\textwidth]{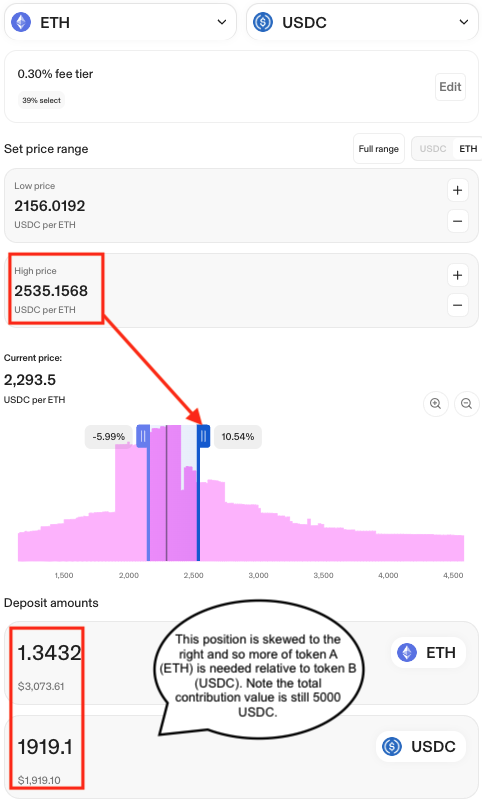}
         \caption{\label{fig:screenshot3}}
     \end{subfigure}
    \caption{Screenshots from the 0.30\% fee Uniswap ETH/USDC pool liquidity interface where the LP wishes to make a contribution valued at 5000 USDC.}
    \label{fig:interface}
\end{figure}
To illustrate the consistency of our equations and the Uniswap interface, note that if $m^* = p^* = 2293.5$ (which was the case at the time of calculation), then we can calculate
\begin{align*}
   & \Delta \ell = 5000/(\sqrt{2293.5} - \sqrt{2156.0192} + m^*(1/\sqrt{2293.5} - 1/\sqrt{2430.8903})), \\
    \implies &\sqrt{2293.5} - \sqrt{2156.0192}*\Delta \ell = 2574.56 \text{ tokens of B}, \\
    &1/\sqrt{2293.5} - 1/\sqrt{2430.8903}*\Delta \ell = 1.05749 \text{ tokens of A},
\end{align*}
which is consistent with \Cref{fig:screenshot2}. Note that if we wished to create this position, once we had chosen 2156.0192 and 2430.8903 as the bounds of the position, we would only need to enter 1.05748 in the pink box OR 2574.56 in the black box. The Uniswap interface would automatically calculate the missing value.
\begin{remark}
  In this example, since our base token is USDC and 1 USDC = 1 US dollar, converting the value of all holdings to USDC is equivalent to converting the value of all holdings to US dollars, which Uniswap does automatically. However, in the case of a pool where neither coin is tied to the value of the US dollar, LPs could choose to value their holdings in terms of one of the coins or in terms of USD. The choice of valuation does not affect the model or our findings.  
\end{remark}

\subsection{Executing Transactions with the Constant Product Function}
\label{sec:swap_explanation}
Liquidity contributions are not the only transactions that take place in a pool. Indeed, the reason why LPs are willing to contribute liquidity is because they hope to earn transaction fees from swaps executed in the pool. The exchange rate received by swappers who wish to exchange tokens is also governed by the CPMM \Cref{eq:cpmm} and depends on the liquidity in the pool. For ease, we have chosen to make the exchange function $\psi$ a function of token B, with output in terms of token A. We will define $\psi: \R \times \R^d_+ \times \R \to \R$ as the exchange function, where $\psi(x, \bell, p^*)$ designates the amount of token A received by a swapper who deposits $x$ of token B into a pool with liquidity distribution $\bell$ and current pool exchange rate $p^*$. The new reserves in tick $i^*$ after a swap of $x$ tokens of B for $y := \psi(x, \bell, p^*)$ tokens of A must satisfy
\begin{equation}
\label{eq:cpmm2}
    \Big(B_{i^*} + \ell_{i^*}\sqrt{p_{i^*}} + x\Big)\Big(A_{i^*} + \ell_{i^*}/\sqrt{p_{i^*+1}} - y\Big) = \ell_{i^*}^2,
\end{equation}
where $x \geq -B_{i^*}$ and $y \leq A_{i^*}$. Note that the reserves of tokens A and B in a tick will fluctuate as users exchange tokens, but $\ell_i$ will remain constant until another user adds or removes liquidity from tick $i$. Swaps will alter the pool exchange rate $p^*$, such that the new pool exchange rate $p^*_{\text{new}}$ after the swap is determined using
\begin{equation}
\label{eq:p_new}
    p^*_{\text{new}}(x, \bell, p^*) := \dfrac{B_{i^*} + \ell_{i^*}\sqrt{p_{i^*}} + x}{A_{i^*} +
    \ell_{i^*}/\sqrt{p_{i^*+1}} - y}.
\end{equation}
\begin{remark}
    We emphasize the following point - when liquidity is deposited or removed by an LP, $p^*$ remains constant but $\bell$ changes. When a swapper makes an exchange, $p^*$ changes but $\bell$ remains constant. Each can be used to determine the other, but at any point, at least one of them is fixed by the design of the pool.
\end{remark}

Intuitively, it is clear that a swapper cannot withdraw more than $A_i$ tokens of A or $B_i$ tokens of B from tick $i$. From \cref{eq:cpmm2} one can also observe that there is a limit on how much of one token can be deposited into a tick before the real reserves of the other token become negative. Then, what happens if a user wishes to deposit or withdraw more of a token than what can be accommodated by the current interval? In this scenario, the transaction must calculate an exchange rate across multiple ticks, based on the full vector $\bell$. Therefore, we can define $\psi$ as a piecewise function where 
    \begin{equation*}
        \psi(x, \bell, p^*) := \sum_{i=1}^d \psi^i(x, \bell, p^*).
    \end{equation*}
We will define the functions $\psi^i$ which correspond to the ticks $i$ momentarily.
As we will show, $\psi$ is monotonically increasing with $\psi(0) = 0$. This means that if $x>0$ (and $\psi(x, \bell, p^*)>0$), then the user exchanges token B for token A. If $x < 0$ (and $\psi(x, \bell, p^*) < 0$), then the user exchanges token A for token B. 

The domain of each $\psi^i$ is determined by $\{B_i(\ell_i;\  p^*)\}_{i=1}^d$ and $\{\ell_i\}_{i=1}^d$ such the domain of $\psi$ is split into $d$ regions with $d+1$ endpoints, defined as $\bm \beta(\bell; i^*) = \{\beta_i(\bell; i^*)\}_{i=1}^{d+1}$. These regions correspond to the exchange rate ticks $\{p_i\}_{i=1}^{d+1}$ in the pool. If the current exchange rate is $p^*$ contained in the current tick $i^*$, we calculate $\beta_i(\bell; i^*)$ as
\begin{equation}
    \label{eq:b}
    \beta_i(\bell; i^*) := \begin{cases}
        -\sum_{j = i+1}^{i^*-1} B_j, & 1 \leq i \leq i^* - 1,\\
        \sum_{j = i^*}^{i} \ell_j(\sqrt{p_{j+1}} - \sqrt{p_j}), & i^* \leq i \leq d.
    \end{cases}
\end{equation}
We derive the deposit limits $\ell_j(\sqrt{p_{j+1}} - \sqrt{p_j})$ from \Cref{eq:cpmm2} and the withdrawal limits $-B_j$ from the requirement that we cannot withdraw more of token B from tick $j$ than what the tick contains. 

Now that we have defined the quantities $\beta_i(\bell; i^*)$, we will at times suppress the dependence on $\bell$ and $p^*$ (as we have done for $A_i$ and $B_i$) for ease of reading. Again using \Cref{eq:cpmm2}, we can now define each of the piecewise function components $\psi^i$ as
\begin{equation}
\label{eq:psi_i}
    \psi^i(x, \bell, p^*) := \begin{cases}
    -\ell_i(1/\sqrt{p_i} - 1/\sqrt{p_{i+1}}), &\ x \leq  \beta_i < 0,\\
     A_i + \ell_i/\sqrt{p_{i+1}} - \dfrac{(\ell_i)^2}{B_i + (x - (\beta_{i+1} \wedge 0)) + \ell_i\sqrt{p_i}}, & \ x \in (\beta_i, \beta_{i+1}], x < 0, \\
     0, &\  \beta_{i+1} < x < 0 \text{ or }  0 < x < \beta_i, \\
     A_i + \ell_i/\sqrt{p_{i+1}} - \dfrac{(\ell_i)^2}{B_i + (x - (\beta_i \vee 0)) + \ell_i \sqrt{p_i}}, & \ 0 < x \in [\beta_i, \beta_{i+1}),\\
     A_i, & \ 0 < \beta_{i+1} \leq x.
    \end{cases}
\end{equation}
A graph of $\psi$ for a simple pool is given in \Cref{fig:psi}, where the multiple colors represent each of the piecewise functions $\psi^i$ which make up the full-pool exchange function. We summarize the various pool parameters and variables in \Cref{tab:params}.

The first-time reader may find the five different cases overwhelming, but the intuition is relatively straightforward. Token B is exchanged for token A by ``using up" the available liquidity in each tick, working outward from $i^*$ in a direction corresponding to the sign of $x$. If there is more liquidity available in a tick than there is token B left to exchange, \Cref{eq:cpmm2} determines the amount of token A which corresponds to that amount of token B. This tick is the new $i^*$. If a swap doesn't ``reach" a tick $i$, its $\psi^i$ returns 0. 

After each swap, $\bell$ remains fixed but $\bm A$ and $\bm B$ have changed to $\bm A^{\text{new}}(\xi, \bell, i^*)$ and $\bm B^{\text{new}}(\xi, \bell, i^*)$. The new pool exchange rate can be calculated via \Cref{eq:p_new} where $i^*$ now refers to the tick in which the swap ended, and which now contains $p^*_{\text{new}}$. The function domain $\bm \beta$ is also updated to $\bm \beta^{\text{new}}$ which corresponds with $\bm A^{\text{new}}$, $\bm B^{\text{new}}$, and $p^*_{\text{new}}$. 
The amount of fees earned by all LPs with active liquidity in tick $i$ when a swap of size $\xi$ is executed is given by
\begin{equation}
    \label{eq:fees}
    \phi(i, \xi; \bell, p^*) := \gamma |B^{\text{new}}_i(\xi, \bell_i; p^*) - B_i(\bell_i;  p^*)|.
\end{equation}
This function essentially determines how many tokens of B were added or removed from each tick in the pool, and calculates the fees per tick accordingly.

A full example is worked through in \Cref{sec:appendix_swap}.
\begin{figure}
    \centering
    \includegraphics[width = 0.6\textwidth]{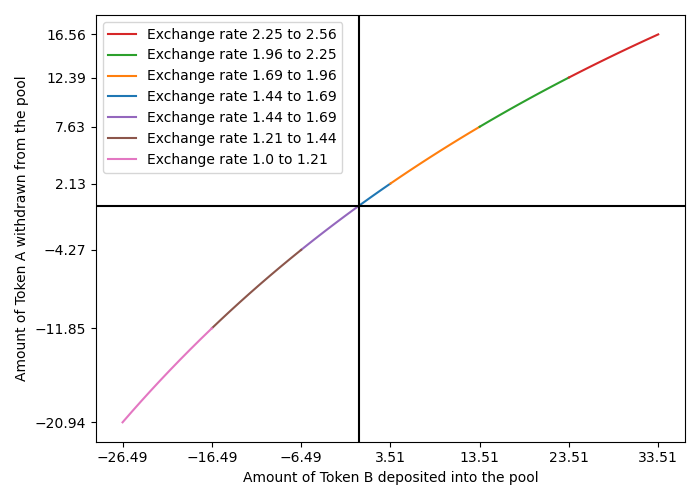}
    \caption{An example of $\psi(x, \bell, p^*)$ when $\bell = [100,100,100,100,100,100]$ and the exchange rate ticks are $[1, 1.21, 1.44, 1.69, 2.25, 2.56]$.}
    \label{fig:psi}
\end{figure}

\begin{table}
    \centering
    \begin{tabular}{|p{0.2\textwidth}|p{0.65\textwidth}|}
        \hline
        Pool Characteristics & Definition \\
        \hline
        $\gamma \in (0,1)$ & Fixed pool transaction fee. \\
        $\{p_i\}_{i=1}^{d+1}$ & The exchange rates associated with ticks in the pool, also fixed. \\
        \hline
        Pool Variables & Definition \\
        \hline
        $p^* \in \R_+$ & Pool exchange rate which evolves with swaps.\\
        $i^* \in [d]$ & Index corresponding to $p^*$. \\
        $\bell \in \R^d_+$ & Liquidity distribution which determines the CPMM for each tick. \\
        $\{A_i(\ell_i, p^*)\}_{i=1}^d \in \R^d_+$ & Token A distribution. \\
        $\{B_i(\ell_i, p^*)\}_{i=1}^d \in \R^d_+$ & Token B distribution.\\
        $\{\beta_i(\bell; i^*)\}_{i=1}^{d+1} \in \R^{d+1}_+$ & Domains of $\psi^i$ functions in terms of token B.\\
        \hline
    \end{tabular}
    \caption{Model parameters and variables, consolidated.}
    \label{tab:params}
\end{table}

\subsubsection*{Example: Toy Swap Calculation}
\label{sec:appendix_swap}
Now, we will provide two examples of the calculations for the exchange rate of a swap transaction using the CPMM algorithm. We will start with a pool containing a distribution of token A and B as displayed in table \Cref{tab:toks1}. The current exchange rate in the pool will be $p^* = 1.25$ and this corresponds to tick 2, so the ``active" tick is $i^* = 2$. For simplicity, we will set the pool transaction fee $\gamma = 0$.

\begin{table}[H]
    \centering
    \footnotesize
    \begin{tabular}{|c|c|c|c|c|c|c|}
    \hline
        Tick & 0 & 1 & $i^* = 2$ & 3 & 4 & 5 \\
        \hline
        Token A & 0 & 0 & 2.154 & 6.25 & 6.25 & 6.25  \\
        \hline
        Token B & 10 & 10 &  6.553 & 0 &  0 & 0 \\
        \hline
        $[\beta_i, \beta_{i+1})$ & [-26.553, -16.553) & [-16.553, -6.553) & [-6.553, 3.542) & [3.542, 14.917) & [14.917, 28.042) & [28.042, 43.042) \\
        \hline
        $\ell_i$ & 100 & 100 &  100.956 & 113.75 &  131.25 & 150 \\
        \hline
        $[p_i, p_{i+1})$ & [1, 1.21) & [1.21, 1.44) & [1.44, 1.69) & [1.69, 1.96) & [1.96, 2.25) & [2.25, 2.56) \\
        \hline
    \end{tabular}
    \caption{The initial setup of a liquidity pool with $p^* = 1.5$ B/A and $i^* = 2$.}
    \label{tab:toks1}
\end{table}

In the first example, we will demonstrate a scenario where the swap does not ``cross ticks" - that is, the composition of the liquidity in the pool changes and the exchange rate $p^*$ changes, but $i^*$ does not change. In \Cref{fig:demo_swap1}, we provide an intuitive illustration of what is happening in the pool after a swap where $\beta_{i^*} \leq \xi \leq \beta_{i^* +1}$. 
\begin{figure}[H]
    \centering
    \begin{subfigure}{0.45\textwidth}
         \centering
    \includegraphics[width=\textwidth]{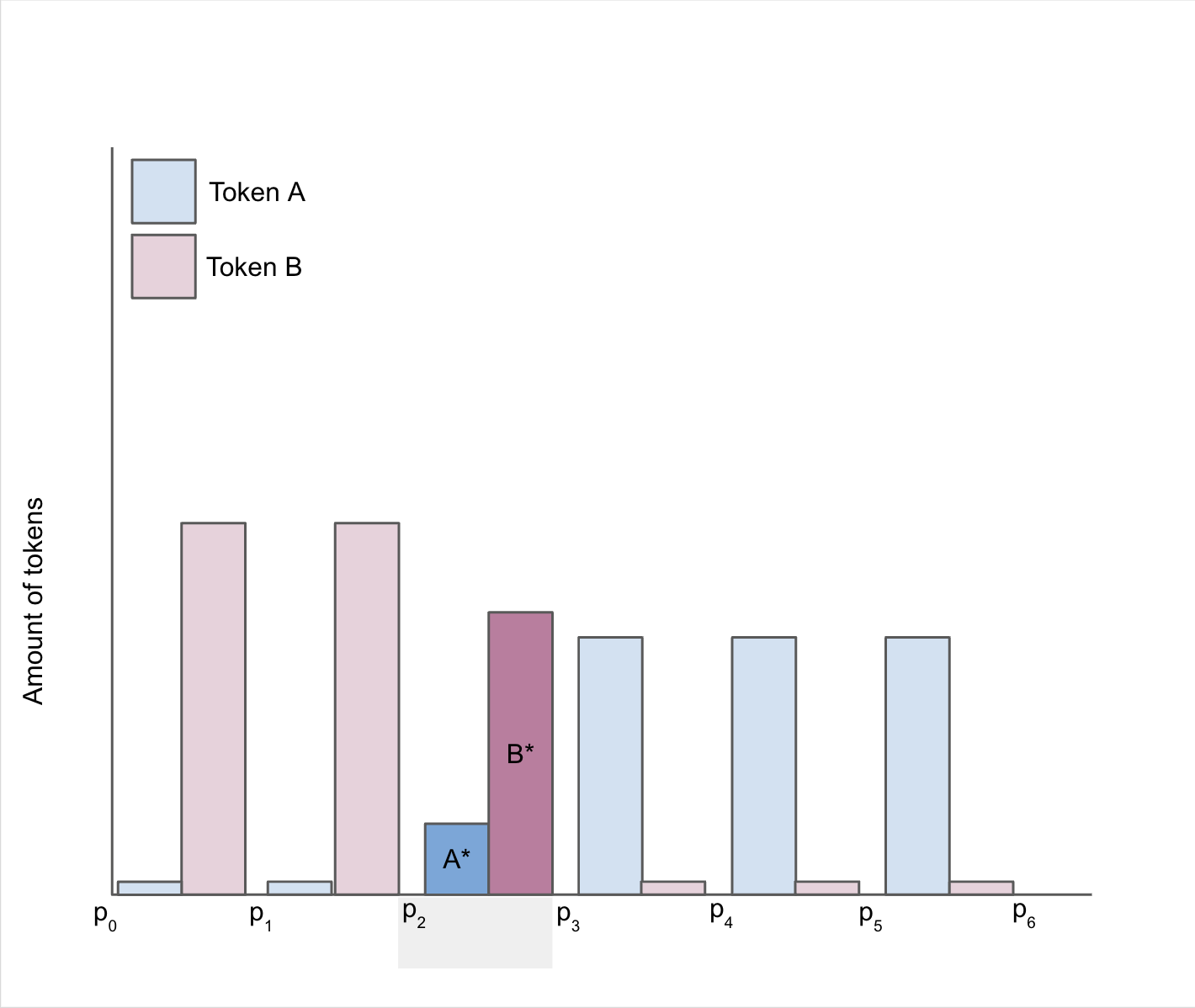}
    \caption{An example initial liquidity distribution with $i^* = 2$.}
     \end{subfigure}
     \begin{subfigure}{0.45\textwidth}
         \centering
         \includegraphics[width=\textwidth]{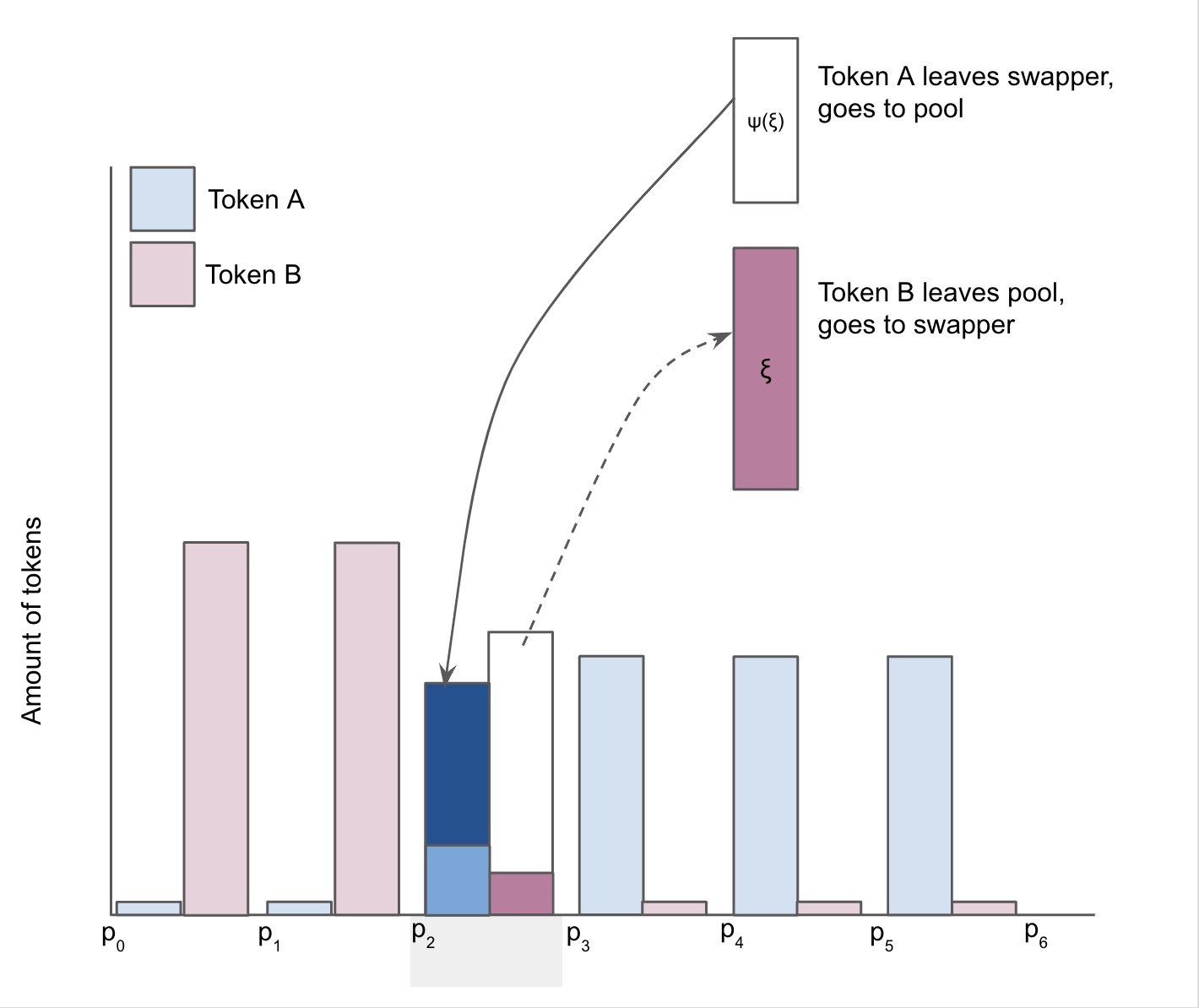}
         \caption{The result after a swap of $-B^* < \xi < 0$, $i^* = 2$ still.}
     \end{subfigure}
    \caption{We illustrate an initial liquidity distribution (left) and the result of a swap where $B^* < \xi < 0$ (right), meaning that token B is withdrawn and token A is deposited.}
    \label{fig:demo_swap1}
\end{figure}

For a more concrete example, say that we choose $\xi = -5$ (the swapper wishes to withdraw 5 tokens of B from the pool). Looking at \Cref{tab:toks1}, we can see that $\beta_2 = -6.553 \leq -5 \leq \beta_3 = 3.542$, so the tokens in tick 2 are enough to facilitate this trade. To determine how many tokens of A must the swapper deposit, we must use \Cref{eq:cpmm2}, giving us
\begin{equation*}
    (6.553 + 100.956\sqrt{1.44} - 5)(2.154 + 100.956/\sqrt{1.69} + y) = (100.956)^2,
\end{equation*}
where we solve for $y$ and find that $y = 3.252$. The new liquidity distribution looks like \Cref{tab:toks2}, where only the Token A, Token B, and $\beta_i$ rows have changed. Tokens A and B have only changed for Tick 2, and $\beta_i$ have all shifted up by $-\xi = 5$. If we wish to determine the new exchange rate from the updated liquidity, we can use \Cref{eq:p_new} to calculate
\begin{equation*}
    \dfrac{1.553 + 100.956\sqrt{1.44}}{5.406+ 100.956/\sqrt{1.69}} = 1.477.
\end{equation*}
Notice that in \Cref{tab:toks2} the $\ell_i$ have not been changed by the swap. In addition, if $\gamma > 0$, only LPs who include tick 2 in their liquidity position will earn fees from this swap.

\begin{table}[H]
    \centering
    \footnotesize
    \begin{tabular}{|c|c|c|c|c|c|c|}
    \hline
        Tick $i$ & 0 & 1 & $i^* = 2$ & 3 & 4 & 5 \\
        \hline
        Token A & 0 & 0 & \cellcolor{gray!25}5.406 & 6.25 & 6.25 & 6.25  \\
        \hline
        Token B & 10 & 10 &  \cellcolor{gray!25}1.553 & 0 &  0 & 0 \\
        \hline
        $[\beta_i, \beta_{i+1}) $& \cellcolor{gray!25}[-21.553, -11.553) & \cellcolor{gray!25}[-11.553,  -1.553) & \cellcolor{gray!25}[ -1.553, 8.542) & \cellcolor{gray!25}[8.542, 19.917) & \cellcolor{gray!25}[19.917, 33.042) & \cellcolor{gray!25}[33.0422, 48.042) \\
        \hline
        $\ell_i$ & 100 & 100 &  100.956 & 113.75 &  131.25 & 150 \\
        \hline
        $[p_i, p_{i+1})$ & [1, 1.21) & [1.21, 1.44) & [1.44, 1.69) & [1.69, 1.96) & [1.96, 2.25) & [2.25, 2.56) \\
        \hline
    \end{tabular}
    \caption{The result of a swap $\xi = -5$. Now, $p^*_{\text{new}} = 1.477 < p^* = 1.5$ and $i^* = 2$ still.}
    \label{tab:toks2}
\end{table}

In our next example, we provide a calculation for a swap transaction in which a user wishes to swap more currency than is contained in the current active tick $i^*$, so liquidity is drawn from multiple ticks. To provide some intuition for the case of a swap which crosses multiple ticks, we have illustrated a generic scenario in \Cref{fig:demo_swap2}.

\begin{figure}[H]
    \centering
    \begin{subfigure}{0.45\textwidth}
         \centering
    \includegraphics[width=\textwidth]{Figures/10-18-23-swap0.png}
    \caption{An example initial liquidity distribution with $i^* = 2$.}
     \end{subfigure}
     \begin{subfigure}{0.45\textwidth}
         \centering
         \includegraphics[width=\textwidth]{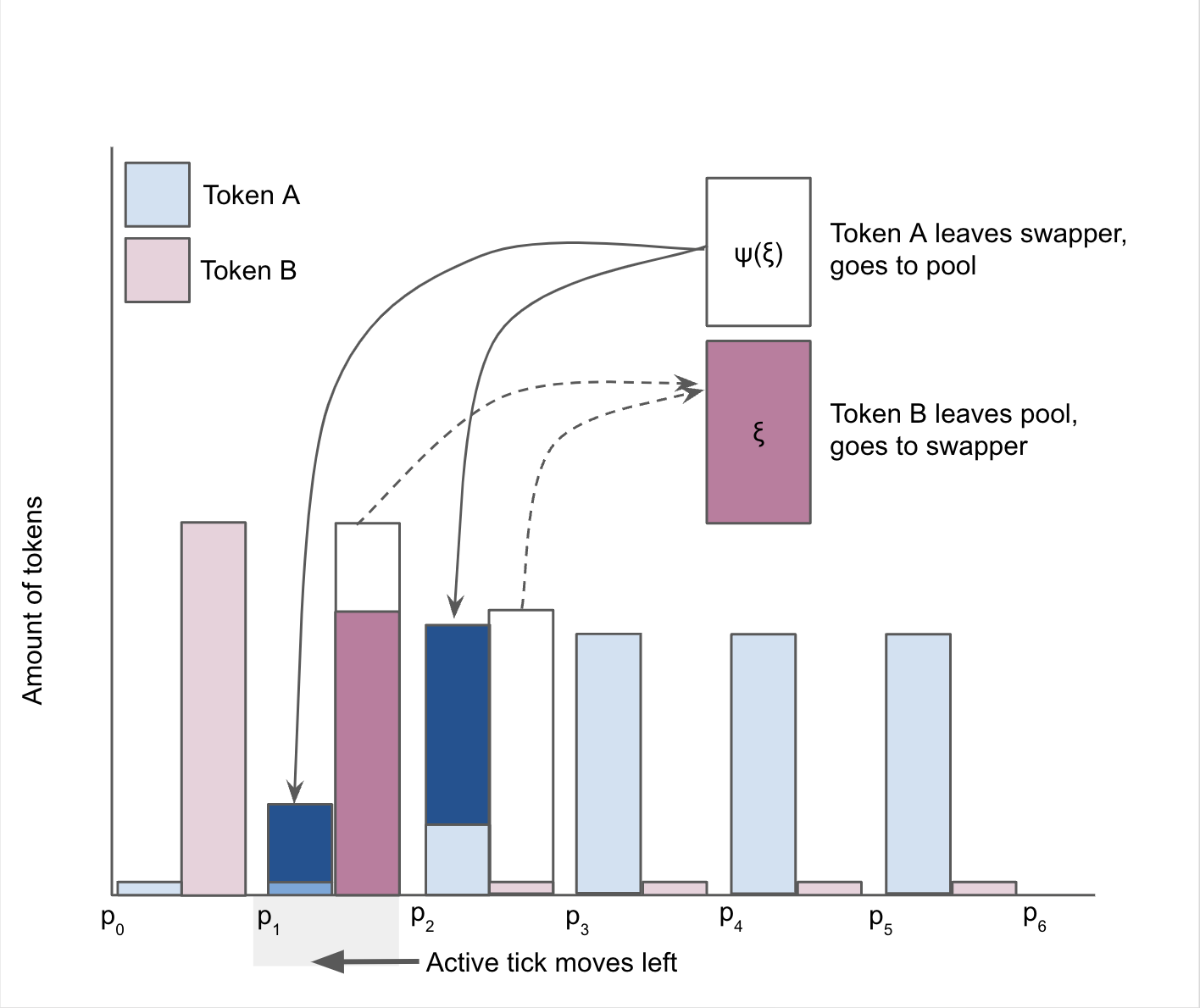}
         \caption{The result after a swap of $\xi < -B^*$, now $i^* = 1$.}
     \end{subfigure}
    \caption{We illustrate an initial liquidity distribution (left) and the result of a swap where $\xi < B^*$ (right), meaning that token B is withdrawn and token A is deposited.}
    \label{fig:demo_swap2}
\end{figure}

Now, we will provide a more thorough explanation with calculations. We again start with \Cref{tab:toks1}, but we choose $\xi = -10$. Notice that $\beta_1 = -16.553 \leq -10 \leq \beta_2 = -6.553$, so we know that we will change the token distributions in both tick 1 and tick 2, and after the swap the active tick will move to $i^* = 1$. Some calculation is required to determine how the tokens are distributed between the two ticks. We move outward tick by tick, starting from the original $i^* = 2$. Since $\beta_2 = -6.553$, we can withdraw a maximum of 6.553 tokens of B from tick 2. Again, we use \Cref{eq:cpmm2} with $i = 2$ to determine that $y_2 = 4.317$ tokens of A should be deposited in tick 2 in exchange for the 6.553 tokens of B. Now that we have drained all the tokens of B in tick 2, we move left to tick 1 (because $\xi < 0$). We still wish to withdraw an additional $(10 - 6.553) = 3.447$ tokens of B, so we use \Cref{eq:cpmm2} applied to tick 1 to determine that $y_1 = 2.464$. Therefore, the total amount deposited in the pool is $y_1 + y_2 = 6.781$ tokens of A. The new liquidity distribution is shown in \Cref{tab:toks3}. Again, use \Cref{eq:p_new} to find $p^*_{\text{new}} = 1.358$. Notice that the larger swap magnitude resulted in a larger change in the pool exchange rate. Additionally, if $\gamma > 0$ then LPs who included at least one of ticks 1 and 2 in their liquidity position will earn fees from the swap, with those that contributed to both ticks earning more fees.

\begin{table}[H]
    \centering
    \footnotesize
    \begin{tabular}{|c|c|c|c|c|c|c|}
    \hline
        Tick $i$ & 0 & $i^* = 1$ & 2 & 3 & 4 & 5 \\
        \hline
        Token A & 0 & \cellcolor{gray!25}2.464 & \cellcolor{gray!25}6.472 & 6.25 & 6.25 & 6.25  \\
        \hline
        Token B & 10 & \cellcolor{gray!25}6.553 &  \cellcolor{gray!25}0 & 0 &  0 & 0 \\
        \hline
        $[\beta_i, \beta_{i+1}) $& \cellcolor{gray!25}[-16.553, -6.553) & \cellcolor{gray!25}[-6.553,  3.447) & \cellcolor{gray!25}[3.447, 13.542) & \cellcolor{gray!25}[13.542, 24.917) & \cellcolor{gray!25}[24.917, 38.042) & \cellcolor{gray!25}[38.042, 53.042) \\
        \hline
        $\ell_i$ & 100 & 100 &  100.956 & 113.75 &  131.25 & 150 \\
        \hline
        $[p_i, p_{i+1})$ & [1, 1.21) & [1.21, 1.44) & [1.44, 1.69) & [1.69, 1.96) & [1.96, 2.25) & [2.25, 2.56) \\
        \hline
    \end{tabular}
    \caption{The result of a swap $\xi = -10$. Now, $p^*_{\text{new}} = 1.358$ and $i^* = 1$, showing more movement change in the previous example.}
    \label{tab:toks3}
\end{table}

\newpage 

\section{Numerical Methods}
\label{sec:appendix_methods}

\subsection{The Single LP Optimization}
\label{sec:method_single}
In Numerical Method \ref{alg:simple}, we solve the single LP optimization problem given by
\begin{equation*}
    V(j^1, j^2;\ \theta, \bell^0) := \E_{\bm \xi, \delta(\theta)} [\pi(j^1, j^2, \bm \xi;\ \theta, \bell^0)] - \lambda(\theta) \Var_{\bm \xi, \delta(\theta)}(\pi(j^1, j^2, \bm \xi;\ \theta, \bell^0,)).
\end{equation*}
\begin{algorithm}
\caption{Single LP Optimization}
\label{alg:simple}
\begin{algorithmic}[1]
    \State{Take $T$, $\theta$, $\gamma$, $\bell^0$, $p^*_0$, and $m^*_0$ as fixed and known.}
    \State{Initialize the running maximum as \texttt{rmax}$\gets-1000$ and initialize the optimal position as $j^{1, *}\gets0$ and $j^{2, *}\gets0$.}
\For{ $(j^1, j^2) \in J$}
            \State{Initialize the value function for the current choice of $(i,j)$ as \texttt{Vcurr}$\gets0$. Define an empty array $V = []$ to hold the total reward for all MC simulations.}
            \For{many Monte Carlo simulations}
                \State{Initialize the total reward for the current MC simulation as $R \gets 0$.}
                \For{$t = 0, \ldots, T$}
                \State{Simulate $\Delta W_t \sim N(0, \Delta t)$ and calculate $m^*_t = m^*_{t-1} e^{(\alpha + \delta) \Delta t + \sigma \Delta W_t}$.}
              \State{Simulate $X_t \sim Bernoulli(\rho(p^*_t - m^*_t))$.}
                \If{$X_t = 1$}
                \State{Simulate $\xi_t$ according to $f_{X|Y}(x \mid Y= p^*_t - m^*_t)$.}
                \State{$R \gets R + r(\xi_t, \bell^1(j^1, j^2, \bell^0), p^*_t, \theta)$.}
                \EndIf \EndFor
                \State{$V$.append($R$)}
        \EndFor
            \State{Calculate \texttt{Vcurr}$= \E[V] - \lambda Var(V)$.}
            \If{\texttt{Vcurr}$\geq$ \texttt{rmax}}
                \State{\texttt{rmax}$ \gets$ \texttt{Vcurr}, $j^{1, *}\gets j^1$ and $j^{2, *}\gets j^2$.}
            \EndIf
     \EndFor
     \State{The optimal control is $(j^{1, *}, j^{2, *})$ which produces utility \texttt{rmax}.}
\end{algorithmic}
\end{algorithm}

\subsection{The N-Player Game}
\label{sec:method_np}
In the $N$-player Bayesian game, all players have the same type-dependent strategy in equilibrium, because all players have symmetric beliefs before types are assigned. Therefore, we want to reverse-engineer $\Theta$ so that each player's equilibrium best response strategy creates a distribution of actions over $J$ which resembles $\bell^0$ in shape. We will therefore begin by assuming that each player's actions corresponding to their strategy are distributed proportionally to $\bell^0$. For each type $\theta \in \Theta$, we will determine the best response strategy of a generic player $n$ to this action distribution. To determine the strategy which maximizes the player's expected payoff given the unknown actions of the other players, we will need to use Monte Carlo trials to repeatedly sample $N-1$ actions from the space. Once $(j^1_n(\theta_n), j^2_n(\theta_n))$ has been calculated for all choices of $\theta_n \in \supp(\Theta)$ based on the sampled actions, we use the estimated action distribution to determine the distribution over types which would most closely replicate $\bell^0$. The method is described in more detail in Method \ref{alg:bayesian}. Note that the method is not complex, but the calculations are expensive. Even though the players are symmetric and so finding one player's strategy is equivalent to finding all their strategies, the computational time grows with the number of players. Notably, the expected payoff of player $n$ must be calculated with respect to the $(N-1)$-dimensional random variable representing the actions of the other players. As $N$ increases, this random vector becomes more and more difficult to handle, even when using Monte Carlo simulations. This is because as $N$ increases, more and more Monte Carlo simulations are required to accurately sample from the probability space.
\begin{algorithm}
\caption{$\Theta$ Calibration: Bayesian $N$-Player Game}
\label{alg:bayesian}
\begin{algorithmic}[1]
\Require{The observed target liquidity distribution $\bell^0$.}
\Ensure{An estimation of the type distribution $\Theta$.}
    \State{Lexicographically order the action space $J$ such that $J = \{(j^1, j^2)_q\}_{q=1}^{|J|}$. For convenience, we will sometimes refer to $(j^1, j^2)_q$ as $\mathcal{j}_q \in J$. Recall $|J| = d(d+1)/2$.}
    \State{Create a $|J| \times d$ matrix $J_{mat}$ where row $q$ is the unit liquidity position corresponding with action $\mathcal{j}_q \in J$, such that
    \begin{equation*}
        \mathcal{j}_q = (j^1, j^2)_q \to [0, \ldots, 0,\underbrace{1, \ldots\ldots, 1,}_{j^1\text{ to }j^2-1}0, \ldots, 0] \in \R^d.
    \end{equation*}}
    \State{Define $\bm x \in \R^{|J|}_+$ which represents how much liquidity each position in $J$ should hold in order to match $\bell^0$. Calculate it by applying a linear solver to
    \begin{equation*}
        \min_{\bm x \in \R^{|J|}_+} || \bm x^T J_{mat} - \bell^0||_2^2.
    \end{equation*}}
    \State{Calculate an empirical distribution $\mu \in \R^{|J|}_+$ over the action space $J$ represented by
   $ \mu  = \dfrac{\bm x}{||\bm x||_1}$,
    where $\mu_q$ represents the estimated frequency of players choosing action $\mathcal{j}_q$.}
    \State{From player $n$'s perspective, use Monte Carlo simulations to generate $S$ samples of a $(N-1)$-dimensional random vector representing the actions of the other players, given by $\{\bm j^{-n}_s\}_{s=1}^S$ where each element of $\bm j^{-n}_s$ is independently sampled from $J$ according to $\mu$.}
    \State{Discretize the space of risk-aversion levels as $\lambda \in \{\lambda_1, \ldots, \lambda_{30}\}$ where $\lambda_1 = 0$ and $\lambda_{30} = \lambda_{max}$.}
    \State{Discretize $\supp(\Theta)$ as $\Delta := \{212400, 3578600, 170603400\} \times \{\lambda_1, \ldots, \lambda_{30}\} \times \{-1, 0, 1\}$.}
    \For{all $\theta \in \Delta$}
        \State{Calculate the player's best response $(j^1_n(\theta), j^2_n(\theta)) \in J$, where
        \begin{equation*}
            (j^1_n(\theta), j^2_n(\theta))\ \in\ \argmax_{(j^1, j^2) \in J}\ \frac{1}{S} \sum_{s=1}^S V_n(j^1, j^2, \theta, \bm j^{-n}_s).
        \end{equation*}}
    \EndFor
    \State{For tie-breaking purposes, count how many types $\theta_n$ result in a best-response action $\mathcal{j}_q \in J$ by defining the counter function 
    \begin{equation*}
        \#(\mathcal{j}_q) := \sum_{\theta \in \Delta} \one_{\{(j^1_n(\theta), j^2_n(\theta)) = \mathcal{j}_q\}}
    \end{equation*}.}
    \State{Calculate the discrete probability distribution $\nu$ over $\Delta$ where, for all $\theta \in \Delta$,
    \begin{equation*}
        \nu(\theta) := \sum_{q=1}^{|J|}\dfrac{\mu_q}{\# (\mathcal{j}_q)} \one_{\{(j^1_n(\theta), j^2_n(\theta)) = \mathcal{j}_q\}} 
    \end{equation*}
    If there exist $\mathcal{j}_q$ where $\mu_q > 0$ but $\forall \theta \in \Delta, (j^1_n(\theta), j^2_n(\theta)) \neq \mathcal{j}_q$, re-scale $\nu$ s.t. $\sum_{\theta \in \Delta} \nu(\theta) = 1$.}
    \State{To satisfy \Cref{asm:theta}, we use kernel density estimation and the discrete distribution $\nu$ over $\Delta$ to calculate the density of $\lambda$ over the domain $[0, \lambda_{max}]$ for each ($k, \delta$). This produces an atomless estimate of $\Theta$.}
\end{algorithmic}
\end{algorithm}

For a given type configuration $\bm \theta$, we approximate a Nash equilibrium using the fictitious play approach which was first introduced in \cite{Brown1951}. Fictitious play has been shown to converge to an equilibrium in a variety of games, and \cite{rabinovich13} proves the convergence of a version of fictitious play applied to private-value games of incomplete information with finite actions and continuous types. 
We define a version of fictitious play which is tailored to our game. In round $I$ of fictitious play, each player $n \in \{1, \ldots, N\}$ calculates its best response to its ``memory" (the average of the liquidity distributions arising from the actions of the other players in the previous $I$ rounds), called $\bell^{\text{hist}}_n$. All players then move simultaneously and the best response actions of all players form an action profile $\bm j^I(\bm \theta)$. Player $n$'s ``memory" is updated based on $\bm j^I(\bm \theta^{-n})$. The historical liquidity distribution arising from \textit{all} players' actions, $\bell^{\text{hist}}$, is also updated based on $\bm j^I(\bm \theta)$. The process repeats until the difference between $\bell(\bm j^I(\bm \theta))$ and $\bell^{\text{hist}}$ falls below a preset threshold.
The fixed point of the iterative method corresponds to a Nash equilibrium, which we denote as $\bm j^*$. We describe our approach in detail in Numerical Method \ref{alg:FP}. 

\begin{algorithm}
\caption{Fictitious Play: $N$-Player Game}
\label{alg:FP}
\begin{algorithmic}[1]
\Require{An initial action profile $\bm j^0 \in J^N$ and a type assignment vector $\bm \theta = (\theta_1, \ldots, \theta_N) \in \supp(\Theta)^N$.}
\Ensure{An estimated mixed-strategy equilibrium given by $\bm j^*$.}
    \State{Initialize $N$ players with characteristics $\theta_n$ based on the $N$-dimensional type vector $\bm \theta$.}
    \State{Initialize an arbitrary action profile for all the players given by $\bm j^0(\bm \theta) = \{(j^1_n, j^2_n)\}_{n=1}^N$. Note $\bm \theta$ is fixed so actions are deterministic.}
    \State{Initialize the ``historical" liquidity distribution as $\bell^{\text{hist}} \gets \bell(\bm j^0(\bm \theta))$, which is updated with each iteration of fictitious play.}
    \State{For each player $n$, initialize the liquidity distribution formed by all other players as $\bell^{\text{hist}}_n \gets \bell(\bm j^0(\bm \theta^{-n}))$. We will refer to $\bell^{\text{hist}}_n$ as player $n$'s ``memory" which is also updated with each iteration of fictitious play.}
    \State{Initialize an iteration counter $I \gets 1$.}
    \State{\texttt{error} $\gets \infty$.}
    \State{Set a error tolerance threshold \texttt{thresh} which determines how closely we wish to approximate a Nash equilibrium.}
    \While{\texttt{error} $\geq$ \texttt{thresh}}
        \For{player $n \in \{1, \ldots, N\}$}
            \State{Determine player $n$'s best response to $\bell^{\text{hist}}_n$, given by 
            \begin{equation*}
                (j^{1,*}_n, j^{2,*}_n)\ \in\ \argmax_{(j^1, j^2) \in J}\ V_n(j^1, j^2, \theta_n, \bell^{\text{hist}}_n).
            \end{equation*}
            }
        \EndFor
        \State{Define the action profile for iteration $I$ as $ \bm j^I(\bm \theta) \gets \{(j^{1,*}_n, j^{2,*}_n)\}_{n=1}^N$.}
        \State{Update \texttt{error} $\gets W_1(\bell(\bm j^I(\bm \theta)),\ \bell^{\text{hist}})$ where $W_1$ is the Wasserstein-1 distance.}
        \State{Update the overall historical liquidity distribution and each player's personalized ``memory" where
            \begin{equation*}
                \bell^{\text{hist}} \gets \frac{1}{I+1} \sum_{i=0}^{I} \bell(\bm j^i(\bm \theta))\quad \text{ and } \quad \bell^{\text{hist}}_n \gets \frac{1}{I+1} \sum_{i=0}^{I} \bell(\bm j^i(\bm \theta^{-n})),\ \forall n \in \{1, \ldots, N\}.
            \end{equation*}}
        \State{Update the counter $I = I + 1$.}
    \EndWhile
    \State{The control vector $\bm j^I(\bm \theta) \in \R^d$ where $W_1(\bell(\bm j^I(\bm \theta)) <$\texttt{thresh} approximates a Nash equilibrium, and we therefore assign $\bm j^*(\bm \theta) := \bm j^I(\bm \theta)$.}
\end{algorithmic}
\end{algorithm}

We note that both Numerical Method \ref{alg:bayesian} and Numerical Method \ref{alg:FP} become more expensive as $N$ increases. In step 9 of Numerical Method \ref{alg:bayesian}, the best response for each type is computed to maximize the player's expected payoff with respect to the $(N-1)$-dimensional random vector of the other players' actions. This expectation is calculated empirically using 1000 Monte Carlo samples of $N-1$ actions when $N = 10$ and as $N$ increases the number of samples must also increase to thoroughly sample the larger probability space. Thus, it becomes more and more expensive to sample the type vectors and to compute the empirical expectation. Furthermore, as the number of players increases, step 9 of Numerical Method \ref{alg:FP} also becomes more expensive, as every player's best response is individually calculated for each iteration of the fictitious play algorithm.

\subsection{The Mean-Field Game}
\label{sec:method_mfg}
The calibration of $\Theta$ in the current MFG setting is simplified because we do not need to work with an $(N-1)$-dimensional random vector representing the actions of the other players. Instead, the representative player reacts to the target liquidity distribution $\bell^0$ directly. The distribution of types is calculated such that its type-dependent best responses to $\bell^0$ form a distribution that mimics $\bell^0$. We show the steps below in Numerical Method \ref{alg:calMFG}.

\begin{algorithm}
\caption{$\Theta$ Calibration: MFG}
\label{alg:calMFG}
\begin{algorithmic}[1]
\Require{An observed target liquidity distribution $\bell^0$.}
\Ensure{An estimation for the type distribution $\Theta$.}
    \State{Initialize $\bell^0$ as the liquidity distribution associated with \Cref{fig:mfg-a}.}
    \State{Discretize the space of risk-aversion levels as $\lambda \in \{\lambda_1, \ldots, \lambda_{30}\}$ where $\lambda_1 = 0$ and $\lambda_{30} = \lambda_{max}$.}
    \State{Define the discretization of $\supp(\Theta)$ as $\Delta := \{2124, 35786, 1706034\} \times \{\lambda_1, \ldots, \lambda_{30}\} \times \{-1, 0, 1\}$ and index $\Delta = \{\theta_i\}_{i=1}^{|\Delta|}$.}
    \State{Define an empty matrix $P \in \R^{|\Delta| \times d}$ which will hold the optimal liquidity positions associated with each type $\theta_i$ in $\Delta$.}
    \For{$\theta_i \in \Delta$, where $i=1,\ldots, |\Delta|$}
        \State{Find the representative player's best response such that
        \begin{equation*}
            (j^1(\theta_i), j^2(\theta_i))\ \in\ \argmax_{(j^1, j^2) \in J}\ V_{\text{MFG}}((j^1, j^2), \theta_i, \bell^0).
        \end{equation*}}
        \State{Set row $i$ of $P$ equal to the liquidity position associated with action $(j^1(\theta_i), j^2(\theta_i))$ and type $\theta_i$ such that 
        \begin{equation*}
            P_i \gets \bell_{\text{new}}\left(u(j^1(\theta_i), j^2(\theta_i), k(\theta_i), m^*_0), j^1(\theta_i), j^2(\theta_i)\right) \in \R^d,
        \end{equation*}
        where $\ell_{\text{new}}$ is given in \Cref{eq:ell-new}.}
    \EndFor
    \State{Define $\bm x \in \R^{|\Delta|}_+$ which represents how much liquidity each position in $P$ should hold in order to match $\bell^0$. Calculate this by solving 
    \begin{equation*}
         \min_{\bm x \in \R^{|\Delta|}_+}\ ||\bm x^T P - \bell^0||^2_2.
    \end{equation*}}
    \State{Calculate a discrete distribution $\nu$ over $\Delta$ where
    \begin{equation*}
        \nu(\theta_i) := \dfrac{\bm x_i}{|| \bm x||_1}.
    \end{equation*}}
    \State{To satisfy \Cref{asm:theta}, we use kernel density estimation and the discrete distribution $\nu$ over $\Delta$ to calculate the density of $\lambda$ over the domain $[0, \lambda_{max}]$ for each combination of $k$ and $\delta$. This produces an atomless estimate of $\Theta$.}
\end{algorithmic}
\end{algorithm}

To approximate the mean-field equilibrium once $\Theta$ was calibrated, we used the method described in Numerical Method \ref{alg:MFG}.
\begin{algorithm}
\caption{Fictitious Play: MFG}
\label{alg:MFG}
\begin{algorithmic}[1]
\Require{ An initial liquidity distribution $\bell^0$ and the type distribution $\Theta$.}
\Ensure{An equilibrium liquidity distribution $\bell^*$.}
    \State{Initialize the ``historical" liquidity distribution $\bell^{\text{hist}} := \bell^0$, the liquidity distribution associated with \Cref{fig:mfg-a}.}
    \State{Initialize an iteration counter $I \gets 0$.}
    \State{\texttt{error} $\gets \infty$.}
    \State{Set an error tolerance threshold \texttt{thresh} which determines how closely we wish to approximate a MFE.}
    \While{\texttt{error} $\geq$ \texttt{thresh}}
        \State{Calculate the representative player's best response to $\bell^{\text{hist}}$, such that for all $\theta \in \supp(\Theta)$, $\bm j^*$ satisfies
        \begin{equation*}
            \bm j^*(\theta) \in \argmax_{(j^1, j^2) \in J} V_{\text{MFG}}((j^1, j^2), \theta, \bell^{\text{hist}}).
        \end{equation*}}
        \State{Calculate the resultant liquidity distribution according to $\ell_{\text{MFG}}$ given in \Cref{eq:ell_MFG}, such that
        \begin{equation*}
            \bell^I = \bell_{\text{MFG}}(\bm j^*).
        \end{equation*}
        }
        \State{Update \texttt{error} $\gets W_1(\bell^{\text{hist}}, \bell^I)$ where $W_1$ is the Wasserstein-1 distance.}
        \State{Update $\bell^{\text{hist}} \gets \sum_{i = 0}^I \bell^i$.}
        \State{Update the counter $I$.}
    \EndWhile
    \State{Denote the final $\bell^I$ where $W_1(\bell^{\text{hist}}, \bell^I) < $ \texttt{thresh} as $\bell^*$ where its associated action strategy is $\bm j^*(\theta)$. These approximate a MFE.}
\end{algorithmic}
\end{algorithm}

\clearpage
\section{Bot threshold calculation}
\label{sec:appendix_bot}
We wish to determine the Bot's threshold value $\bar \xi(\bell)$ for which
\begin{equation*}
    \frac{L}{\ell^0_{i^*} + L} \Big(\xi -m^* \psi(\xi, \bell^1, p^*) + \gamma |\xi| \Big) - G > 0.
\end{equation*}
We first replace $\psi(\xi, \bell^1, p^*)$ with its functional form in terms of $\xi$ and $\bell^1$. We make the simplifying assumption that the original liquidity $\bell^*_{i^*}$ plus the Bot's contribution $L$ will be sufficient to accommodate a swap $\xi$ of any size, meaning that we will not consider the case where the swap ``spills over" into neighboring ticks (for a detailed explanation of what this means, refer to \Cref{sec:appendix_swap}). Therefore, we will drop the subscript $i^*$ from $A_{i^*}$, $B_{i^*}$, and $\ell^0_{i^*}$, and shorten $\ell^1_{i^*}$ to $\bell^* + L$ to improve the readability of our calculations and emphasize the dependence on $L$. We now must solve
\begin{equation*}
    \xi - m^* \left( \left(\frac{\ell^0 + L}{\ell^0}\right)A + (\ell^0 + L)/\sqrt{p_{i^* + 1}} - \dfrac{(\ell^0+L)^2}{\left(\frac{\ell^0 + L}{\ell^0}\right)B + (\ell^0+L) \sqrt{p_{i^*}} + \xi} \right) + \gamma |\xi| - \dfrac{G(\ell^0 + L)}{L} > 0,
\end{equation*}
where we note that the amounts of token A and B in tick $i^*$ increase proportionally with respect to the new liquidity $(\ell^0 + L)$. We can rearrange and simplify with the help of \Cref{eq:cpmm2} and \Cref{eq:p_new}, and finally we will obtain the quadratic inequality
\begin{align*}
    &C\xi^2 + D\xi + E > 0, 
\end{align*}
where
\begin{align*}
    & C = (1 + \gamma \sgn(\xi)), \\
    &D = \left(\frac{\ell^0 + L}{\ell^0}\right)((1 + \gamma \sgn(\xi))p^* - m^*)( A + \ell^0/\sqrt{p_{i^* + 1}})  - \dfrac{G(\ell^0 + L)}{L},\\
    &E = -(B + \ell^0 \sqrt{p_{i^*}})\dfrac{G(\ell^0 + L)^2}{\ell^0 L}.
\end{align*}
We now solve for $\xi$ to find
\begin{equation}
\label{eq:thresholds}
\begin{split}
    &\bar \xi^+(\bell) = \dfrac{-D + \sqrt{D^2 - 4(1 + \gamma)E}}{2 + 2\gamma} \vee 0, \\
    &\bar \xi^-(\bell) = \dfrac{-D - \sqrt{D^2 - 4(1 - \gamma)E}}{2 - 2\gamma} \wedge 0.
\end{split}
\end{equation}
Finally, we briefly justify that we have chosen the correct roots in each case. First, we assert that $(1 + \gamma \sgn(\xi)) > 0$ for all $\xi$ by the definition of $\gamma \in (0, 1)$ so the parabola opens upward for all values of $\xi$, $\bell$, $L$, and $G$. This implies that $\bar \xi^+(\bell) < \infty$ and $\bar \xi^-(\bell) > -\infty$. 

Given the shape of the parabola, $\bar \xi^+(\bell)$ will correspond to the rightmost real root for the parabola with $\xi > 0$ (if it exists) and $\bar \xi^-(\bell)$ will correspond to the leftmost real root for the parabola corresponding to $\xi < 0$ (if it exists). In the case where the solution for $\bar \xi^+(\bell)$ has two negative roots or the solution for $\bar \xi^-(\bell)$ has two positive roots, the threshold is set at zero instead. If the parabola has no real roots, then an attack of size $L$ is always profitable. However, this will not happen since $E \leq 0$, given that all quantities with $E$ must be positive.

\section{Parameter Calibration}
\label{sec:appendix_calib}
\subsection{Bernoulli Arrival of Swaps}
\label{sec:appendix_arrivals}
To determine how the probability of a swap arriving at block $t$ with pool exchange rate $p^*_t$ and market exchange rate $m^*_t$, we first plotted the relationship between $|p^* - m^*|$ and the frequency of at least one swap arriving and calculated a line of best fit.
From this, we were able to determine that the relationship between $|p^* - m^*|$ and the likelihood of a block containing at least one swap was approximately linear with slope 0.0026 and intercept 0.3505. However, $|p^* - m^*|$ is unbounded above and $\rho(x)$ must be contained in $[0,1]$. Therefore, we chose a sigmoid function to avoid a discontinuity in $\rho$ and scaled it such that $\rho(0) = 0.3505$ and $\frac{\partial \rho}{\partial x}(0) = 0.0026$. The resulting function is
\begin{equation*}
    \rho(x) = \dfrac{1}{1 + e^{-0.01145|x| + 0.6169}}.
\end{equation*}
Notice that as $|x| \to \infty$, $\rho(x) \to 1$, as we desired. The scatter plot of frequencies, the line of best fit, and the fitted sigmoid function are displayed in \Cref{fig:rho_calib}.
\begin{figure}[H]
    \centering
    \includegraphics[width = 0.6\textwidth]{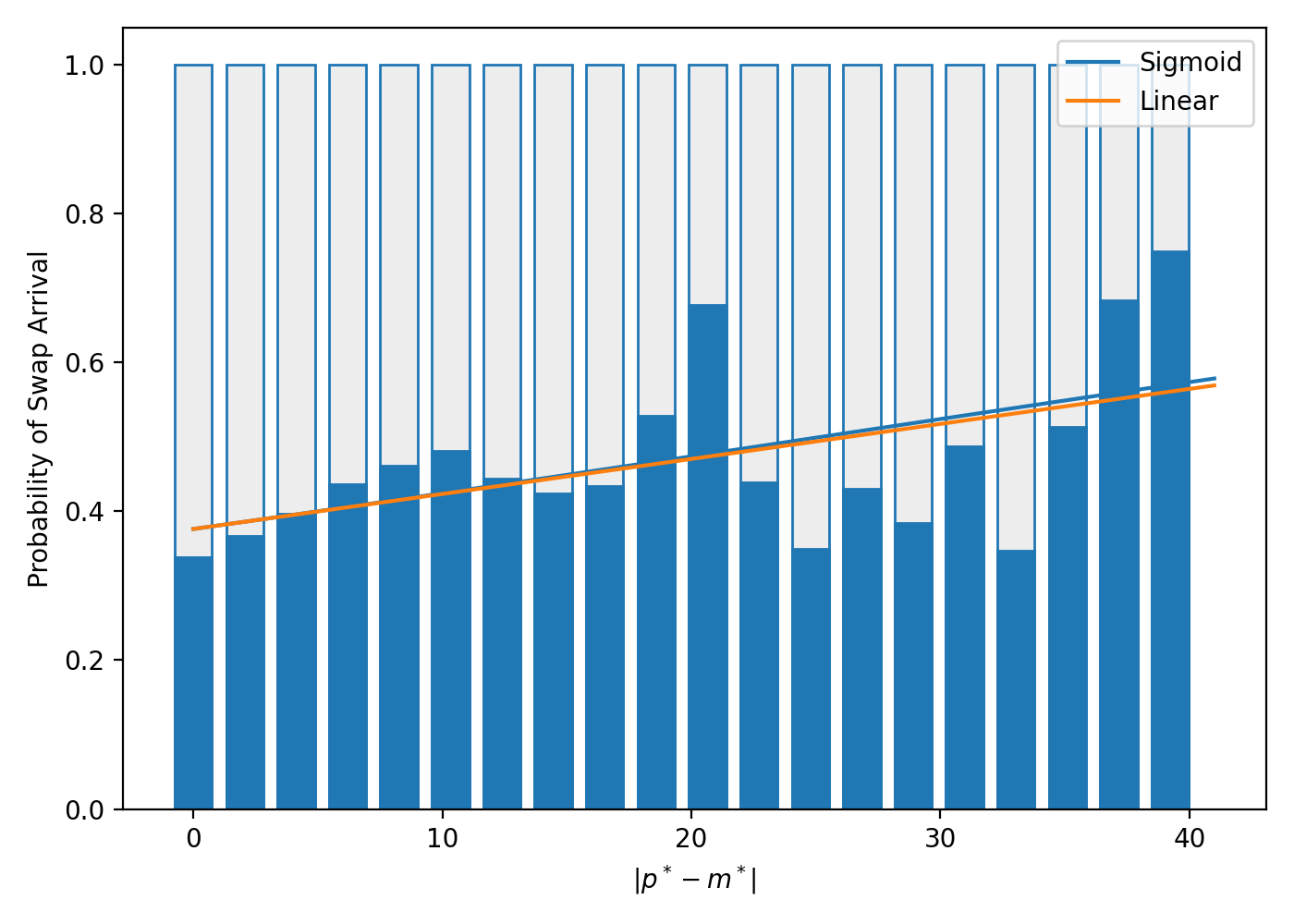}
    \caption{Comparison of scatter plot, line of best fit, and sigmoid function to determine the probability of a swap arriving given the arbitrage level in the pool.}
    \label{fig:rho_calib}
\end{figure}

\subsection{Capital Endowment}
\label{sec:appendix_calib_k}
To determine how to discretize the space of possible capital endowments, we examined the distribution of LP contributions to the pool over the period from September 1 to September 30, 2023. We looked at the Uniswap ETH/USDC pool with a 0.5\% fee, as usual. By looking at the density of the logarithm (base 10) of the value in USDC of liquidity contributions over the month, we observed three distinct regions. These liquidity additions could be roughly sorted into ``small" contributions with value less than $10^{3.78}$, ``big" contributions with value more than $10^{6.17}$, and ``medium" contributions which fell between these two extremes. Corresponding to those three regions, the highest densities were assigned to contributions with value $10^{3.33}$, $10^{4.55}$, and $10^{6.23}$. Therefore, we chose $k \in \{2124, 35786, 1706034\}$. In addition, by analyzing the data we found that 47.25\% of transactions were ``small", 48.94\% of transactions were ``medium", and 3.81\% of transactions were ``big". We did not enforce this ratio during the calibration, but we note that an accurate $\Theta$ would reflect the relative frequencies of differently sized transactions, and this is indeed the case in \Cref{tab:mfgweights}.
\begin{figure}[H]
    \centering
    \includegraphics[width = 0.6\textwidth]{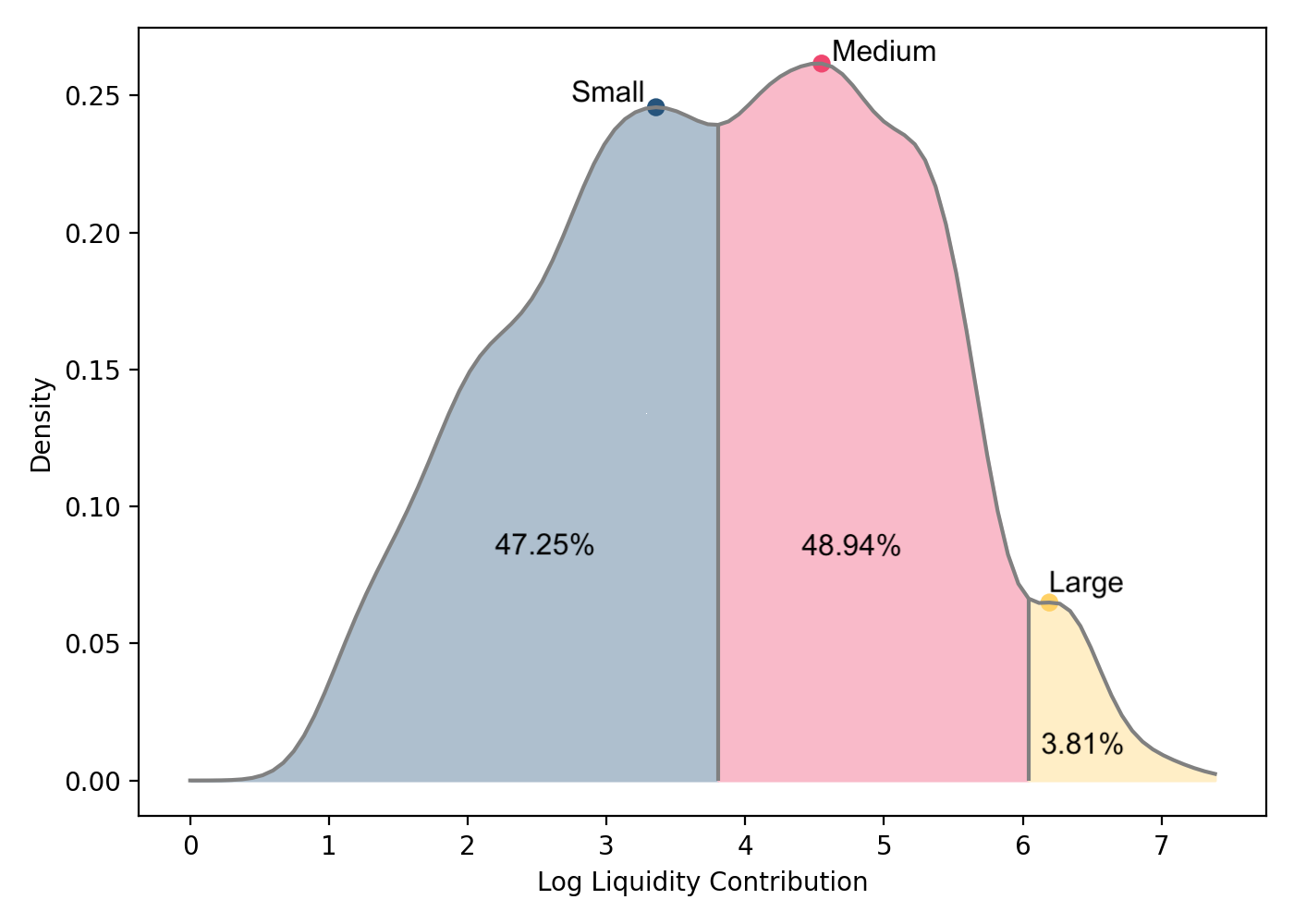}
    \caption{The cluster analysis performed on LP contributions to the ETH/USDC Uniswap pool with 0.5\% fee, between September 1 and September 30, 2023. The points represent the small, medium, and large capital endowments that contribute to $\supp(\Theta)$.}
    \label{fig:k-cluster}
\end{figure}
\end{document}